\documentclass[11pt,letterpaper]{article}

\usepackage[margin=1in]{geometry}

\usepackage[T1]{fontenc}
\usepackage[utf8]{inputenc}
\usepackage{microtype}

\usepackage{amsmath}
\usepackage{amsthm}
\usepackage{amssymb}

\usepackage[dvipsnames]{xcolor}

\usepackage{algorithm2e}

\usepackage{graphicx}
\usepackage{tikz}

\usepackage{thmtools}
\usepackage{thm-restate}

\usepackage{url}
\usepackage{comment}
\usepackage{cancel}
\usepackage{multicol}

\usepackage[pdfusetitle]{hyperref}

\definecolor{revisioncolor}{RGB}{200,0,0}

\newcommand{\rev}[1]{#1}
\newenvironment{reve}{}{}

\newcommand{\A}{\mathcal{A}}
\newcommand{\D}{\mathcal{D}}
\newcommand{\E}{\mathcal{E}}
\newcommand{\F}{\mathcal{F}}
\newcommand{\I}{\mathcal{I}}

\newcommand{\2}{\mathbf{2}}
\newcommand{\twobar}{\overline{2}}
\newcommand{\threebar}{\overline{\mathbf{3}}}

\newcommand{\B}{\mathbb{B}}
\newcommand{\N}{\mathbb{N}}
\newcommand{\R}{\mathbb{R}}
\newcommand{\Rbar}{\overline{\mathbb{R}}}

\newcommand{\qedclaim}{\hfill $\diamond$ \medskip}

\theoremstyle{definition}
\newtheorem{definition}{Definition}

\newtheorem{example}{Example}

\theoremstyle{plain}
\newtheorem{property}{Property}
\newtheorem{lemma}{Lemma}
\newtheorem{theorem}{Theorem}
\newtheorem{corollary}{Corollary}
\newtheorem{conjecture}{Conjecture}
\newtheorem{claim}{Claim}
\newtheorem{remark}{Remark}
\newtheorem{fact}{Fact}

\DeclareMathOperator*{\argmin}{argmin}

\newenvironment{proofclaim}{\noindent{\em Proof of the claim.}}{\qedclaim}

\title{New Perspectives on Semiring Applications to Dynamic Programming}

\author{
  Ambroise Baril\thanks{Université de Lorraine, CNRS, LORIA, F-54000 Nancy, France} \and
  Miguel Couceiro\thanks{Université de Lorraine, CNRS, LORIA, F-54000 Nancy, France; 
  INESC-ID, Instituto Superior Técnico, Universidade de Lisboa, Portugal} \and
  Victor Lagerkvist\thanks{Linköping University, Sweden}
}

\date{}

\begin{document}

\maketitle

\begin{abstract}
Semiring algebras have been shown to provide a suitable language to formalize many noteworthy combinatorial problems. For instance, the \textsc{Shortest-Path} problem can be seen as a special case of the \textsc{Algebraic-Path} problem when applied to the tropical semiring. The application of semirings typically makes it possible to solve extended problems without increasing the computational complexity.
In this article we further exploit the idea of using semiring algebras to address and tackle several extensions of classical computational problems by dynamic programming. 

We consider a general approach which allows us to define a semiring extension of {\em any} problem with a reasonable notion of a certificate (e.g., an {\sf NP} problem).
This allows us to consider cost variants of these combinatorial problems, as well as their counting extensions where the goal is to determine how many solutions a given problem admits. The approach makes no particular assumptions (such as idempotence) on the semiring structure. We also propose a new associative algebraic operation on semirings, called $\Delta$-product, which enables our dynamic programming algorithms to count the number of solutions of minimal costs. We illustrate the advantages of our framework on two well-known but computationally very different {\sf NP}-hard problems, namely, \textsc{Connected-Dominating-Set} problems and finite-domain \textsc{Constraint Satisfaction Problems} ({\sc Csp}s). In particular, we prove fixed parameter tractability ({\sf FPT}) with respect to clique-width and tree-width of the input. \rev{This also allows us to count solutions of minimal cost, which is an overlooked problem in the literature.}
\end{abstract}


\section{Introduction}
\label{sec:intro}
In this article we investigate {\em semiring extensions} of computational problems. First, we take an algebraic viewpoint and define a novel semiring operation that increases the range of semiring algorithms. Second, we apply our semiring framework algorithmically and construct fixed-parameter tractable algorithms for semiring extensions of {\sc Connected-Dominating-Set} and the {\sc Constraint Satisfaction Problem} ({\sc Csp}). 

\subsection{Background}

A {\em semiring} is an algebra $\A=(A,+,\times,0_A,1_A)$ where $(A,+,0_A)$ is a commutative monoid, $(A,\times,1_A)$ is a monoid, $\times$ distributes over $+$, and where $0_A$ is absorbing with respect to $\times$. Well-known examples include the Boolean semiring $(\{\bot, \top\},\vee,\land,\bot,\top)$, the integer semiring $(\N,+,\times,0,1)$, and the tropical semiring $(\R \cup \{\infty\},\min,+,\infty,0)$.
Semirings have proven to be versatile in extending various types of algorithms. A classic example is the Floyd-Warshall algorithm used for computing the minimal path between any two nodes in a weighted graph~\cite{floyd1968,warshall1962}.
Hence, in this problem we are interested in computing the minimum value of all sums of weights corresponding to paths in the given graph, and from the semiring perspective we are then simply evaluating an expression with respect to the tropical semiring.
Extending the basic Floyd-Warshall algorithm to an arbitrary semiring is relatively easy and increases the expressive strength tremendously (see, e.g., Lehmann~\cite{lehmann1977} for a general treatment). 
For example, if we use the Boolean semiring instead of the tropical semiring then the Floyd-Warshall algorithm can be used to compute the transitive closure of a given graph. Generally, when a problem is generalized by a semiring then the resulting problem is called a {\em semiring extension}. This problem has attracted significant attention and is in its most general form known as the {\em algebraic path }problem, whose roots can be traced back to Kleene's algorithm for converting regular expressions to finite automata~\cite{kleene1956}. For a comprehensive discussion of the literature, see e.g.\ Mohri~\cite{mohri2002}. However, semiring extensions are by no means limited to graph problems. Notable formalisms include the {\em semiring constraint satisfaction }problem ({\sc Scsp}) by Bistarelli et al.~\cite{bistarelli1997semiring} which greatly generalizes various forms of fuzzy reasoning where constraints are allowed to be soft (an arbitrary semiring value) rather than crisp (true or false). This problem is in turn subsumed by the {\em sum-of-product CSP} problem which is currently of central importance in artificial intelligence~\cite{bacchus2009solving,eiter2023semiring}.

In this article we are interested in semiring extensions and take a very general approach and define a semiring extension over {\em any} computational problem in {\sf NP}.
Thus, for any problem $\mathbf{\Pi}$ in {\sf NP} and any reasonable notion of a solution space (e.g., the set of certificates) and any semiring we consider the problem of computing a {\em join/union expression} of the set of solutions of a given instance of $\mathbf{\Pi}$. We refer the reader to Section \ref{sec:semirings} for a definition of a join/union expression but remark that the basic idea is to compute an  expression which in combination with a semiring can be used to compute the desired property. This strictly generalizes the problem of {\em counting} the number of solutions (typically denoted $\#\mathbf{\Pi}$) and the problem of finding a solution of minimal cost ({\sc Cost}-$\mathbf{\Pi}$). 
Let us consider two examples. First, the well-known problem of determining whether an input graph $G$ admits a homomorphism to the template graph $H$ is known as the {\em $H$-{\sc Coloring}} problem. Here, a natural certificate of a yes-instance is simply the homomorphism itself. Then $\#H$-{\sc Coloring} is the problem of counting the number of homomorphisms, while in the {\sc Cost}-$H$-{\sc Coloring} each assignment of a variable to a value is associated with a weight, and the goal is to find the homomorphism which minimizes the total sum of weights. Notice that the decision problem $H$-{\sc Coloring} asks for an element of the Boolean semiring $\2=\{\bot,\top\}$ ($\top$ if a homomorphism exists and $\bot$ otherwise), and that the counting extension $\#H$-{\sc Coloring} expects an answer in the natural semiring $(\N,+,\times,0,1)$. The cost version {\sc Cost}-$H$-{\sc Coloring} can be modeled using the tropical semiring $\R_{\min}=(\Rbar,\min,+,\infty,0)$. Moreover, the well studied \textsc{List-$k$-Coloring} problem, where each vertex  is restricted to a list of ``allowed colors'' \cite{jensen2011graph,thomassen19953}, can also be seen as an extension of \textsc{$k$-Coloring}, even if it also lies in the boolean semiring. The semiring extension of $H$-{\sc Coloring} then generalizes  all of these problems under a single umbrella, potentially allowing one to benefit from a single algorithm. For a second example, the {\sc Dominating-Set} problem is the problem of deciding whether a given graph contains a subset of $k$ vertices such that every other vertex has at least one neighboor in the subset. In the {\sc Connected-Dominating-Set} problem, we, as the name suggests, additionally require that the subset of vertices is connected. This problem is important in, for example, network communication where the connected dominating set is viewed as network backbone that the other nodes can communicate via (for more applications see e.g.\ the book~\cite{10.5555/2412083}). A natural certificate is then simply the subset of vertices, and the semiring extension of ({\sc Connected-}){\sc Dominating-Set} then makes it possible to count the number of (connected) dominating sets as well as finding a (connected) dominating set of minimal cost.

Solving semiring extensions is thus highly desirable from a practical perspective since one effectively gets many algorithms for the price of one and can reuse the algorithm in different applications simply by choosing new semirings. %
From a complexity perspective many complexity classes have been introduced in order to study semiring extensions of {\sf NP} problems (corresponding to the Boolean semiring $\B$ since they are decision problems). For example, the class {\sf \#P} \cite{DBLP:journals/tcs/Valiant79} for the counting extension (lying in the natural semiring $\N$), the class {\sf MOD$_p$} \cite{hertrampf1990relations} for the counting problems modulo an integer $p\ge 2$ (in the semiring $\mathbb{Z}_p$), and the class {\sf OptP} for optimization problem \cite{krentel1986complexity} (often computed in the tropical semiring $\R_{\min}$). More generally, Eiter \&  Kiesel~\cite{eiter2023semiring} define a semiring extension \textsf{NP}$(\mathcal{R})$ of \textsf{NP} for any commutative semiring $\mathcal{R}$ and a notion of a semiring Turing machine. This makes it possible to prove a unifying meta-theorem which identifies the corresponding variant of {\sc Sat} as a complete problem for each of the aforementioned classes. 
Despite this, there has been comparably little research done on {\em combining} semiring extensions.
For example, consider problems consisting in counting the number of solutions of minimal cost (described by the class {\sf \#$\cdot$OptP} \cite{hermann2009complexity}). Thus, intuitively, one wishes to combine the natural semiring $\N$ (counting) and the tropical semiring $\R_{\min}$ (weighted) into a new semiring with the hope of efficiently solving the associated problem of counting solutions of minimal cost. This raises an intriguing question: can counting solutions of minimal cost (or similar variations) be addressed using the same algorithmic techniques as counting problems or minimal cost problems alone? And can we accomplish this while staying in the familiar semiring framework? Or do these combinations yield fundamentally different problems?

\subsection{Our contribution}

We begin (in Section~\ref{sec:semirings}) by introducing our semiring framework. As remarked, we adopt a very general approach and define a semiring extension for {\em every} problem in {\sf NP}.  This drastically increases the reach of our semiring framework and makes it possible to study both the aforementioned finite domain {\sc Csp} problem and the {(\sc Connected-){\sc Dominating-Set}} even though these problems have a very different flair. Furthermore, to obtain as general results as possible we avoid the (otherwise commonly assumed) assumption that  the semiring in question is idempotent. We remark that idempotence is typically assumed in the \textsc{algebraic path} problem as well as in the {\sc Scsp} formalism by Bistarelli~\cite{bistarelli1997semiring} via so-called {\em $c$-semirings}. Semiring approaches that requires the milder, but arguably less natural, assumptions such as the {\em $k$-closure}, was considered by Mohri~\cite{mohri2002} in the context of the \textsc{Algebraic-Path} problem. Idempotency is typically assumed since it makes it easier (in an algorithmic context) to combine (via union) two partial (but not necessarily disjoint) sets of solutions into a larger one. 
For instance, if we want to solve a problem of the form \textsc{Cost-}$\mathbf{\Pi}$, we can find the minimal cost of the union of two sets (corresponding to partial solutions of the problem) by computing the minimal costs in both sets, and keep the minimum. This approach is correct because of the idempotence of the underlying tropical semiring $\R_{\min}$. However, this approach would fail to solve the \#$\mathbf{\Pi}$ problem, as the cardinality of the union can only be expressed as the sum of the cardinality of both sets if the sets are disjoint. The approach fails because of the non-idempotence of the natural semiring $\N$. In our approach one is instead given a computational problem with a reasonable notion of a solution and is tasked with constructing a join/union expression representing the set of solutions of a problem instance. The interest is that individual components in the join/union expression can be interpreted as dynamic programming operations, e.g., combining two solutions to subproblems to a solution to a larger problem. Given this link one might suspect that polynomial sized join/union expressions cannot always be computed efficiently, and we indeed provide examples (in Section~\ref{sec:semiring_expressions}) where this is not possible (unless {\sf P = NP}). Since this is a purely combinatorial statement it is an interesting open question whether this could be proved unconditionally.

In Section~\ref{sec:Delta-product} we then completely resolve the problem of combining different semiring extensions. Our main technical tool is the {\em $\Delta$-product} which, given two semirings $\mathcal{D}$ and $\mathcal{A}$ effectively produces a new semiring $\mathcal{D} \Delta \mathcal{A}$ which combines the relevant properties. More specifically, we assume that $\mathcal{A}$ is a commutative semiring and $\mathcal{D}$ a totally ordered, idempotent, commutative dioid, which allows it to represent a set of weights.
The algorithmic applications of this new semiring correspond exactly to applying the semiring $\A$ to optimal solutions with respect to the weights in $\D$. In particular, by choosing $\A=\N$ and $\D=\R_{\min}$, it enables us to count the solutions of minimal cost. Our construction therefore makes the problems in the class {\sf $\#\cdot$OptP} \cite{hermann2009complexity} fall into the category of semiring extensions. In particular, the results of completeness of the problems {\sc $\#$Min-Card-Sat} and {\sc $\#$Min-Weight-Sat} \cite{hermann2009complexity} now become particular improved cases of the meta-theorem over the completeness of any semiring variant of {\sc Sat} \cite{eiter2023semiring}. Moreover, since the $\Delta$-products implies that {\sf $\#\cdot$OptP} is subsumed by the semiring formalism, any general positive result applicable to all semirings are now applicable specifically to these problems.

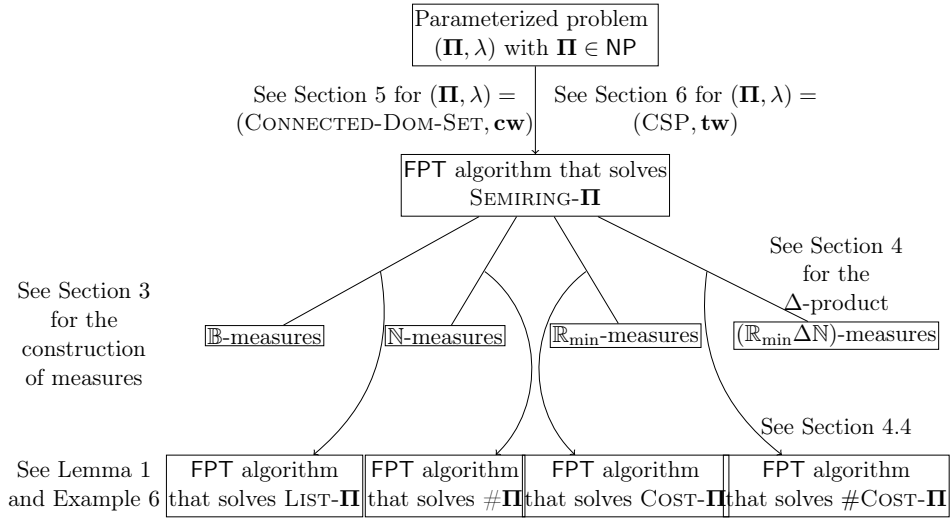
\begin{figure}[!ht]
\begin{center}
\scalebox{0.8}{
\begin{tikzpicture}

\tikzstyle{vertex}=[circle, draw, inner sep=1pt, minimum width=6pt]
\tikzstyle{sq}=[rectangle, draw, inner sep=1pt, minimum width=6pt]

\node[sq] (NPproblem) at (0,5) {$\begin{matrix}\text{Parameterized problem } \\ (\mathbf{\Pi},\lambda)\text{ with }\mathbf{\Pi}\in{\sf NP}\end{matrix}$};

\node[sq] (SemiringNP) at (0,2.5) {$\begin{matrix}{\sf FPT}\text{ algorithm that solves} \\ \textsc{Semiring-}\mathbf{\Pi}\end{matrix}$}
    edge[<-] (NPproblem);

\node () at (-2.5,3.75) {$\begin{matrix} \text{See Section }\ref{sec:dom-set}\text{ for } (\mathbf{\Pi},\lambda)= \\ (\textsc{Connected-Dom-Set},\mathbf{cw}) \end{matrix}$};

\node () at (2.5,3.75) {$\begin{matrix} \text{See Section }\ref{sec:Csp}\text{ for } (\mathbf{\Pi},\lambda)= \\ (\textsc{CSP},\mathbf{tw}) \end{matrix}$};

\node () at (-7.5,0) {$\begin{matrix} \text{See Section }\ref{sec:semirings} \\ \text{for the} \\\text{construction} \\ \text{of measures}\end{matrix}$};

\node[sq] (B) at (-4.5,0) {$\B$-measures}
    edge (SemiringNP);

\node[sq] (N) at (-1.5,0) {$\N$-measures}
    edge (SemiringNP);

\node[sq] (Rmin) at (1.5,0) {$\R_{\min}$-measures}
    edge (SemiringNP);

\node[sq] (Delta) at (5,0) {$(\R_{\min}\Delta\N)$-measures}
    edge (SemiringNP);

\node[] () at (5,1) {$\begin{matrix} \text{See Section }\ref{sec:Delta-product} \\ \text{for the} \\ \Delta\text{-product}\end{matrix}$};

\node (midB) at (-2.58,1.2) {};
\node (midN) at (-0.96,1.2) {};
\node (midRmin) at (0.96,1.2) {};
\node (midDelta) at (2.85,1.2) {};

\node[sq] (Decision) at (-4.5,-2.5) {$\begin{matrix} \textsf{FPT}\text{ algorithm} \\ \text{that solves }\textsc{List-}\mathbf{\Pi}  \end{matrix}$}
    edge[<-,bend right] (midB);

\node[sq] (Counting) at (-1.5,-2.5) {$\begin{matrix} \textsf{FPT}\text{ algorithm} \\ \text{that solves \#}\mathbf{\Pi}  \end{matrix}$}
    edge[<-,bend right=50] (midN);

\node[sq] (Cost) at (1.5,-2.5) {$\begin{matrix} \textsf{FPT}\text{ algorithm} \\ \text{that solves }\textsc{Cost-}\mathbf{\Pi}  \end{matrix}$}
    edge[<-,bend left=50] (midRmin);

\node[sq] (CountCost) at (5,-2.5) {$\begin{matrix} \textsf{FPT}\text{ algorithm} \\ \text{that solves }\#\textsc{Cost-}\mathbf{\Pi}  \end{matrix}$}
    edge[<-,bend left] (midDelta);

\node () at (-7.5,-2.5) {$\begin{matrix}\text{See Lemma } \ref{lem:expression_computation} \\ \text{and Example }\ref{example:measures} \end{matrix}$};

\node at (5,-1.5) {See Section \ref{sec:count_cost_application}};

\end{tikzpicture}
}
\end{center}
\caption{General overview of our contributions.}
\label{fig:plan_of_the_paper}
\end{figure}

We showcase our framework on two distinct problems. First, in Section~\ref{sec:dom-set} we target the {\sc Connected-dominating-set} problem under the lens of the {\em clique-width} of the input graph. Clique-width is, together with {\em tree-width}, the most well-known graph parameter, and measures how many basic operations that are required to construct the graph. The class of graphs (with labeled vertices) with clique-width at most $k\geq 1$ is defined as the smallest class of graphs that contains the one vertex graphs $\bullet_i$ with 1 vertex labeled by $i\in [k]$, and that is stable by the following operations for $(i,j)\in [k]^2$ with $i\neq j$: $(i)$ disjoint union of graphs, $(ii)$ relabeling every vertex of label $i$ to label $j$, and $(iii)$ constructing edges between every vertex labeled by $i$ and every vertex labeled by $j$. Note that the class of cographs (which contains cliques) is exactly that of graphs with clique-width at most $2$. For graph parameters such as clique-width the task is usually to prove {\em fixed parameter tractability} ({\sf FPT}), {\it  i.e.,} an algorithm with a running time bounded by $f(k) \cdot ||I||^{O(1)}$ where (1) $f \colon \mathbb{N} \to \mathbb{N}$ is a computable function, (2) $k$ is the parameter associated with the instance $I$, and (3) $||I||$ is the number of bits required to represent $I$. For classes of instances where \rev{the combination of the parameter $k$ and the function $f$} are believed to be reasonably small an {\sf FPT} algorithm is thus essentially as good as a polynomial time algorithm. 

It is known that the {\sc Connected-dominating-set} problem can be solved by a {\sf FPT} (with respect to the clique-width of the input graph) one-sided error Monte-Carlo algorithm~\cite{hegerfeld2023tight}, as it can be solved in $O^*(5^{\mathbf{cw}})$ time\footnote{\rev{The notation $O^*(\cdot)$ suppresses polynomial factors.}}, knowing that an exact algorithm solving this problem can not run in $O^*((5-\varepsilon)^{\mathbf{cw}})$ time for any $\varepsilon>0$, unless the \textsf{SETH} fails. The {\sf SETH} \cite{impagliazzo2009} is a well-known conjecture within complexity theory which states that for all $\varepsilon>0$, there exists $k\ge 3$ such that the {\sc $k$-Sat} problem can not be solved in time $O^*((2-\varepsilon)^n)$ (where $n$ is the number of variables).

Second, in Section~\ref{sec:Csp} we turn our attention to the finite-domain {\em constraint satisfaction problem}, {\it  i.e.,} the problem of verifying whether a set of finite-domain constraints admits at least one satisfying assignment. For this problem we consider the tree-width of the {\em Gaifman graph} of an input instance and produce an {\sf FPT}-algorithm for {\sc Semiring-Csp}. The algorithm is based on the algorithm by Ganian et al.~\cite{ganian2016combining} for \#{\sc Csp} but is strictly more general since it can also solve the associated cost problem as well as counting solutions of minimal cost (via the $\Delta$-product). We additionally remark that for every finite domain $D$ our $O^*(|D|^k)$ (where $k$ is the tree-width of the Gaifman graph) algorithm is optimal (under the \textsf{SETH}) since it as a special case includes the $H$-{\sc coloring} problem which is known to not admit an improved algorithm under the {\sf SETH} via Okrasa \& Rz\c ażewski~\cite{okrasa2020finegrained}. In Section~\ref{sec:sum_product} we also notice that our approach can be extended to solve the more general \textsc{Sum-Product-Csp} problem.

\subsection{Organization of the article} \label{sec:org}

In Figure~\ref{fig:plan_of_the_paper} we summarize the role of each contribution to our novel semiring formalism. The upper half of the Figure \ref{fig:plan_of_the_paper} (that refers to Sections \ref{sec:dom-set} and Section \ref{sec:Csp}) deal with algorithmic results: we explain how to solve the  \textsc{Semiring} extension of \textsc{Connected-Dominating-Set} and \textsc{Csp} in \textsf{FPT} time (with respect to respectively the clique-width and the tree-width). Whereas, the lower half of Figure \ref{fig:plan_of_the_paper} refers to the algebraic results of Section \ref{sec:semirings} and Section \ref{sec:Delta-product}: we construct algebraic tools that enables us to construct an algorithm that solves the counting, cost and list extensions of any \textsf{NP} problem, provided that an algorithm for the \textsc{Semiring} version is given.

\section{Preliminaries}

\subsection{Sets and functions}

For any integer $n$, we let $[n]:=\{1,\dots,n\}$. Given a set $V$, we denote by $2^V$ the set of subsets of $V$.

For two finite sets $S$ and $T$, let $T^S$ be the set of functions from $S$ to $T$. Note that its cardinal is $|T^S|=|T|^{|S|}$ and that, in particular, the cardinal of $T^{\emptyset}$ is $|T|^0=1$ since it contains the only function with codomain $T$ and empty domain. \rev{Furthermore, we adopt the following notation:
\begin{itemize}
    \item For all $s\in S$ and $t\in T$, we denote by $(s\mapsto t)\colon \{s\}  \to  T$ the function given by $(s\mapsto t)(s)=t$.
    \item Given $S'\subseteq S$ and $f\in T^S$, {\em the restriction of $f$ to $S'$} is the function $f|_{S'}\colon S' \to T$ defined by $f|_{S'}(s)= f(s)$, for every $s\in S'$.
    \item For a set $V$ and $S\subseteq V$, we denote the {\it indicator function} by $\mathbf{1}_S^V\colon V  \to  \{0,1\}$, which is  given by $\mathbf{1}_S^V(u)= 
     \left\{ \begin{array}{cc}
     1 & \text{ if }u\in S  \\
     0 & \text{ if }u\notin S
\end{array}. \right. $
\end{itemize} 
}
\noindent Notice that the mapping that sends any $S\in 2^V$ to $\mathbf{1}_S^V$ is an isomorphism between $2^V$ and $\{0,1\}^V$.

\subsection{Permutations}

Let $n\ge 0$. We write $\mathfrak{S}_n$ for the set of {\em permutations} of $[n]$.
A permutation $\sigma\in\mathfrak{S}_n$ is denoted by $\sigma=\begin{pmatrix} 1 & \dots & n \\ \sigma(1) & \dots & \sigma(n) \end{pmatrix}$.

For $k\in \{2,\dots,n\}$ and $k$ pairwise distinct elements $a_1,a_2\dots,a_k$ of $[n]$, $(a_1 a_2\dots a_k)$ denotes the element of $\mathfrak{S}_n$ that maps $a_1$ to $a_2$, $a_2$ to $a_3$, $\dots$, and $a_k$ to $a_1$ (the other elements of $[n]$ are sent to themselves). A permutation $\sigma\in\mathfrak{S}_n$ is said to be a {\em $k$-cycle} if it is of the form $(a_1 a_2\dots a_k)$. We write $\mathfrak{C}_n$ for the set of $n$-cycles of $\mathfrak{S}_n$.

\subsection{Graph and trees}

By a {\em graph} we mean a tuple $G=(V_G,E_G)$ where $V_G$ is a finite set called the set of {\em vertices} of $G$, and $E_G$ is \rev{a set of pairs of vertices of $G$} called the set of {\em edges} of $G$. The set of vertices and edges of a graph $G$ will always be denoted by $V_G$ and $E_G$, respectively. Given $S\subseteq V_G$, we let $G[S]:=(S,E_G\cap S^2)$ be the graph induced by $S$ on $G$.

\rev{Two vertices $u$ and $v$ of a graph $G$ are said to be \textit{neighbors} (in the graph $G$) if $\{u,v\}\in E_G$. A \textit{path} in a graph $G$ is a finite sequence $u_0u_1\dots u_p$ of vertices, where $p$ is a non negative integer, and such that for all $i\in \{0,\dots,p-1\}$, $u_i$ and $u_{i+1}$ are neighbors in $G$. A \textit{simple path} in a graph $G$ is a path without vertex repetition. A \textit{cycle} in a graph $G$ is a path in $G$ of the form $u_0u_1\dots u_p$, with $p\ge 3$, $u_0=u_p$, and with $u_0u_1\dots u_{p-1}$ being a simple path.}

A {\em tree} is a connected graph $T=(N_T,E_T)$ with no cycle, $N_T$ is said to be the set of {\em nodes} of $T$. A {\em rooted tree} is a pair $(T,r)$ where $T=(N_T,E_T)$ is a tree and $r\in N_T$ is said to be the {\em root} of $(T,r)$. Abusing notation, we will say that $r$ is the root of $T$. The {\em ancestors} of a node $N\in N_T$ in a rooted tree $(T,r)$ are the nodes on the unique simple path from $r$ to $N$. A {\em descendant} of a node $N\in N_T$ is a node $N'\in N_T$ such that $N$ is an ancestor of $N'$. The {\em children} of a node $N\in N_T$ are the descendants of $N$ that are also neighbors of $N$. A {\em leaf} of a rooted tree is a node that has no children.

\subsection{The dominating set problem}

The {\sc Dominating-Set} problem is a widely known and important {\sf NP}-complete problem over graphs, that is not subsumed by the {\sc Csp} formalism.%

Given a graph $G$ and $u\in V_G$, the {\em closed neighborhood} of $u$ in $G$ is the set $N_G[u]=\{ v\in V_G \mid \{u,v\} \in E_G \}\cup \{u\}$. \rev{For a set $U \subseteq V_G$ we similarly define $N_G[U] = \bigcup_{u \in U}N_G[u]$.}

\begin{definition}

A {\em dominating set} of a graph $G$ is a subset $S\subseteq V_G$ such that $N_G[S] = V_G$.

\end{definition}

The goal of the {\sc Dominating-Set} problem is then to determine the smallest size of a dominating set of an input graph. We present here the optimization problem associated with the {\sf NP}-problem.

\hfill

\underline{\sc Dominating-Set:}

\textbf{Input:} A graph $G$.

\textbf{Output:} The smallest $k\ge 0$ such that there exists $S\subseteq V_G$ where $|S|\le k$ and $N_G[S] = V_G$.

\hfill

Recall that a join/union expression can only represent a set of functions. It is therefore natural to represent a subset $S\subseteq V_G$ by its indicator function $\mathbf{1}_S^{V_G}$.

\hfill

In this section, we focus on the significantly more difficult variant where in addition we require that the dominating set is connected. This problem is important in, for example, network communication~\cite{10.5555/2412083} and is essentially an entire research field in itself. 

\hfill

\underline{\sc Connected-Dominating-Set:}

\textbf{Input:} A graph $G$.

\textbf{Output:} The smallest $k\ge 0$ such that there exists $S\subseteq V_G$ where $|S|\le k$, $N_G[S] = V_G$ and $G[S]$ is connected.

\hfill

Such a set $S$ (satisfying $N_G[S] = V_G$ and that $G[S]$ is connected) is called a {\em connected dominating-set}.

\subsection{The constraint satisfaction problem}

A finite set $\Gamma$ of relations over a finite domain $D$ is called a {\em constraint language}.
We define the {\em constraint satisfaction problem} (\textsc{Csp}) over a constraint language $\Gamma$ as follows.

\hfill

\underline{{\sc Csp}$(\Gamma)$:}

\textbf{Input:} $I = (V,C)$ where $V$ is a set of variables and $C$ a set of \rev{{\em constraints} of the form $((x_1, \ldots, x_m), R)$} where $R\in\Gamma$ is an $m$-ary relation and $x_1, \ldots, x_m \in V$.

\textbf{Question:} Does there exist a function $f:V \mapsto D$ which satisfies all constraints, {\it  i.e.,} $(f(x_1), \ldots, f(x_m)) \in R$ for every \rev{$((x_1, \ldots, x_m), R) \in C$}?

\hfill

\rev{We often write $R(x_1, \ldots, x_m)$ rather than $((x_1, \ldots, x_m), R)$.}
The function $f \colon V \to D$ is sometimes called a {\em model}, or a {\em solution}. A function $f' \colon V' \to D$ with $V'\subseteq V$ which satisfies every constraint $R(x_1, \ldots, x_m) \in C$ where $x_1, \ldots, x_m \in V'$ is called a {\em partial solution}.

\begin{remark}

The finite domain $D$ of a \textsc{Csp} instance is the {\em codomain} of the valuations $f:V\mapsto D$ (their domain is the set of variables $V$).

\end{remark}

Let us also remark that the well-known problem of deciding whether an input graph $G$ admits a homomorphism to a template graph $H$ (the {\em $H$-{\sc Coloring} }problem) can be seen as a particular case of {\sc Csp}, by simply letting $\Gamma=\{E_H\}$. It is also easy to see that the \textsc{Boolean satisfiability} problem (\textsc{SAT}) for instances in conjunctive normal form can be seen as a special case of \textsc{Csp}s with Boolean domain.

\hfill

\underline{$H${\sc -Coloring}:}

\textbf{Input:} A graph $G$.

\textbf{Question:} Does there exist a function $f:V_G \to V_H$ which satisfies for all $\{u,v\}\in E_G$ that $\{f(u),f(v)\}\in E_H$?

\hfill

For example, if we let $K_k$ be the $k$-vertex clique, $K_k${\sc -{\sc Coloring}} is an alternative formulation of the well-known computational problem of deciding whether a graph with $n$ vertices can be colored with $k$ distinct colors ($k${\sc -{\sc Coloring}}).

The {\em Gaifman graph} \cite{rossi2006handbook} of an instance $\mathcal{I}=(V,C)$ of {\sc Csp} is the graph $G=(V,E)$, with $\{x,y\}\in E$ if there exists a constraint $c\in C$ that involves both $x$ and $y$ (for all pair $\{x,y\}\subseteq V$). Note that in the case of {\sc $H$-Coloring}, the Gaifman graph of any input graph $G$ is $G$ itself.

\begin{reve}
\subsection{Tree- and Cliquewidth}\label{sec:width}
We now introduce two important structural graph parameters. 

\begin{definition}
A {\em tree decomposition} \cite{bodlaender1991better} of a graph $G=(V_G,E_G)$ is a tuple $(T,bag)$ where $T=(N_T,E_T)$ is a tree and $bag \colon N_T\mapsto 2^{V_G}$ is a function that satisfies the following properties.

\begin{enumerate}
    \item For all $v\in V_G, \exists N\in N_T$ with $v\in bag(N)$.
    \item For all $\{u,v\}\in E_G$, $\exists N\in N_T$ such that $\{u,v\}\subseteq bag(N)$.
    \item For all $(N_1,N_2,N_3)\in (N_T)^3$, if $N_2$ is on the (unique) path between $N_1$ and $N_3$ in $T$, then $bag(N_1)\cap bag(N_3)\subseteq bag(N_2)$.
\end{enumerate}

\end{definition}

The {\em treewidth} of a tree decomposition $((N_T,E_T),bag)$ defined as the integer $\max\limits_{N\in N_T} |bag(N)|-1$. The {\em treewidth} of a graph $G$ denoted by $\mathbf{tw}(G)$ is then the minimum of the treewidth of its tree decompositions. In algorithmic applications it is often desirable to make additional assumptions on the tree decomposition.

Additionally, a {\em nice tree decomposition} \cite{bodlaender1991better} of a graph $G$ is a rooted tree decomposition of $G$ where every node $N$ is of one of the four following types.

\begin{enumerate}
    \item \textbf{Leaf:} $N$ is a leaf of $T$ and its bag is empty.
    \item \textbf{Forget:} $N$ has one child and the bag of $N$ is of the form $bag(N)=bag(child(N))\setminus\{v\}$ with $v\in bag(child(N))$.
    \item \textbf{Introduce:} $N$ has one child and the bag of $N$ is of the form $bag(N)=bag(child(N))\uplus\{v\}$ with $v\in V_G$ and $v\notin bag(child(N))$.
    \item \textbf{Join:} $N$ has two children $N_1$ and $N_2$ with $bag(N)=bag(N_1)=bag(N_2)$.
\end{enumerate}

Here, $bag(N)$ designates the bag associated to a node $N$ of the rooted tree decomposition. We will also use the notation $$descbag(N):=\bigcup\limits_{N' \text{descendant of }N} bag(N').$$ 
We remark that any tree decomposition of a graph can easily be transformed into a nice tree decomposition with the same treewidth and in linear time (in the size of the tree decomposition)~\cite{bodlaender1991better}.

Next, we introduce {\em clique-width}. The general idea behind clique-width is that graphs with low clique-width can be decomposed into simpler structures, making it amenable to efficient algorithmic solutions for problems that are otherwise intractable~\cite{courcelle2000linear,courcelle2000upper}. 

For $k\geq 1$, let $[k]=\{1,\dots, k\}$. A {\em $k$-labeled graph} $G$ is a tuple $(V_G,E_G,l_G)$, where $(V_G,E_G)$ is a graph and $l_G: V_G \to [k]$. For $i\in [k]$ and a $k$-labeled graph $G$, denote by $V_G^i= l_G^{-1}(\{i\})$ the set of vertices of $G$ of label $i$.
A {\em $k$-expression} $\varphi$ of a $k$-labeled graph $G$, denoted $[\varphi]=G$, is an expression defined inductively~\cite{courcelle2000upper} using:
\begin{enumerate}
\item \textbf{Single vertex:} $\bullet_i$ with $i\in [k]$: $[\bullet_i]$ is a $k$-labeled graph with one vertex labeled by $i$ (we sometimes write $\bullet_i(u)$ to state that the vertex is named $u$),
\item \textbf{Disjoint Union:} $\varphi_1\oplus \varphi_2$: $[\varphi_1\oplus\varphi_2]$ is the disjoint union of the graphs $[\varphi_1]$ and $[\varphi_2]$.
\item \textbf{Relabeling:} $\rho_{i\rightarrow j}(\varphi)$ with $(i,j)\in [k]^2$ and $i\neq j$: $[\rho_{i\rightarrow j}(\varphi)]$ is the same graph as $[\varphi]$, but where all vertices of $G$ with label $i$ now have label $j$, 
\item \textbf{Edge Creation:} $\eta_{i,j}(\varphi)$ with $(i,j)\in [k]^2$ and $i\neq j$: $[\eta_{i,j}(\varphi)]$ is the same graph as $[\varphi]$, but where all pairs of the form $\{u,v\}$ with $\{l_G(u),l_G(v)\} = \{i,j\}$ is now an edge, and
\end{enumerate}

A graph $G$ has a $k$-expression $\varphi$ if there exists $l:V_G\mapsto [k]$ such that $[\varphi]=(V_G,E_G,l)$. The {\em clique-width} of a graph $G$ (denoted by $\mathbf{cw}(G)$) is the minimum $k\geq 1$ such that $G$ has a $k$-expression.

Moreover, if $S\subseteq V_G$ is a subset of vertices of a $k$-labeled graph $G$, we denote by $l_G(S):=\{ l_G(u)\mid u\in S \}\in 2^{[k]}$ the set of labels of vertices in $S$.
\end{reve}

\subsection{Computational complexity}

Given two infinite sequences $(u_n)_{n\in\mathbb{N}}$ and $(v_n)_{n\in\mathbb{N}}$ in $\mathbb{R}^{\mathbb{N}}$, we write $u_n=O^*(v_n)$ if there exists a polynomial $p$ such that $u_n=O(v_n\times p(n))$. \rev{This notation is convienient since we when working with NP-hard problems typically are much more interested in the superpolynomial factor.}

\hfill

The class {\sf NP} is the class of computational decision problems $\mathbf{\Pi}$ for which there exists a deterministic Turing machine $M$ running in polynomial time and a polynomial $p$ such that for all $n\ge 0$ and $x\in \{0,1\}^n$:
$$ x\in \mathbf{\Pi} \iff \exists y\in\{0,1\}^{p(n)},\ M(x,y) \text{ accepts},$$
in which case such an $y\in \{0,1\}^{p(n)}$ is called a {\em certificate} for the instance $x$. Here, we stick to the fixed alphabet $\{0,1\}$ to simplify the presentation, but, naturally, any other alphabet involving at least two symbols could be used instead.

\hfill

Consider a computational problem $\mathbf{\Pi}$ in {\sf NP}. A certificate $y=y_1\dots y_{p(n)}\in \{0,1\}^{p(n)}$ of an instance of size $n\ge 0$ of $\mathbf{\Pi}$ (where $p$ is polynomial) can be represented bijectively by \rev{the function $\Tilde{y}: [p(n)]  \to  \{0,1\}$ given by $\Tilde{y}(i)= y_i$. }However, most of the time it is more natural to give a higher-level representations of ``solutions'' (i.e., certificates) as functions that have other domain and codomain than $[p(n)]$ and $\{0,1\}$. For instance, a ``solution'' of $k$-{\sc Coloring} (with $k\ge 1$) on an input graph $G$ will be seen as a function from $V_G$ to $[k]$. Generally, we assume that the set of ``solutions'' of our problem $\mathbf{\Pi}$ on a given instance is a subset of $T^S$, for some sets finite sets $S$ and $T$ having polynomial size in the size of the input (which can safely be assumed since $\mathbf{\Pi}$ is in {\sf NP}). In this context we refer to $S$ and $T$ as the {\em domain} and {\em codomain} of the solutions.

\hfill

The class {\sf \#P} is the class of the counting versions of the {\sf NP} problems. Formally, the class {\sf \#P} is described by problems of the form:

\hfill

\textbf{Input:} $x\in \{0,1\}^n$ with $n\ge 0$.

\textbf{Output:} How many $y\in \{0,1\}^{p(n)}$ are there such that $M(x,y)$ accepts?

\hfill

\noindent where $M$ is a deterministic Turing machine running in polynomial time and $p$ is a polynomial.

\hfill

Let $\mathbf{\Pi}$ be a computational problem, and let $\mathcal{I}$ be the set of instances of $\mathbf{\Pi}$. A {\em parameter} of $\mathbf{\Pi}$ is a function $\lambda:\mathcal{I}\to \N$ computable in polynomial time\footnote{\rev{In the case of treewidth and cliquewidth, it is common to assume that the tree decomposition and cliquewidth expression are given as part of the input, since computing this is NP-hard in general.}}. The computational problem $\mathbf{\Pi}$ is said to be {\em fixed-parameter tractable} ({\sf FPT}) when parameterized by $\lambda$ if there exists a computable function $f \colon \mathbb{N} \to \mathbb{N}$  and an algorithm that solves \rev{$\mathbf{\Pi}$ in time
$f(\lambda(x)) \times \|x\|^{O(1)}$
on any instance $x\in\mathcal{I}$ (where $\|x\|$ is the size of $x$). 
}
\section{Semirings and Measures} \label{sec:semirings}

In this section we introduce {\em semiring extensions} of computational problems and our method based on {\em measures}.

\subsection{Computational problem extensions}\label{sec:comp}

Let $\mathbf{\Pi}$ be a problem in {\sf NP}, and $p$ a polynomial such that for every $n\ge 0$, the size of the certificates of instances of size $n$ is bounded by $p(n)$. Note that $\mathbf{\Pi}$ can be reformulated as:

\hfill

$\underline{\mathbf{\Pi:}}$

\textbf{Input:} An instance $x$ of $\Pi$.

\textbf{Output:} Does $x$ have a certificate?

\hfill

For example, it is easy to see that {\sc Csp} is in {\sf NP} for every finite-domain constraint language $\Gamma$ since a satisfying assignment can be used as a certificate. It is common to consider some extensions of $\mathbf{\Pi}$ consisting in asking a more general question than determining whether the set of certificates (often simply called solutions) is empty or not. For example, let us consider the following extended problems.

\hfill

$\underline{\#\mathbf{\Pi:}}$

\textbf{Input:} An instance $x$ of $\mathbf{\Pi}$.

\textbf{Output:} How many solutions does $x$ have?

\hfill

The problem $\#\mathbf{\Pi}$ is refered as the {\em counting version} of the decision problem $\mathbf{\Pi}$. Notice that if the decision problem $\mathbf{\Pi}$ is in {\sf NP}, then its counting version $\#\mathbf{\Pi}$ is in the class $\#${\sf P} \rev{(provided that we count the certificates and not a more restrictive notion of solution).}

\hfill

$\underline{\text{\sc List-}\mathbf{\Pi:}}$

\textbf{Input:} An instance $x$ \rev{of $\mathbf{\Pi}$} and for all $i\in [p(n)]$, $L_i\subseteq \{0,1\}$.

\textbf{Output:} Does $x$ have a solution $y_1,\dots,y_{p(n)}$ such that for all $i\in [p(n)]$, $y_i\in L_i$?

\hfill

\rev{For example, for $k \geq 1$ the \textsc{$k$-Coloring} problem is the problem of deciding if an input graph can be colored with at most $k$ colors, and the \textsc{List-$k$-Coloring} problem is a well-studied extension of \textsc{$k$-Coloring} where each vertex can be restricted to a list of ``allowed colors'' \cite{jensen2011graph,thomassen19953}.}

\hfill

$\underline{\text{\sc Cost-}{\mathbf{\Pi}:}}$

\textbf{Input:} An instance $x$ \rev{of $\mathbf{\Pi}$}, and for all $i\in [p(n)]$, $(C_i^0,C_i^1)\in\R^2$.

\textbf{Output:} \rev{The solution $y_1, \ldots, y_{p(n)}$ of $x$ with minimal $\sum\limits_{i=1}^{p(n)} C_i^{y_i}$.}

\hfill

The goal of this section is to show how all of these extensions of $\mathbf{\Pi}$ are subsumed by one particular extension:

\hfill

$\underline{\text{\sc Semiring-}\mathbf{\Pi:}}$

\textbf{Input:} An instance $x$ \rev{of $\mathbf{\Pi}$.}

\textbf{Output:} A join/union expression of the set of solutions of $x$. %

\hfill

See Definition \ref{def:syntax} and \ref{def:semantic} for the formal definition of a join/union expression of a set. Intuitively, the point of a join/union expression is to give a compact description of the set of solutions that factors the computations necessary to compute the counting, list and cost extensions. Note also that the {\sc Semiring-}$\mathbf{\Pi}$ problem depends on the set of ``solutions'' considered for $\mathbf{\Pi}$. However, this set will most of the time be clear from the context (and since we assume that $\mathbf{\Pi}$ is in {\sf NP} a set of certificates of polynomial size can always be chosen).

Note that the $\text{\sc Semiring-}\mathbf{\Pi}$ problems do not belong to complexity classes of the form \textsf{NP}$(\mathcal{R})$ (for any semiring $\mathcal{R}$) \cite{eiter2023semiring} since the output (a join/union expression) does not belong to any specific semiring.

\begin{reve}
\begin{example} \label{ex:projected}
    For many problems certificates and solutions exactly coincide. However, there are examples of problems where the natural notion of solution can be derived from the certificate but might not fully cover it. A well-known example is the {\em projected SAT} problem where we are tasked to decide satisfiability of a formula $\exists X, \exists Y \colon \phi(X,Y)$, where $X$ and $Y$ are Boolean variables. For this problem a certificate is a satisfying assignment over $X \cup Y$ but a solution is typically defined only with respect to $X$. While the decision variant of projected SAT is not interesting from a theoretical perspective the corresponding counting problem is a far richer with many applications~\cite{FICHTE2023103810}.
\end{example}

However, we stress that we only refer to certificate to make it possible to define these problems in general terms. For any problem where it is undesirable to use certificates rather than a more problem specific notion of a solution it is still possible to define a corresponding semiring problem: we just require the join/union expression covers the solutions, and not the certificates. Hence, in the sequel we typically use the terms certificates/solutions interchangeably.
\end{reve}

\subsection{Semirings and Dioids}

Many algorithms that are able to solve counting versions of a problem can often be adapted to solve a cost version. This adaptation uses the fact that both algorithms implicitly rely on a common algebraic structure, the {\em semiring}.

\begin{definition}[Semiring]
    A {\em semiring} is a structure $\A=(A,+,\times,0_A,1_A)$ such that $(A,+,0_A)$ is a commutative monoid, $(A,\times,1_A)$ is a monoid, $\times$ is distributive over $+$,
    $0_A$ is absorbing for $\times$, and $0_A\neq 1_A$. The semiring $\A$ is said to be {\em commutative} if $\times$ is commutative.
\end{definition}

\begin{fact}
Observe that a {\em ring} is a semiring $\A=(A,+,\times,0_A,1_A)$ where $(A,+_A,0_A)$ is a group.
\end{fact}

Throughout this paper, we implicitly assume that the operations of the semiring can be performed in constant time. \rev{This is mainly to simplify the presentation and any reasonable polynomial-time bounds would suffice for our parameterized complexity results.} For more precise computational considerations we refer the reader to the model of semiring Turing machines~\cite{eiter2023semiring}. %

\begin{example}\label{ex:semirings}
Typical examples of semirings subsume:%

\begin{itemize}
    
    \item the {\em Boolean semiring} $\mathbb{B}=(\2,\vee,\land,\bot,\top)$, where $\2:=\{\bot,\top\}$. Here, $\bot$ and $\top$ are constants (to be interpreted as true and false, respectively), and $\vee$ and $\wedge$ denote the disjunction and conjunction.
    
    \item the {\em integer semiring} $(\N,+,\times,0,1)$, where $+$ and $\times$ respectively denote the usual addition and multiplication over the natural numbers $\N$, with their respective neutral elements $0$ and $1$,  and
    
    \item the {\em tropical semiring} $\Rbar_{min}:=(\Rbar,\min,+,\infty,0)$ where $\Rbar=\R\cup \{\infty\}$ is the set of real numbers together with the neutral element $\infty$ for $\min$. We also have the  {\em dual tropical semiring} $\Rbar_{max}:=(\Rbar,\max,+,-\infty,0)$, with the neutral element $-\infty$ for $\max$.
    
\end{itemize}
\end{example}

Throughout this paper, we will further assume that $\A=(A,+,\times,0_A,1_A)$ is {\it commutative}, that is, $x\times y=y\times x$, for every $x,y\in A$.

\begin{remark}\label{rem1}
 Tropical semirings have been used in the context of finding the minimal cost of a path: the cost of the path is the product of the costs of its edges, and the best cost amongst two distinct paths is given by the sum. Here, the product of the tropical semiring is interpreted as $+$, whereas the sum is interpreted as $\min$.
\end{remark}

\rev{More generally the cost associated to an element can be modeled by {\em dioids}}.

\begin{definition}[Dioid]
A {\it dioid} is a semiring $(D,\min,+_D,\infty_D,0_D)$, where $\leq_D$ is the \rev{partial order} given by \[
\forall (d_1,d_2)\in D^2, d_1\leq_D d_2\iff \exists c\in D, \min(d_2,c)=d_1,
\] 
that orders $D$. Such a dioid $D$ is said to be {\em idempotent} if $\forall d\in D, \min(d,d)=d$, and it is said to be {\em totally ordered} if the order $\leq_D$ is total, {\it i.e.}, $\forall d_1,d_2\in D$, $d_1\leq_D d_2$ or $d_2\leq_D d_1$.
\end{definition}

\rev{In the rest of the paper, all dioids are assumed to be totally ordered, as we use dioids for modeling costs in the context of counting minimal solutions of computational problems.}

\begin{example}
The structures of Example~\ref{ex:semirings} give rise to natural examples of dioids.

\begin{itemize}
    \item The Boolean semiring $\B=(\{\bot,\top\},\lor,\land,\bot,\top)$ is endowed with a totally ordered idempotent dioid structure with the ordering $\top\le_{\B}\bot$. This ordering is rooted in Remark~\ref{rem1} where the elements of the dioid are thought of as weights to minimize. %

    \item The natural semiring $(\N,+,\times,0,1)$ ordered by the usual $\geq $ on $\N$ is a totally ordered dioid, 
    which is not idempotent.%

    \item The tropical semiring $\Rbar_{min}=(\Rbar,\min,+,\infty,0)$ is also a totally ordered idempotent dioid, its order coincides with the usual order $\le$ on $\Rbar$.

    \item If we let $\omega$ be the first infinite ordinal number \cite{kleene1938notation}, the structure $( \Omega , \min , + , \infty , 0  )$ with $\Omega:=\{ a+b\times\omega \mid (a,b)\in\N^2  \} \cup \{\infty\}$ (where the symbol $\infty$ is neutral for $\min$ and absorbing for $+$) is a totally ordered idempotent dioid. More generally, for any set $\Omega$ of ordinal numbers, where a neutral symbol $\infty$ for $\min$ has been added, $(\Omega,\min,+,\infty,0)$ is a totally ordered idempotent dioid. %

\end{itemize}
\end{example}

From now on, $\D=(D,\min,+,\infty_D,0_D)$ will denote a totally ordered idempotent dioid. We remark that the assumption of total ordering will most of the time not be necessary in our results. However, it seems natural to require that the costs are totally ordered.

Let us now verify that the sum ``$\min$'' of a dioid effectively returns the minimum of its two arguments relatively to the order $\leq_D$.

\begin{property}\label{prop:minimum}
For a totally ordered idempotent dioid $\D=(D,\min,+,\infty_D,0_D)$,  the following properties hold. \rev{For every $(d_1,d_2,d_3)\in D^3$:
\begin{enumerate}
    \item  if $d_1\leq_D d_2$, then $\min(d_1,d_2)=d_1$, and
    \item  if   $d_1\leq_D d_2$ and $d_1\le_D d_3$, then $d_1\le_D \min(d_2,d_3)$.
\end{enumerate}
}\end{property}

\begin{proof}
{\flushleft (1).} If $d_1\leq_D d_2$, then there exists $d_3\in D$ such that $d_1=\min(d_2,d_3)$. By commutativity and associativity of $\min$, it follows that

 $$\min(d_1,d_2)=\min( \min(d_2,d_3),d_2 ) = \min(\min(d_2,d_2),d_3),$$
 
and, by idempotency of $\min$, $\min(d_1,d_2)=\min(d_2,d_3)=d_1$. 

{\flushleft (2).} If   $d_1\leq_D d_2$ and $d_1\le_D d_3$, then  $\min(d_1,d_2)=\min(d_1,d_3)=d_1$. Hence,

$$\min( \min(d_2,d_3),d_1 ) = \min(\min(d_1,d_2),d_3 ) =\min(d_1,d_3) =d_1.$$

Thus $d_1\le_D \min(d_2,d_3)$.
\end{proof}

\subsection{Uplus and Join}

In dynamic programming, one generally designs algorithms that inductively compute over incomplete sets of solutions (where some solutions are missing), or over sets of solutions that only cover a subset of the whole domain (meaning that some variables are not assigned any value). In the former case, these partial sets of solutions $\F_1$ and $\F_2$ are then combined together to form a larger set of solutions via union $\cup$ or disjoint union $\uplus$.
In the latter case, two (or more) partial solutions $f_1$ and $f_2$ over disjoint subsets of the whole domain are combined to obtain a solution denoted by $(f_1\Join f_2)$ over a larger domain, an operation that is isomorphic to Cartesian product.
The idea is to establish a correspondence between the semiring operations $+$ and $\times$ and the  two operations for combining (partial) solutions. 

To simplify the presentation, we 
let $S$ and $T$ be two finite sets. Recall from Section~\ref{sec:comp} that   we aim to associate functions $S \to T$ %
with solutions to a given  computational problem. Under this interpretation, two disjoint subsets $S_1$ and $S_2$ of $S$ correspond to domains of partial solutions, which can be combined into a solution with a larger domain.

\begin{definition}[Disjoint union $\uplus$]

Let $\F_1 \subseteq T^S$ and $\F_2 \subseteq T^S$, with $\F_1\cap \F_2 =\emptyset$. The {\it disjoint union} $\F_1\uplus \F_2$ is defined by
$$ \F_1\uplus\F_2 = \{ f \in T^S \mid f\in \F_1 \text{ or } f\in \F_2 \}. $$
\end{definition}

Notice that $\emptyset$ is neutral for $\uplus$: for every $\F\subseteq T^S$. Indeed, $\F\uplus\emptyset$ is well defined and $\F\uplus\emptyset=\F$. We can observe the correspondence  between $\emptyset$ and $\uplus$, and $0_{\mathcal{A}}$ and $+$ of a semiring $\mathcal{A}$. This observation will be further exploited %
in Section \ref{sec:measures,matrix}.

We proceed by describing the operator corresponding to the product $\times$ in the underlying semiring. 
Here, the idea is to combine partial functions $f_1$ and $f_2$ over two disjoint domains $S_1$ and $S_2$ into a function $f_1\Join f_2$ with the greater domain $S_1\uplus S_2$. The condition that $S_1$ and $S_2$ are disjoint is indeed necessary to ensure that $f_1\Join f_2$ is well-defined.

\begin{definition}[Join $\Join$]
Given $f_1\in T^{S_1}$ and $f_2\in T^{S_2}$, we define the {\em join} $f_1\Join f_2 \in T^{S_1\uplus S_2}$ of $f_1$ and $f_2$ by

$$(f_1\Join f_2)(s)  =  \left\{\begin{array}{cc}
     f_1(s) & \text{ if }s\in S_1  \\
     f_2(s) & \text{ if }s\in S_2
\end{array}\right.$$
\end{definition}

Then, given $\F_1 \subseteq T^{S_1}$ and $\F_2 \subseteq T^{S_2}$, the set $\F_1\Join\F_2$ is defined as

$$ \F_1\Join\F_2 = \{ f_1\Join f_2 \mid (f_1,f_2)\in \F_1\times\F_2 \} $$

Note that $T^{\emptyset}$ is neutral for $\Join$: for every $\F\subseteq T^S$: $\F\Join T^{\emptyset}$ is well defined and $\F\Join T^{\emptyset}=\F$.
Also, $\emptyset$ is absorbing for $\Join$: for every $\F\subseteq T^S$: $\F\Join\emptyset$ is well defined and $\F\Join\emptyset=\emptyset$. Thus, to relate to semirings, if $\Join$ corresponds to the product of the semiring, it is expected that $\emptyset$ corresponds to the zero, and that $T^{\emptyset}$ correspond to the unit. This remark justifies the links between $T^{\emptyset}$ and the unit element of the semiring developed in Section \ref{sec:measures,matrix}. 

This operator, although typically not explicitly named, is central in dynamic programming. For example,  ``divide-and-conquer algorithms'' works by solving smaller sub-problems and then ``combining'' the solutions of the small problems to a solution to the larger problem. The way small solutions are combined is often implicit with the aforementioned $\Join$ operator. However, it should be noted that this operator is typically not given an explicit name.

\begin{example}

Recall that a {\em cograph} $G$ is by definition either

\begin{itemize}
    \item the graph $K_1$ (i.e., $G$ is a single vertex),
    \item the disjoint union $G=G_1 \oplus G_2 := (V_{G_1}\uplus V_{G_2},E_{G_1}\uplus E_{G_2})$ of two cographs $G_1$ and $G_2$ with disjoint set of vertices,
    \item the {\em joint union} $G=G_1\otimes G_2 = (V_{G_1}\uplus V_{G_2},E_{G_1}\uplus E_{G_2} \uplus \{ \{u_1,u_2\} \mid (u_1,u_2)\in V_{G_1}\times V_{G_2}\})$ of two cographs $G_1$ and $G_2$ with disjoint sets of vertices.
\end{itemize}

As an example of an implicit use of the $\Join$ operator, note that deciding whether a cograph $G$ on $n$ vertices is $k$-colorable (with $k\ge 1$) can be done in $O(kn)$ time by dynamic programming, by the following observations.

\begin{itemize}
    \item $K_1$ is $k$-colorable.
    \item $G_1\oplus G_2$ is $k$-colorable if and only if $G_1$ and $G_2$ are $k$-colorable.
    \item $G_1\otimes G_2$ is $k$-colorable if and only if there exists $k'\in [k-1]$ such that $G_1$ is $k'$-colorable and $G_2$ is $(k-k')$-colorable.
\end{itemize}

Indeed, the second item relies implicitly on the fact that if $G_1$ and $G_2$ are both $k$-colorable by the $k$-colorings $f_1$ and $f_2$, then so is $G_1\oplus G_2$, with the $k$-coloring $f_1\Join f_2$. A similar remark also holds for the third item. Note that in order to efficiently compute a $k$-coloring of the cograph with this method, the $\Join$ operator would be used implicitly.

\end{example}

\subsection{Join/union expressions}\label{sec:join/union expressions}

We now turn to the problem of building a join/union expression representing the set of solutions of a problem instance. 

\begin{definition}[Syntax of an expression]\label{def:syntax}

A {\em join/union expression} $E$ over the sets $(S,T)$ is an element of the grammar:

\[E: \emptyset \ | \ T^{\emptyset} \ | \ (s\mapsto t) \ | \ E\uplus E \ | \ E\Join E\]

with $(s,t)\in S\times T$.
\end{definition}

The {\em size} of an expression $E$ denoted $\|E\|$ is its number of leaves (viewing $E$ as a binary tree, where the leaves are labeled by $\emptyset$, $T^{\emptyset}$ and the $(s\mapsto t)$ for $(s,t)\in S\times T$, and where the internal nodes are labeled by either $\uplus$ or $\Join$).

The {\em domain} of an expression $E$ (denoted $dom(E)$) is the set of all $s\in S$ such that $E$ has a leaf of the form $(s\mapsto t)$ with $t\in T$.

\begin{definition}[Semantics of an expression]\label{def:semantic}

Given a {\em join/union expression} $E$ over the sets $(S,T)$, its {\em semantic} $[E]$ is either a special symbol FAIL, or a subset of $T^{dom(E)}$ defined inductively as:

\begin{itemize}

    \item $[\emptyset] = \emptyset$,
    
    \item $[T^{\emptyset}] = T^{\emptyset}$,
    
    \item $[s\mapsto t]= \{(s\mapsto t)\}$ for all $(s,t)\in S\times T$,
    
    \item $[E_1\uplus E_2]=\left\{ \begin{array}{cc}
        \text{FAIL} & \text{ if } [E_1] = \text{FAIL} \text{ or } [E_2]=\text{FAIL} \\
        \empty [E_1] & \text{ else, if } [E_2] = \emptyset \\
        \empty [E_2] & \text{ else, if } [E_1] = \emptyset \\
        \text{FAIL} & \text{ else, if } dom(E_1) \neq dom(E_2) \\
        \text{FAIL} & \text{ else, if }  [E_1]\cap [E_2]\neq \emptyset \\
        \empty [E_1]\uplus [E_2] & \text{else}, \\
    \end{array} \right.$
    
    \item $[E_1\Join E_2]=\left\{ \begin{array}{cc}
        \text{FAIL} & \text{ if } [E_1] = \text{FAIL} \text{ or } [E_2]=\text{FAIL} \\
        \text{FAIL} & \text{ else, if } dom(E_1) \cap dom(E_2) \neq \emptyset  \\
        \empty [E_1]\Join [E_2] & \text{else}. \\
    \end{array} \right.$

\end{itemize}
\end{definition}

We say that a join/union expression $E$ is {\em legal} if  $[E]\neq$ FAIL.
A {\em join/union expression for a set} $\F\subseteq T^S$ is a legal join/union expression $E$ such that $\F=[E]$. Alternatively, we say that $E$ {\em encodes} or {\em represents} $\F$ if $[E]=\F$.

\begin{example}

\rev{Up to allowing the commutation of the operations $\biguplus$ and $\Join$, and the distributivity of $\Join$ over $\biguplus$ (which do not affect the semantics of expressions), the set of join/union expressions over the sets $(S, T)$ forms a semiring, where $\biguplus$ plays the role of addition and $\Join$ that of multiplication.}

\rev{This semiring shares a key property with free structures such as free groups: every expression (i.e., every join/union expression) admits a unique normal form, up to the equivalence induced by the semiring axioms.
That is, if two expressions represent the same element of the semiring, then one can be transformed into the other using only the identities of commutativity, associativity, and distributivity.}

\end{example}

\subsection{Measures and Matrix Representations}\label{sec:measures,matrix}

Using dynamic programming, it is sometimes possible to compute a join/union expression of the set of solutions of an instance of a computational problem $\mathbf{\Pi}$. Then, in order to take advantage of this join/union expression to solve computational variants of $\mathbf{\Pi}$ (such as $\#\mathbf{\Pi}$ or {\sc Cost}-$\mathbf{\Pi}$) we introduce the concept of a {\em measure}, which takes  its values in a semiring while mapping the fundamental operations $\uplus$ and $\Join$ of the join/union expressions to the sum and product of the semiring.

\begin{definition}[Measure]

Let $\A$ be a commutative semiring. An $\A$-{\em measure} over the finite sets $(S,T)$ is a function $\mu$ such that for all $S' \subseteq S$ and $\F\subseteq T^{S'}$, $\mu$ maps $\F$ to $\mu(\F)\in A$, respecting:

\begin{enumerate}

    \item \textbf{zero axiom:} $\mu(\emptyset)=0$,
    
    \item \textbf{unit axiom:} $\mu(T^{\emptyset})=1$,
    
    \item \textbf{additivity:}  for $\F_1,\F_2 \subseteq T^{S'}$ disjoint:
    $\mu(\F_1\uplus \F_2) = \mu(\F_1) + \mu(\F_2)$,
    
    \item\textbf{elementary multiplicativity:}  for all $f_1\in T^{S_1}$ and $f_2\in T^{S_2}$ (with $S_1$ and $S_2$ disjoint):
$\mu(\{f_1\Join f_2\}) = \mu(\{f_1\})\times\mu(\{f_2\})$.

\end{enumerate}
\end{definition}

Once we know an expression of the set SOL of solutions of an instance of a computational problem, one deduces easily for any $\A$-measure $\mu$ how to compute $\mu(SOL)$ as a sum and product of the $\mu(\{s\mapsto t\})$ with $(s,t)\in S\times T$ by transforming the $\uplus$ into $+$, and the $\Join$ into $\times$.

Let us remark that the distributivity of $\times$ over $+$ in the semiring ensures that the elementary multiplicativity axiom of measures can be extended to true multiplicativity (over whole sets).

\begin{property}[\textbf{multiplicativity}]\label{prop:multiplicativity_sets}

Let $\mu$ be an $\A$-measure. Let $\F_1\subseteq T^{S_1}$ and $\F_2\subseteq T^{S_2}$. Then
$\mu(\F_1\Join\F_2)=\mu(\F_1)\times\mu(\F_2).$

\end{property}

\begin{proof}
The proof follows from the following sequence of identities that are not difficult to verify:
\begin{eqnarray*}
\mu(\F_1\Join\F_2)&=& \mu(\{ f_1\Join f_2 \mid (f_1,f_2)\in\F_1\times\F_2 \}) \quad (\text{By definition of }\F_1\Join\F_2)\\
&=& \mu(\underset{(f_1,f_2)\in\F_1\times\F_2}{\uplus} \{f_1\Join f_2\})\\
&=&
\sum\limits_{(f_1,f_2)\in\F_1\times\F_2} \mu(\{f_1\Join f_2\}) \quad  (\text{By additivity of }\mu)\\
&=&
\sum\limits_{(f_1,f_2)\in\F_1\times\F_2} \mu(\{f_1\}) \times \mu(\{f_2\}) \quad  (\text{By elementary multiplicativity of }\mu)\\
&=&
(\sum\limits_{f_1\in\F_1}\mu(\{f_1\}))\times (\sum\limits_{f_2\in\F_2}\mu(\{f_2\})) \quad  (\text{By distributivity of }\times\text{ over }+)\\
&=& \mu(\underset{f_1\in\F_1}{\uplus} \{f_1\}) \times \mu(\underset{f_2\in\F_2}{\uplus} \{f_2\}) \quad  (\text{By additivity of }\mu) \\
&=&\mu(\F_1)\times\mu(\F_2). \qedhere
\end{eqnarray*}
\end{proof}

We summarize the conclusion of Property \ref{prop:multiplicativity_sets} by saying that {\em measures are multiplicative}.
Given a set $\F\subseteq T^S$, and a measure $\mu$, we then realize that $\mu(\F)$ only depends\footnote{Similarly to how a linear application depend only on its image on a linear basis.} on the values of  $\mu(\{s\mapsto t\})$, for $(s,t)\in S\times T$. Indeed, remarking 

$$\F=\underset{f\in\F}{\uplus} \{f\} = \underset{f\in\F}{\uplus} \underset{s\in S}{\Join} \{s\mapsto f(s)\},$$ 

it then follows from the properties of measures that 

$$\mu(\F)=\sum\limits_{f\in\F} \prod\limits_{s\in S} \mu( \{s\mapsto f(s)\} ).$$

Now that we have identified the minimal information to describe a measure, we can easily represent this as a matrix.

\begin{definition}[Matrix of a measure]

The {\em matrix of the measure $\mu$} is the matrix $M_{\mu}=(\mu(\{s\mapsto t\}))_{(s,t)\in S\times T}\in A^{S\times T}$.

\end{definition}

Reciprocally, we can associate a measure $\mu_M$ to every matrix $M\in A^{S\times T}$ by $\forall S'\subseteq S, \forall \F\subseteq T^{S'}$,  $\mu_M(\F) := \sum\limits_{f\in\F} \prod\limits_{s\in S'} M[s,f(s)]$. It is indeed easy to verify that the function $\mu_M$ defined thereby is a measure, and that for every measure $\nu$, $\mu_{M_{\nu}} = \nu$ and that for every matrix $N$, $M_{\mu_N}=N$.

The practical relevance of join/union expressions and measures is summarized in Lemma \ref{lem:expression_computation}. Here, the link between join/union expressions and measures is ensured by the additivity and multiplicativity axioms of measures. %

\begin{lemma}\label{lem:expression_computation}

Given a legal join/union expression $E$ over the finite sets $S$ and $T$, and an $\A$-measure $\mu$ represented by its matrix $M_{\mu}=(M_{\mu}[s,t])_{s,t}\in A^{S\times T}$, $\mu([E])$ can be computed in $O(\|E\|)$ time.

\end{lemma}

\begin{proof}

We inductively compute the result by following the structure of a join/union expression.

\begin{itemize}
    \item If $E=\emptyset$ we can return $0$ by the zero axiom.
    \item If $E=T^{\emptyset}$ we can return $1$ by the unit axiom.
    \item If $E=(s\mapsto t)$ with $(s,t)\in S\times T$, we can return $M_{\mu}[s,t]$ by definition of the matrix $M_{\mu}$ of $\mu$.
    \item If $E=E_1\uplus E_2$, inductively compute $\mu([E_1])$ and $\mu([E_2])$. We can then return $\mu([E_1])+\mu([E_2])$ by additivity.
    \item If $E=E_1\Join E_2$, inductively compute $\mu([E_1])$ and $\mu([E_2])$. We can then return $\mu([E_1])\times\mu([E_2])$ by multiplicativity (Property \ref{prop:multiplicativity_sets}).
\end{itemize}

The algorithm performs as many operations as there are nodes in the join/union expression $E$ seen as a binary tree. Recall that any binary tree with $m\ge 1$ leaves has $2m-1$ nodes in total: the number of nodes of a join/union expression is linear in its size.
Hence, we compute $\mu([E])$ in $O(\|E\|)$ time.
\end{proof}

The main interest of Lemma \ref{lem:expression_computation} lies in the fact that many usual computational extension of {\sf NP} problems can be formulated as computing the image of the set of solutions by a semiring-measure.

\begin{example}\label{example:measures}
Given a computational problem $\mathbf{\Pi}$ of {\sf NP}, the problems $\mathbf{\Pi}$, $\#\mathbf{\Pi}$, {\sc List}-$\mathbf{\Pi}$ and {\sc Cost}-$\mathbf{\Pi}$ all consists in computing the image by a measure of the set of solutions SOL of $\mathbf{\Pi}$. For all of these extensions, let us now consider the semiring and the matrix of the corresponding measure.
\begin{itemize}
    \item \rev{Solving $\mathbf{\Pi}$ means computing $\mathbf{1}_{\neq\emptyset}$(SOL), with $\mathbf{1}_{\neq\emptyset}(\F) = \left\{ \begin{array}{cc}
         \top & \text{ if } \F\neq\emptyset  \\
         \bot & \text{ if } \F=\emptyset
    \end{array} \right.$
$\mathbf{1}_{\neq\emptyset}$ is a $\B$-measure, and its matrix only has $\top$ coefficients.}

    \item \rev{Solving $\#\mathbf{\Pi}$ means computing $\#$(SOL), with $\#(\F) = |\F|$.
$\#$ is a $\N$-measure, and its matrix only has $1$ as coefficients.}

    \item \rev{Solving {\sc List}-$\mathbf{\Pi}$ means computing $L$(SOL) with $M_L$ the $\B$-measure given by: for all $(s,t)\in S\times T$, $M_L[s,t]=\top$ if sending $s$ to $t$ is allowed, and $M_L[s,t]=\bot$ otherwise.
}
\item \rev{Solving {\sc Cost}-$\mathbf{\Pi}$ means computing $C$(SOL), where $C$ is the $\mathbb{R}_{\min}$-measure of the cost matrix $M_C$, and $M_C[s,t] \in \mathbb{R}$ is the cost of mapping $s$ to $t$.
}\end{itemize}

\end{example}

\subsection{Computing join/union expressions} \label{sec:semiring_expressions}

In order to take advantage of Lemma \ref{lem:expression_computation}, it is useful to be able to compute a join/union expressions of small size representing the set of solutions of the considered problem. 
Note that small join/union expression can encode large sets. For instance, for $n\ge 0$, the set $[n]^{[n]}$ of all functions from $[n]$ to $[n]$ has the join/union expression of size $n^2$:

$$ E_{[n]^{[n]}} = \overset{n}{\underset{i=1}{ \Join}}  (\overset{n}{\underset{j=1}{\uplus}} (i\mapsto j) ).$$

However, this join/union expression does not seem relevant to efficiently encode the set $\mathfrak{S}_n$ of all permutations over $[n]$. Intuitively, the operation $E_1\Join E_2$ does not allow one to choose the image in the domain of $E_2$ dependent on the images of the elements chosen in $E_1$. Indeed, as we will now prove, such a join/union expression is unlikely to be computable in polynomial time.

\begin{theorem}\label{thm:no_semiring_permutations}

Unless \rev{{\sf P=\#P}}, a join/union expression of $\mathfrak{S}_n$ can not be computed in \rev{time $P(n)$, for any polynomial $P$, even if $n$ is represented in unary}.

\end{theorem}

\begin{proof}

Recall the definition of the {\sc Permanent} problem:

\hfill

\underline{\sc Permanent}:

\textbf{Input:} A matrix $M\in \{0,1\}^{n\times n}$ with $n\ge 0$.

\textbf{Output:} The permanent of $M$, {\it  i.e.,} $perm(M):=\sum\limits_{\sigma\in \mathfrak{S}_n} \prod\limits_{i=1}^n m_{i,\sigma(i)}$.

\hfill

The {\sc Permanent} problem is well-known to be {\sf \#P}-complete due to Valiant's theorem~\cite{DBLP:journals/tcs/Valiant79}. Therefore, if it is solvable in \rev{time $P(n)$ for any polynomial $P$}, then \rev{{\sf P=\#P}}.

Assume that there exists an algorithm that takes an integer $n$ as input \rev{(represented in unary encoding)} and returns a join/union expression $E_{\mathfrak{S}_n}$ of $\mathfrak{S}_n$ \rev{in time $P(n)$ for some polynomial $P$}. This expression has size \rev{at most $P(n)$} since it has been computed in \rev{$P(n)$} time.

Let $M\in \{0,1\}^{n\times n}$ with $n\ge 0$ be an instance of {\sc Permanent}.
Then, by Lemma \ref{lem:expression_computation} the $\N$-measure $\mu_M$ of matrix $M$ applied to the set $\mathfrak{S}_n$ can be computed in \rev{time $O(P(n))$} time, {\it  i.e.,} $\mu_{M}(\mathfrak{S}_n)=\sum\limits_{\sigma\in\mathfrak{S}_n} \prod\limits_{i=1}^n m_{i,\sigma(i)} = perm(M)$ can be computed in \rev{time $O(P(n))$}. This proves that \rev{{\sf P=\#P}}.
\end{proof}

By Theorem \ref{thm:no_semiring_permutations}, we know that a join/union expression of $\mathfrak{S}_n$ is unlikely to be computable in $poly(n)$ time. We view it as an interesting open question to prove this unconditionally and pose the following conjecture.

\begin{conjecture}\label{conj:no_semiring_permutations} \rev{For all polynomial $P$, there exists an integer $n$ such that 
$\mathfrak{S}_n$ does not have a join/union expression of size at most $P(n)$, even if $n$ is represented in unary.}
\end{conjecture}

Note that a join/union expression $E_{\mathfrak{S}_n}$ of $\mathfrak{S}_n$ can be defined recursively by $E_{\mathfrak{S}_{n+1}}=\underset{m\in [n+1]}{\uplus} ((n+1\mapsto m) \Join E_{\mathfrak{S}_n}[m\leftarrow n+1])$, with $E_{\mathfrak{S}_n}[m\leftarrow n+1]$ being a copy of $E_{\mathfrak{S}_n}$ where every leaf of the form $(i\mapsto m)$ have been replaced by a leaf $(i\mapsto n+1)$. However, this would result in a formula with $n!$ leaves. The issue is that the replacement of leaves must be operated dependent on the image we have chosen for $n+1$, which rules out the possibility of an efficient factorization of the expression.

A similar result can be proven for the set $\mathfrak{C}_n$ of $n$-cycles of $\mathfrak{S}_n$.

\begin{theorem}\label{thm:no_semiring_cycles}

Unless {\sf P=NP}, a join/union expression of $\mathfrak{C}_n$ can not be computed in \rev{time $P(n)$, for any polynomial $P$}.

\end{theorem}

\begin{proof}

Recall the definition of the {\sc Traveling-Salesman} problem.

\hfill

\underline{\sc Traveling-Salesman}:

\textbf{Input:} A graph $G$ on $n\ge 0$ vertices, and a weight function $w:E_G\mapsto\R_+$.

\textbf{Output:} The minimal weight of a Hamiltonian cycle in $G$.

\hfill

The {\sc Traveling-Salesman} problem is known to be {\sf NP}-complete. Therefore, if it is solvable in time \rev{$P(n)$ for some polynomial $P$}, then {\sf P=NP}.
Assume that there exists an algorithm that takes as an input an integer $n$, and returns a join/union expression $E_{\mathfrak{C}_n}$ of $\mathfrak{C}_n$ in\rev{time $P(n)$, for some polynomial $P$}. This expression has size \rev{at most $P(n)$} since it has been computed in \rev{$P(n)$} time.

Let $(G,w)$ an instance of {\sc Traveling-Salesman}. \rev{Let $n = |V_G|$ be the number of vertices and assume without loss of generality that $V_G=[n]$.}
Consider the matrix $W=(w_{i,j})\in \Rbar^{[n]\times [n]}$, with for all $(i,j)\in [n]^2$, 
$$w_{i,j}= \left\{ \begin{array}{cc}
     w(\{i,j\}) & \text{ if }\{i,j\}\in E_G \\
     \infty & \text{ if }\{i,j\}\notin E_G
\end{array} \right.$$

By Lemma \ref{lem:expression_computation}, the $\Rbar_{min}$-measure $\mu_{W}$ (of the tropical semiring $(\Rbar,\min,+,\infty,0)$) of the matrix $W$ applied to the set $\mathfrak{C}_n$ can be computed in $poly(n)$ time, {\it  i.e.,} $\mu_{W}(\mathfrak{C}_n)=\min\limits_{c\in\mathfrak{C}_n} \sum\limits_{i=1}^n w_{i,c(i)}$ can be computed in \rev{$O(P(n))$} time. This solves {\sc Traveling-Salesman} in polynomial time, implying that {\sf P=NP}.
\end{proof}

Again, Theorem \ref{thm:no_semiring_cycles} indicates that a join/union expression of $\mathfrak{C}_n$ of size $poly(n)$ is unlikely to exist, and we pose the following explicit conjecture.

\begin{conjecture}\label{conj:no_semiring_cycles}

\rev{For all polynomial $P$, there exists an integer $n$ such that the set $\mathfrak{C}_n$ of $n$-cycles of $\mathfrak{S}_n$ does not have a join/union expression of size at most $P(n)$, even if $n$ is represented in unary}.

\end{conjecture}

One can also remark that Conjecture \ref{conj:no_semiring_cycles} implies that a join/union expression of size \rev{$P(n)$ (for some polynomial $P$)} does not exist for the set of {\em all cycles} of $\mathfrak{S}_n$. Indeed, if such an expression exists, we can use it to obtain a join/union expression of size $poly(n)$ of $\mathfrak{C}_n$ by replacing every leaf of the form $(i\mapsto i)$ by $\emptyset$, contradicting Conjecture \ref{conj:no_semiring_cycles}.

\subsection{Differences with {\sc Sum-Product-Csp}} \label{sec:sum_product}

At first, the concept of a measure may seem very reminiscent of the {\sc Sum-Product-Csp} framework \cite{eiter2021complexity,fan2023fine}. Here, we are given a semiring $\A$, a domain $D$ and a set of constraints $\Gamma$ where each constraint is a polynomial time computable function $C:D^{ar(C)}\mapsto A$, where $ar(C)\ge 0$ is the {\em arity} of $C$. The problem {\sc Sum-Product-Csp}$(\Gamma)$ is then defined as follows.

\hfill

$\underline{\text{\sc Sum-Product-Csp}(\Gamma):}$

\textbf{Input:} A set of variable $V$, and a finite set of constraints $\{ (C_i,x_i^1,\dots,x_i^{ar(C)}) \mid i\in [m] \}$ with $m\ge 0$, and $\forall i\in [m], C_i\in\Gamma$.

\textbf{Output:} The value $\sum\limits_{f\in D^V} \prod\limits_{i=1}^m C_i(f(x_i^1),\dots,f(x_i^{ar(C)})) \in A$.

\hfill

Note that the constraints defined in the context of {\sc Sum-Product-Csp} subsumes the less general constraints defined for {\sc Csp}, as a relation $C$ over $D$ of arity $ar(C)$ can be seen as a function $C: D^{ar(C)} \mapsto \{\bot,\top\}$ taking its value in the Boolean semiring. With this identification, {\sc Csp} can be seen as a the particular case of {\sc Sum-Product-Csp}, where the only semiring considered is the Boolean semiring $\B$.

However, the general {\sc Sum-Product-Csp} is not subsumed by the {\sc Semiring-Csp} problem. The main difference is that we in the {\sc Sum-Product-Csp} problem use the semiring to relax the constraints, which can now take arbitrary values in a semiring instead of only $\top$ or $\bot$ in the Boolean semiring.
In comparison, in our {\sc Semiring} approach we can only work with a combination of arbitrary arity ``hard constraints'' (either true or false), and only unary ``relaxed constraints'' (taking values in the semiring).
However, recall that our approach applies to every {\sf NP} problem and not only {\sc Csp}, since it is defined as long as the considered problem has a ``set of solutions'', for instance the set of certificate of the {\sf NP} problem. For example, the {\sc Connected-Dominating-Set} problem does not benefit from the {\sc Sum-Product-Csp} formalism, but the {\sc Semiring-Connected-Dominating-Set} problem is perfectly well defined, and subsumes the counting, list and cost version of {\sc Connected-Dominating-Set}. %

Moreover, even if we restrict the study to {\sc Csp}s, the {\sc Sum-Product-Csp} approach is often too general to benefit from unified algorithms, and we frequently need to distinguish cases depending on the nature of the semiring~\cite{eiter2021complexity}. In contrast, any upper bound on the complexity of {\sc Semiring-Csp} is automatically derived as an upper bound of {\sc Csp}, {\sc \#Csp}, {\sc List-Csp}, and {\sc Cost-Csp}. %

\section{Combining Semirings: the $\Delta$-product}\label{sec:Delta-product}

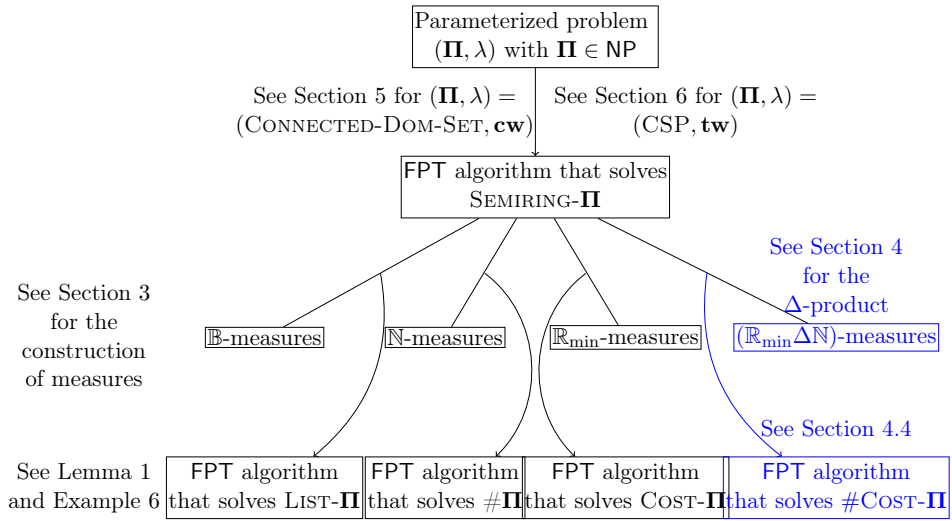
\begin{figure}[!ht]
\begin{center}
\scalebox{0.8}{
\begin{tikzpicture}

\tikzstyle{vertex}=[circle, draw, inner sep=1pt, minimum width=6pt]
\tikzstyle{sq}=[rectangle, draw, inner sep=1pt, minimum width=6pt]

\node[sq] (NPproblem) at (0,5) {$\begin{matrix}\text{Parameterized problem } \\ (\mathbf{\Pi},\lambda)\text{ with }\mathbf{\Pi}\in{\sf NP}\end{matrix}$};

\node[sq] (SemiringNP) at (0,2.5) {$\begin{matrix}{\sf FPT}\text{ algorithm that solves} \\ \textsc{Semiring-}\mathbf{\Pi}\end{matrix}$}
    edge[<-] (NPproblem);

\node () at (-2.5,3.75) {$\begin{matrix} \text{See Section }\ref{sec:dom-set}\text{ for } (\mathbf{\Pi},\lambda)= \\ (\textsc{Connected-Dom-Set},\mathbf{cw}) \end{matrix}$};

\node () at (2.5,3.75) {$\begin{matrix} \text{See Section }\ref{sec:Csp}\text{ for } (\mathbf{\Pi},\lambda)= \\ (\textsc{CSP},\mathbf{tw}) \end{matrix}$};

\node () at (-7.5,0) {$\begin{matrix} \text{See Section }\ref{sec:semirings} \\ \text{for the} \\\text{construction} \\ \text{of measures}\end{matrix}$};

\node[sq] (B) at (-4.5,0) {$\B$-measures}
    edge (SemiringNP);

\node[sq] (N) at (-1.5,0) {$\N$-measures}
    edge (SemiringNP);

\node[sq] (Rmin) at (1.5,0) {$\R_{\min}$-measures}
    edge (SemiringNP);

\node[sq,color=blue] (Delta) at (5,0) {$(\R_{\min}\Delta\N)$-measures}
    edge (SemiringNP);

\node[color=blue] () at (5,1) {$\begin{matrix} \text{See Section }\ref{sec:Delta-product} \\ \text{for the} \\ \Delta\text{-product}\end{matrix}$};

\node (midB) at (-2.58,1.2) {};
\node (midN) at (-0.96,1.2) {};
\node (midRmin) at (0.96,1.2) {};
\node (midDelta) at (2.85,1.2) {};

\node[sq] (Decision) at (-4.5,-2.5) {$\begin{matrix} \textsf{FPT}\text{ algorithm} \\ \text{that solves }\textsc{List-}\mathbf{\Pi}  \end{matrix}$}
    edge[<-,bend right] (midB);

\node[sq] (Counting) at (-1.5,-2.5) {$\begin{matrix} \textsf{FPT}\text{ algorithm} \\ \text{that solves \#}\mathbf{\Pi}  \end{matrix}$}
    edge[<-,bend right=50] (midN);

\node[sq] (Cost) at (1.5,-2.5) {$\begin{matrix} \textsf{FPT}\text{ algorithm} \\ \text{that solves }\textsc{Cost-}\mathbf{\Pi}  \end{matrix}$}
    edge[<-,bend left=50] (midRmin);

\node[sq,color=blue] (CountCost) at (5,-2.5) {$\begin{matrix} \textsf{FPT}\text{ algorithm} \\ \text{that solves }\#\textsc{Cost-}\mathbf{\Pi}  \end{matrix}$}
    edge[<-,bend left,color=blue] (midDelta);

\node () at (-7.5,-2.5) {$\begin{matrix}\text{See Lemma } \ref{lem:expression_computation} \\ \text{and Example }\ref{example:measures} \end{matrix}$};

\node[color=blue] at (5,-1.5) {See Section \ref{sec:count_cost_application}};

\end{tikzpicture}
}
\end{center}
\caption{A general overview of our contributions where the contribution in this section is colored in blue.}
\label{fig:plan_of_the_paper_Delta}
\end{figure}

We have seen that semiring extensions generalize many natural extensions of computational problems such as counting, list and cost. We now explain how to {\em combine} some of these algorithmic extensions together by defining a novel operation over semirings called the {\em $\Delta$-product}.
As a concrete application, we will see that Lemma \ref{lem:expression_computation} enables us to ``count the number of solutions of minimal cost''. Therefore, solving the semiring version of an {\sf NP} problem is sufficient to count its solution of minimal cost. To the best of our knowledge, the algebraic constructions we present are completely novel, and no general semiring based algorithms for counting solutions of minimal cost exist in the literature. This new tool enables to create a semiring that makes the class {\sf \#$\cdot$OptP} \cite{hermann2009complexity} fall into the semiring formalism, similarily as the classes {\sf \#P} \cite{DBLP:journals/tcs/Valiant79} corresponds to the natural semiring $(\N,+,\times,0,1)$ for instance, via the $\N$-measure \rev{$\#(\F) = |\F|$.}

The contribution of this section to the semiring formalism presented in this paper is summarized in Figure \ref{fig:plan_of_the_paper_Delta}.

For the rest of this section, we let $\D=(D,\min,+_D,\infty_D,0_D)$ be a totally-ordered idempotent commutative dioid and $\A=(A,+_A,\times,0_A,1_A)$ be a commutative semiring. \rev{We assume the semiring to be commutative, as the product of the semiring will eventually be mapped to Joins through measues, and the Join is commutative. Note that every exemples of semiring provided in this paper are indeed commutative.}

\subsection{Regularity}

We begin by introducing the notion of multiplicative regularity.
Let $w$ a $\D$-measure over some finite sets $(S,T)$, and $S'\subseteq S$. For $\F\subseteq T^{S'}$, we say that a function $f\in\F$ is of {\em minimal weight} in $\F$ (with respect to $w$) if $w(\{f\})=w(\F)$. This definition is motivated by the fact that $w(\F) = \min\limits_{f\in \F} w(\{f\})$ (since $\F$ is a $\D$-measure). We also denote by ${\argmin}_w(\F)$ the sets of elements of minimal weight of $\F$ (with respect to $w$).

Given a join/union expression $E$ of $\F$, in order to count the number of solutions of minimal weight, we in particular need to handle the case where $E$ is of the form $E_1\Join E_2$ in order to perform a structural induction. Ideally, it would be desirable to have ${\argmin}_w([E]) = {\argmin}_w([E_1])\Join {\argmin}_w([E_2])$. The issue is that this equality is not always true, an instance where it fails is is illustrated in Example~\ref{ex:motivation_regularity}.

\begin{example}\label{ex:motivation_regularity}

Let $S_1$ and $S_2$ be two disjoint subsets of $S$, $f_1\in T^{S_1}$, $f_2\in T^{S_2}$ and $f'_2\in T^{S_2}$, with $f_2\neq f'_2$. Let also $\F_1:=\{f_1\}$ and $\F_2:=\{f_2,f'_2\}$. Finally, let $w$ be a $\Rbar_{min}$-measure (i.e., of the tropical semiring) with $w(\{f_1\})=\infty$, $w(\{f_2\})=7$ and $w(\{f'_2\})=31$.

\begin{itemize}
    
    \item We have ${\argmin}_w(\F_1\Join\F_2) = \{ f_1\Join f_2 , f_1\Join f'_2 \}$, as both functions $f_1\Join f_2$ and $f_1\Join f'_2$ have weight $\infty=\infty+7=\infty+31$ by multiplicativity of $w$ (Property \ref{prop:multiplicativity_sets}). Recall that $+$ is the {\em product} of the tropical semiring.

    \item However, ${\argmin}_w(\F_2) = \{f_2\}$, since $w(f_2)=7 < 31=w(f'_2)$. Thus ${\argmin}_w(\F_1)\Join{\argmin}_w(\F_2) = \{f_1\Join f_2\}$.

\end{itemize}

We see that in this case, ${\argmin}_w(\F_1\Join\F_2) \neq {\argmin}_w(\F_1)\Join{\argmin}_w(\F_2)$.

\end{example}

The issue in Example \ref{ex:motivation_regularity} comes from the property of $\infty$ that $\infty + 7 = \infty + 31$, even though $7<31$. In order to guarantee that ${\argmin}_w([E_1\Join E_2])={\argmin}_w([E_1])\Join{\argmin}_w([E_2])$, we need to separately handle the elements that behaves similarly to $\infty$ in Example \ref{ex:motivation_regularity}. In Definition \ref{def:mult_regularity}  we introduce an algebraic property that guarantees that this undesirable behaviour never occurs.

\begin{definition}[Multiplicative Regularity]\label{def:mult_regularity}

A {\em multiplicatively regular element} of $\A$ is an element $a\in A\setminus\{0\}$ such that 

$$\forall (b,c)\in A^2, a\times b=a\times c \implies b=c.$$

\end{definition}

We denote by $reg(\A)$ the set of multiplicatively regular elements of the semiring $\A$, \rev{and let $\overline{reg}(\A):=A\setminus reg(\A)$  be its complement.}

We have the following link between multiplicatively regular elements and products.

\begin{property}\label{prop:mult regular}

Let $(a,b)\in A^2$. Then $a\times b\in reg(\A) \iff (a,b)\in reg(\A)^2$.

\end{property}

\begin{proof}
 We prove the two cases separately.
\begin{itemize}
    
    \item Assume that $(a,b)\in reg(\A)^2$.
    Let $(c,d)\in A^2$ with $(a\times b)\times c = (a\times b)\times d$.
    Then $a\times(b\times c) = a\times (b\times d)$. Since $a$ is regular, $b\times c= b\times d$, and since $b$ is regular, $c=d$, which proves that $(a\times b)$ is regular.
    
    \item Assume that $a\times b\in reg(\A)$.
        Let $(c,d)\in A^2$ such that $a\times c= a\times d$. Then, $a\times c\times b = a\times d\times b$ {\it  i.e.,} $(a\times b)\times c = (a\times b)\times d$, and by regularity of $a\times b$, $c=d$, which proves that $a$ is regular. Symmetrically, $b$ is also regular, which concludes the proof.\rev{\qedhere}
\end{itemize}
\end{proof}

We now establish that the product $+$ of the dioid $\D$ behaves as expected with respect to the order $\leq_D$ of $\D$. We also obtain stronger properties under a hypothesis of multiplicative regularity.

We remind the reader that we study the multiplicative regularity in the dioid $\D$ where $+$ is the {\em product} and not the sum.

\begin{property}\label{prop:multiplicatively regular order}

Let $\D$ be a dioid. Then, \rev{for all $(a,b,c,d)\in D^4$,
\begin{itemize}
\item  $a\leq_D b$ and $c\leq_D d \implies a+c\leq_D b+d$, and
\item if, in addition, $c$ or $d$ is a multiplicatively regular element of $D$ where $a<_D b$, then $a+c<_D b+d$.
\end{itemize}
}
\end{property}

\begin{proof}
We prove each case in turn.
\begin{itemize}

\item %

Let us prove first that $a+c\leq_D b+c$. Since $a\leq_D b$, there exists $b'\in D$ such that $a=\min(b,b')$. By distributivity of $+$ over $\inf$ (axiom of the semiring $(D,\min,+,\infty,0)$), we have $a+c=\min(b+c,b'+c)$, from which we can conclude $a+c\leq_D b+c$.

\rev{We have proven that $a\le_D b \implies a+c\le_D b+c$. Using this result with $c\le_D d$, we obtain similarly that $c+b\le_D d+b$. By commutativity of $+$, we have proven $a+c\le_D b+c\le _D b+d$, which gives the desired inequality $a+c\le_D b+d$ by the transitivity of $\le_D$.}

\item Assume that $c$ is a multiplicatively regular element of $D$ and $a<_D b$.
Since $b+c\leq_Db+d$ (first item), it is sufficient to prove that $a+c<_D b+c$. By the first item, we already have $a+c\leq_Db+c$.
Assume by contradiction that $a+c = b+c$. Then, since $c$ is multiplicatively regular, $a=b$ which contradicts the hypothesis $a<_Db$.
The same reasoning holds if $d$ is a multiplicatively regular element of $D$.\qedhere
\end{itemize}
\end{proof}

The result in the second item of Property \ref{prop:multiplicatively regular order} guarantees that the unwanted behaviour described in Example \ref{ex:motivation_regularity} does not occur for regular elements.

\subsection{$\Delta$-product of semirings}

In this section, we construct a new type of semiring, obtained by combining together a dioid $\D$ and a semiring $\A$, through an operation that we call a {\em $\Delta$-product}. The algorithmic applications of this new semiring will typically be to ``count the solutions of minimal weight'' of an {\sf NP} problem.

In order to improve the reading experience the proofs of the theorems of this section have been moved to~\ref{app:ProofsDeltaProduct}.

We first define which elements of $D\times A$ will belong to the new semiring. Any regular element $d$ of $D$ can be associated with any $a\in A$, but the non-regular elements can only be associated to $0_A$. The reason is that when the minimal weight of a set $\F=\F_1\Join\F_2$ is not regular, we can not guarantee that ${\argmin}_w(\F)={\argmin}_w(\F_1)\Join{\argmin}_w(\F_2)$. We thus give up on the operations in $A$ and instead output $0_A$ by default.

\begin{definition}
For $d\in D$ and $a\in A$, set  $$d\Delta a = \left\{\begin{array}{cc}
     (d,a) & \text{ if }d\in reg(\D) \\
     (d,0_A) & \text{ if }d\notin reg(\D) 
\end{array} \right\}.$$
\end{definition}

We are now ready to build a semiring that will be able to handle the combination of the $\D$-measure $w$ and the $\A$-measure $\mu$ that is intended to be applied to  ${\argmin}_w(\F)$.

\begin{reve}
\begin{definition}[Delta product]
The {\em delta product} $\mathcal{D}\Delta\mathcal{A}$ of a dioid $\mathcal{D} = (D, +_D, \min, \infty_D, 0_D)$ and a commutative semiring $\mathcal{A} = (A, +_A, \times_A, 0_A, 1_A)$ is defined as \[\mathcal{D}\Delta\mathcal{A}=(D\Delta A,\oplus,\otimes,(\infty_D,0_A),(0_D,1_A))\]
where

\begin{itemize}
 \item 
 $  D\Delta A := \{d\Delta a\mid (d,a)\in D\times A\}
   =
   \{ (d,a)\in D\times A \mid d\in reg(\D) \text{ or } a=0_A \};
 $  

\item $\oplus:  (D\Delta A)^2  \to  D\Delta A $ defined by 
$$(d_1,a_1)\oplus(d_2,a_2) =\big (\min(d_1,d_2),\left\{ \begin{array}{cc}
     a_1 & \text{ if } d_1 <_D d_2  \\
     a_2 & \text{ if } d_2 <_D d_1 \\ 
     a_1+_A a_2 & \text{ if } d_1=d_2
\end{array} \right\}\big);$$

\item $\otimes: (D\Delta A)^2 \to  D\Delta A$  given by $$(d_1,a_1)\oplus(d_2,a_2) = ((d_1+_Dd_2),(a_1\times_Aa_2)).$$

\end{itemize}
\end{definition}

We now prove that the resulting structure is indeed a commutative semiring.

\begin{restatable}{theorem}{DeltaProductSemiring}\label{thm:DeltaProductSemiring}
Let $\mathcal{D}$ be a dioid and $\mathcal{A}$ be a commutative semiring. Then the delta product $\mathcal{D}\Delta\mathcal{A}$ is a commutative semiring.
\end{restatable}
\end{reve}

Note that, by Property \ref{prop:minimum}, we can reformulate the definition of $\otimes: (D\Delta A)^2 \to  D\Delta A$ as

$$ (d_1,a_1)\oplus(d_2,a_2)=\left\{ \begin{array}{cc}
     (d_1,a_1) & \text{ if } d_1 <_D d_2  \\
     (d_2,a_2) & \text{ if } d_2 <_D d_1 \\
     (d_1,a_1+_A a_2) & \text{ if } d_1=d_2
\end{array} \right\}. $$
 We also observe that $\otimes$ is the Cartesian product of the two products $+_D$ and $\times_A$ of the semirings $D$ and $A$.
Intuitively, $\oplus$ is defined in order to mimic how we would determine the minimum, the number of minimal elements of a disjoint union of two sets $\F_1$ and $\F_2$, given their respective minima $d_1,d_2$ and the number of minimum elements $a_1,a_2$.
In order to study the associativity and commutativity of this operation, we let $\mathcal{D}_1=(D_1,\min_1,+_1,\infty_1,0_1)$ and $\mathcal{D}_2=(D_2,\min_2,+_2,\infty_2,0_2)$ be two commutative totally-ordered idempotent dioids.

It is now natural to study the algebraic properties of this newly defined $\Delta$-product. We prove in the following that the $\Delta$-product is associative \rev{and that the resulting order $\leq_{\mathcal{D}_1 \Delta \mathcal{D}_2}$ is the {\em lexicographical order} $\leq_{\mathcal{D}_1 \Delta \mathcal{D}_2}$ defined as $(d_1, d_2) \leq_{\mathcal{D}_1 \Delta \mathcal{D}_2} (d_3, d_4)$ if and only if (1) $d_1 <_{\mathcal{D}_1} d_3$, or (2) $d_1 \leq_{\mathcal{D}_1} d_3$ and $d_2 \leq_{\mathcal{D}_2} d_4$.}

\begin{restatable}{theorem}{LexOrder}\label{thm:Delta lex order}
\rev{Let $\mathcal{D}_1$ and $\mathcal{D}_2$ be two commutative, totally ordered, idempotent dioids.}
The semiring $\mathcal{D}_1\Delta\mathcal{D}_2$ is also a commutative totally ordered idempotent dioid, and the associated order $\leq_{\mathcal{D}_1\Delta\mathcal{D}_2}$ is the lexicographical order $\leq_{\mathcal{D}_1 lex \mathcal{D}_2}$.

\end{restatable}

\begin{example}

\rev{For instance, the set $\Rbar_{\min} \Delta \Rbar_{\min}$ is defined as $\R \times (\R \cup \{\infty\}) \cup \{(\infty, \infty)\}$. Minimization of weights in this dioid is performed lexicographically: we first minimize the first coordinate, and in case of a tie, we then minimize the second coordinate. If the only possible value for the first coordinate is $\infty$, then the second coordinate becomes irrelevant.}

\end{example}

\rev{The $\Delta$-product is in general non-commutative, and the order of the dioid within the $\Delta$-product encodes the coordiante whose minimization is prioritized.}

\begin{example}
\rev{
The dioid $\mathbf{2} \Delta \overline{\mathbb{R}}_{\min}$ is defined as $\left(\{\top\} \times \overline{\mathbb{R}}_{\min}\right) \cup \{(\bot, \infty)\}$, where $(\bot, \infty)$ is the greatest element. When minimizing weights in this dioid, the priority is to select elements whose first coordinate is in $\{\top\}$; only among those is the second coordinate minimized.}

\rev{
Conversely, in the dioid $\overline{\mathbb{R}}_{\min} \Delta \mathbf{2} = \left(\mathbb{R} \times \{\bot, \top\}\right) \cup \{(\infty, \bot)\}$, the minimization process first seeks the smallest value $m$ of the first coordinate. Among all such elements with first coordinate $m$, preference is then given to those whose second coordinate is $\top$ rather than $\bot$.}

\end{example}

Studying the associativity of the $\Delta$-product requires to compare $(\D_1\Delta\D_2)\Delta\A$ with $\D_1\Delta(\D_2\Delta\A)$. Even though the definition of $\D_1\Delta(\D_2\Delta\A)$ does not raise any issue, considering $(\D_1\Delta\D_2)\Delta\A$ requires that we justify that $(\D_1\Delta\D_2)$ is a commutative totally ordered idempotent dioid, which is ensured by Theorem \ref{thm:Delta lex order}.

\begin{restatable}{theorem}{Associativity}\label{thm:delta_product_semiring_associative}
\rev{Let $\mathcal{D}_1$ and $\mathcal{D}_2$ be two commutative, totally ordered, idempotent dioids.}
The $\Delta$-product of semirings is associative. More precisely, up to identifying $(D_1\times D_2)\times A$ with $D_1\times (D_2\times A)$, $(\mathcal{D}_1\Delta \mathcal{D}_2)\Delta \mathcal{A} = \mathcal{D}_1\Delta( \mathcal{D}_2\Delta \mathcal{A})$.

\end{restatable}

This is essentially due to the fact that the lexicographic order is associative.
Note that the $\Delta$-product of semirings is not commutative, even up to isomorphism, since the lexicographic order depends on the ordering of its coordinates.

\begin{example}
A surprising application of $\Delta$-product is to build the dioids of the Kleene algebras of {\em multi-valued logic} \cite{berman2001stipulations}.
The trilean dioid $( \mathbf{3} , \lor , \land , \bot, \top)$ (with $\mathbf{3}=\{\bot,\top,?\}$) defined by the table:

\begin{multicols}{2}

\begin{tabular}{c|c|c|c|}
    $\vee$ & $\top$ & $?$ & $\bot$  \\
    \hline
     $\top$ & $\top$ & $\top$ & $\top$ \\
     $?$ & $\top$ & $?$ & $?$ \\
     $\bot$ & $\top$ & $?$ & $\bot$
\end{tabular}

\columnbreak

\begin{tabular}{c|c|c|c|}
    $\wedge$ & $\top$ & $?$ & $\bot$  \\
    \hline
     $\top$ & $\top$ & $?$ & $\bot$ \\
     $?$ & $?$ & $?$ & $\bot$ \\
     $\bot$ & $\bot$ & $\bot$ & $\bot$
\end{tabular}

\end{multicols}

is isomorphic to the dioid $\B\Delta\B$, through the isomorphism $\Phi$ defined by $\Phi(\bot,\bot)=\bot, \Phi(\top,\top)=\top, \Phi(\top,\bot)=?$. Recall that $(\bot,\top)\notin \2\Delta\2$, since $\bot$ is not regular.

\rev{
More generally, for all $n\ge 2$ the dioid of the Kleene-algebra of the $n$-valued logic $([n],\max,\min,1,n)$ is isomorphic to \rev{$\B^n \Delta := \underbrace{\B\Delta\B\Delta\dots\Delta\B}_{n-1 \text{ occurrences of }\B}$}, through the isomorphism
$\Phi: \2^n\Delta\to  [n]$ defined by $\Phi\underbrace{(\top,\dots,\top,\bot,\dots,\bot)}_{i \text{ occurrences of }\top} = i+1. $}
\end{example}

However, the algorithmic applications of the dioids of $n$-valued logic are limited by the fact that their only multiplicatively regular element is $(\top,\dots,\top)$. Therefore studying the measures of the dioid $[n]\Delta\A$ with $\A$ will not be especially relevant, since $[n]\Delta A = (\{n\}\times A) \uplus \{ (i,0_A) \mid i\in [n-1] \}$.

\subsection{$\Delta$-product of measures}

The main purpose of $\Delta$-product of semirings is their ability to shelter $\Delta$-product of measures. More precisely, if $w$ is a $\D$-measure (seen as a weight function) and $\mu$ is an $\A$-measure (for instance, the $\N$-measure of cardinality $\#:\F\mapsto |\F|$), $(w\Delta\mu)(\F)$ is the image by $\mu$ of the set of minimal functions (with respect to $w$) of $\F$.

\begin{restatable}{theorem}{DeltaProductMeasure}\label{thm:Delta-composition of measures}

The function $w\Delta\mu$ that associates every $\F\subseteq T^{S'}$ with $S'\subseteq S$:

\begin{eqnarray*}
(w\Delta\mu)(\F)&:=& \left\{ \begin{array}{cc} ( w(\F) , \mu({\argmin}_w(\F))) & \text{ if }w(\F)\in reg(D) \\ (w(\F),0_A) & \text{ if }w(\F)\notin reg(D)  \end{array}\right. \\
&=&
(w(\F))\Delta(\mu({\argmin}_w(\F)))
\end{eqnarray*}
with ${\argmin}_w(\F):= \{f\in\F, w(\{f\}) = w(\F) \}$ is a $(\mathcal{D}\Delta\mathcal{A})$-measure.

\end{restatable}

The main interest of the $\Delta$-compositions of semiring and measures is then that it makes new extensions of {\sf NP} computational problems fall under the various semiring formalisms. This is true for the {\sc Semiring}-$\mathbf{\Pi}$ problems (for $\mathbf{\Pi}\in$ {\sf NP}) introduced in Section~\ref{sec:semirings} as well as for the related problems in the {\sc Sum-Product} family (introduced in Section~\ref{sec:sum_product}).

\subsection{Computational extensions}\label{sec:count_cost_application}

Let $\mathbf{\Pi}$ be a problem in {\sf NP}, and let $x$ be an instance of $\mathbf{\Pi}$. As usual we assume, without loss of generality, that the set of solutions of $x$ are functions from a set $S$ to a set $T$ (where $S$ and $T$ have polynomial size with respect to $x$). \rev{Also, recall that since $\mathbf{\Pi}$ is an NP-problem a reasonable notion of a solution always exists (a certificate) but we stress that we actually do not need any particular assumptions, and that one for specific problems is always free to choose a suitable problem specific notion.}

By Lemma \ref{lem:expression_computation}, for any computable function $g$, if {\sc Semiring-}$\mathbf{\Pi}$ is solvable in time $O(g(n))$, then so is the {\sc \#Cost-}$\mathbf{\Pi}$ problem defined as:  

\hfill

$\underline{\#\text{\sc Cost-}{\mathbf{\Pi}:}}$

\textbf{Input:} An instance $x$ of $\mathbf{\Pi}$, and a cost matrix $C\in (\Rbar)^{S\times T}$.

\textbf{Output:} How many solutions $f$ of $x$ of minimal cost $\sum\limits_{s\in S} C[s,f(s)]$ are there, and what is this minimal cost?

\hfill

Solving $\#\text{\sc Cost-}{\Pi}$ is indeed equivalent to computing the image of the set of solutions $SOL$ of the instance by the $(\Rbar_{min}\Delta\N)$-measure $(C\Delta \#)$ in the case where there exists at least one solution of cost $\neq\infty$. If all solutions have infinite cost, $(C\Delta\#)(SOL)$ will output $(\infty,0)$, and then all solutions have minimal cost: we can count them by computing $\#(SOL)$.

Note that in practice, to avoid having to compute both $(C\Delta\#)(SOL)$ and $\#(SOL)$ by using the algorithm described in Lemma \ref{lem:expression_computation} twice, it is possible to use the cartesian product 
$$(C\Delta\#)\times\# : \F \to (C\Delta\#)(\F)\times\#(\F) $$ 
of the measures $(C\Delta\#)$ and $\#$.  The measure $(C\Delta\#)\times\#$ is indeed a measure of the cartesian product of semirings $(\Rbar_{min}\Delta\N)\times\N$.

In particular we observe that the well-known SAT problem of counting solutions of minimal cardinality~\cite{hermann2009complexity} is a particular case of {\sc \#Cost-SAT}. \rev{Recall that an instance of {\sc Sat} is given by a Boolean formula $\varphi(X)$ is conjunctive normal form and the task is to find a function $\sigma \colon X \mapsto \{\bot, \top\}$ which satisfies every clause.}

\hfill

\underline{\sc \#Min-Card-Sat:}

\textbf{Input:} An instance $\varphi(X)$ of {\sc Sat} on a set of variables $X$.

\textbf{Output:} The number of models $\sigma:X\mapsto\{\bot,\top\}$ of $\varphi$ of minimal cardinality (i.e., where $|\sigma^{-1}(\{\top\})|$ is minimal).

\hfill

Similarly, the problem

\hfill

\underline{\sc \#Min-Lex-Sat:}

\textbf{Input:} An instance $\varphi(X)$ of {\sc Sat} on a set of variable $X$, and $(x_1,\dots,x_{\ell})\in X^{\ell}$ (with $\ell\ge 0$).

\textbf{Output:} The number of models $\sigma:X\mapsto\{\bot,\top\}$ of $\varphi$ where $(\sigma(x_1),\dots,\sigma(x_{\ell}))$ is lexicographically minimal with respect to the order $\bot \le \top$.\medskip

\noindent from Hermann and Pichler~\cite{hermann2008counting} and even the more general problem

\hfill

\underline{\sc \#Min-Weight-Sat:}

\textbf{Input:} An instance $\varphi(X)$ of {\sc Sat} on a set of variables $X$, and a weight function $w:X\mapsto\N$.

\textbf{Output:} The number of models $\sigma:X\to\{\bot,\top\}$ of $\varphi$ of minimal weight, {\it  i.e.,} where $$\sum\limits_{x\in \sigma^{-1}(\{\top\})} w(x)\quad \text{ is minimal.}$$\\

\noindent from Hermann and Pichler~\cite{hermann2009complexity} are also particular cases of {\sc \#Cost-Sat}. Indeed, if we for all $w:X\mapsto\N$ define the matrix $W\in(\Rbar)^{X\times\{\bot,\top\}}$ by $W[x,\top]=w(x)$ and $W[x,\bot]=0$, the {\sc \#Min-Weight-Sat} problem consists in applying the $(\R_{min}\Delta\N)$-measure $\mu_W\Delta\#$ to the set of solutions of an instance of {\sc Sat}. Also, the {\sc \#Min-Card-Sat} problem is the particular case where $w$ is constantly $1$, and the {\sc \#Min-Lex-Sat} problem is the particular case where
$w\colon  X  \to  \N$ is given by  $$w(y) = \left\{ \begin{array}{cc}
     2^{\ell -i} & \text{ if }y=x_i \text{ with }i\in [\ell] \\
     0 & \text{ if }y\notin \{x_1,\dots,x_{\ell}\}
\end{array} \right. $$

\section{Semiring-Connected-Dominating-Set and Cliquewidth}\label{sec:dom-set}

\begin{figure}[!ht]
\begin{center}
\scalebox{0.8}{
\begin{tikzpicture}

\tikzstyle{vertex}=[circle, draw, inner sep=1pt, minimum width=6pt]
\tikzstyle{sq}=[rectangle, draw, inner sep=1pt, minimum width=6pt]

\node[sq,color=blue] (NPproblem) at (0,5) {$\begin{matrix}\text{Parameterized problem } \\ (\mathbf{\Pi},\lambda)\text{ with }\mathbf{\Pi}\in{\sf NP}\end{matrix}$};

\node[sq,color=blue] (SemiringNP) at (0,2.5) {$\begin{matrix}{\sf FPT}\text{ algorithm that solves} \\ \textsc{Semiring-}\mathbf{\Pi}\end{matrix}$}
    edge[<-,color=blue] (NPproblem);

\node[color=blue] () at (-2.5,3.75) {$\begin{matrix} \text{See Section }\ref{sec:dom-set}\text{ for } (\mathbf{\Pi},\lambda)= \\ (\textsc{Connected-Dom-Set},\mathbf{cw}) \end{matrix}$};

\node () at (2.5,3.75) {$\begin{matrix} \text{See Section }\ref{sec:Csp}\text{ for } (\mathbf{\Pi},\lambda)= \\ (\textsc{CSP},\mathbf{tw}) \end{matrix}$};

\node () at (-7.5,0) {$\begin{matrix} \text{See Section }\ref{sec:semirings} \\ \text{for the} \\\text{construction} \\ \text{of measures}\end{matrix}$};

\node[sq] (B) at (-4.5,0) {$\B$-measures}
    edge (SemiringNP);

\node[sq] (N) at (-1.5,0) {$\N$-measures}
    edge (SemiringNP);

\node[sq] (Rmin) at (1.5,0) {$\R_{\min}$-measures}
    edge (SemiringNP);

\node[sq] (Delta) at (5,0) {$(\R_{\min}\Delta\N)$-measures}
    edge (SemiringNP);

\node[] () at (5,1) {$\begin{matrix} \text{See Section }\ref{sec:Delta-product} \\ \text{for the} \\ \Delta\text{-product}\end{matrix}$};

\node (midB) at (-2.58,1.2) {};
\node (midN) at (-0.96,1.2) {};
\node (midRmin) at (0.96,1.2) {};
\node (midDelta) at (2.85,1.2) {};

\node[sq] (Decision) at (-4.5,-2.5) {$\begin{matrix} \textsf{FPT}\text{ algorithm} \\ \text{that solves }\textsc{List-}\mathbf{\Pi}  \end{matrix}$}
    edge[<-,bend right] (midB);

\node[sq] (Counting) at (-1.5,-2.5) {$\begin{matrix} \textsf{FPT}\text{ algorithm} \\ \text{that solves \#}\mathbf{\Pi}  \end{matrix}$}
    edge[<-,bend right=50] (midN);

\node[sq] (Cost) at (1.5,-2.5) {$\begin{matrix} \textsf{FPT}\text{ algorithm} \\ \text{that solves }\textsc{Cost-}\mathbf{\Pi}  \end{matrix}$}
    edge[<-,bend left=50] (midRmin);

\node[sq] (CountCost) at (5,-2.5) {$\begin{matrix} \textsf{FPT}\text{ algorithm} \\ \text{that solves }\#\textsc{Cost-}\mathbf{\Pi}  \end{matrix}$}
    edge[<-,bend left] (midDelta);

\node () at (-7.5,-2.5) {$\begin{matrix}\text{See Lemma } \ref{lem:expression_computation} \\ \text{and Example }\ref{example:measures} \end{matrix}$};

\node at (5,-1.5) {See Section \ref{sec:count_cost_application}};

\end{tikzpicture}
}
\end{center}
\caption{A general overview of our contributions where the contribution in this section is colored in blue.}
\label{fig:plan_of_the_paper_dom_set}
\end{figure}
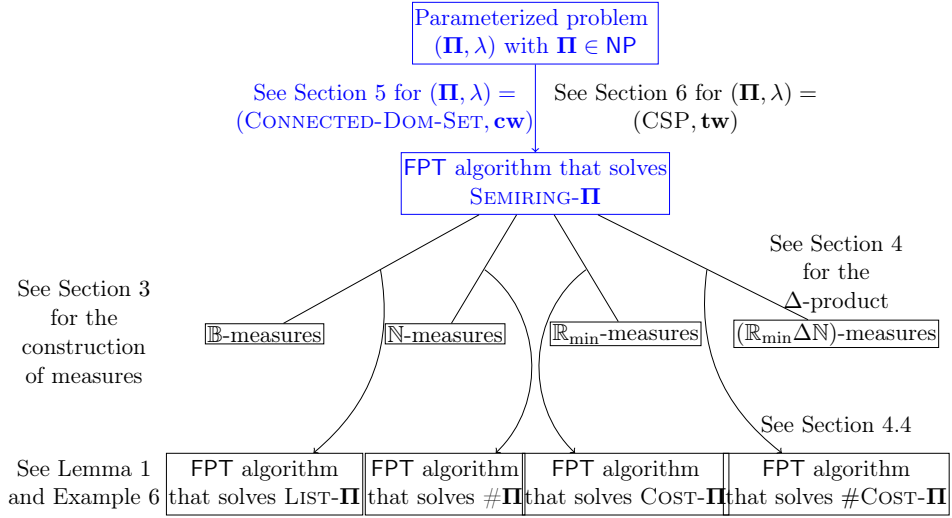

Recall that the framework is built upon the existence of an algorithm for {\sc Semiring-}$\mathbf{\Pi}$ (with $\mathbf{\Pi}\in$ {\sf NP}), which via Lemma~\ref{lem:expression_computation} can then be exploited to solve the usual semiring extensions $\#\mathbf{\Pi}$, {\sc List-}$\mathbf{\Pi}$, {\sc Cost-}$\mathbf{\Pi}$. The $\Delta$-product presented in Section \ref{sec:Delta-product} also makes it possible to subsume the problems of the form \#{\sc Cost-}$\mathbf{\Pi}$ in the semiring formalism. Thus, an algorithm for {\sc Semiring-$\mathbf{\Pi}$} is highly desirable since it can be used to solve many different (combinations of) problem extensions. To exemplify that this is indeed feasible we consider two well-known problems: {\sc Connected-Dominating-Set}, and {\sc Csp}. For both these problems, we construct {\sf FPT} algorithms that solve the semiring extended problems, and,  as a corollary, may derive upper bounds on all of the aforementioned problem extensions. In particular, we are able to solve both \#{\sc Cost-Connected-Dominating-Set} and \#{\sc Cost-CSP} by Section~\ref{sec:Delta-product}. The contribution of this section to the semiring formalism presented in this paper is summarized in Figure \ref{fig:plan_of_the_paper_dom_set}.

Both of these algorithms (in Section \ref{sec:dom-set} and \ref{sec:Csp}, respectively)  follow a similar idea since they are both based on bounded width properties of graphs. We utilize a ``method of construction of the graph'', such as a $k$-expression in Section~\ref{sec:dom-set}, or a tree decomposition in Section~ \ref{sec:Csp}. We will explain in Section \ref{sec:treewidth} how tree decompositions can be seen as a method of construction of a graph.

\begin{remark}\label{rem:trace}

Our algorithms only require that we ``compress'' each candidate solution by considering its so-called {\em trace}. The definition of the trace depends strongly on the problem and the parameter considered, but the trace must always verify these three properties:

\begin{enumerate}
    
    \item there must be only an {\sf FPT} number of possible traces,

    \item it is possible to decide whether a candidate function is a solution knowing only its trace, and

    \item the new trace of a candidate solution after performing a step of the ``method of construction of the graph'', depends only on its trace before this step.

\end{enumerate}
\rev{The technique of using traces to compress the solution space is a common approach in parameterized algorithms, particularly in dynamic programming on graphs with small treewidth or cliquewidth. For a comprehensive treatment see, e.g., Cygan et al.~\cite{cygan2015parameterized}.}
\end{remark}

Then, our algorithm maintains the invariant that for every possible trace $\tau$, we know a join/union expression $E_{\tau}$ of the set of all candidates that have the trace $\tau$. Note that if $\tau$ and $\tau'$ are two different traces, $E_{\tau}$ and $E_{\tau'}$ represent disjoint sets, and thus, the join/union expression $E_{\tau}\uplus E_{\tau'}$ is always legal.
By property 3., this invariant can be preserved through every step of the construction of the graph, and is therefore eventually true for the whole graph. Then, by property 2., this enables us to recover a join/union expression of the set of solutions via disjoint union, thereby answering the \textsc{Semiring} extension of our problem. Due to property 1., all of these operations are done in {\sf FPT} time.

Also, a variant of this paradigm would be to only encode the set of {\em solutions} with a given trace, instead of all {\em candidate solutions}, then  property 2. is not necessary, but property 3. must apply to {\em sets of solutions} instead of {\em candidate solutions}. This is what we will do in Section \ref{sec:Csp} to solve {\sc Csp}.

\subsection{FPT Algorithm}

We propose an algorithm that solves {\sc Semiring-Connected-Dominating-Set} and is {\sf FPT} when parameterized by the clique-width of the input graph. Thus, we manage to solve all common problem extensions (minimal cost, list, and counting) as well as problem extensions via the $\Delta$-product from Section~\ref{sec:Delta-product}. Let us remark that the basic problem {\sc Connected-Dominating-Set} is known to be solvable in $O^*(2^{(\omega + 4)k})$ time \cite{bergougnoux2019fast}, where $\omega  <2.37188$ is the matrix multiplication exponent \cite{duan2022faster}, and, additionally, in $O^*(5^k)$ time via a randomized (Monte-Carlo) algorithm \cite{hegerfeld2023tight}. However, no {\sf FPT} algorithm is known for the counting extension (which we vastly generalize). More formally, we study the problem:

\hfill

\underline{\sc Semiring-Connected-Dominating-Set:}

\textbf{Input:} A graph $G$.

\textbf{Output:} A join/union expression for the set of indicator functions of connected dominating-sets of $G$.

\hfill

and provide an \textsf{FPT} algorithm parameterized by the clique-width of the input graph.

\begin{theorem}\label{thm:semiring-connected-dominating-set-FPT}

The problem {\sc Semiring-Connected-Dominating-Set} is {\sf FPT} when parameterized by the cliquewidth of the input graph (assuming that a $\mathbf{cw}(G)$-expression of the input graph $G$ is given).

\end{theorem}

\begin{proof}

We give an algorithm that solves {\sc Semiring-Connected-Dominating-Set} in time $O^*( (3^{2^{cw(G)}}\times 2^{cw(G)})^2 ) = O^*(3^{2^{cw(G)+1}}\times 4^{cw(G)})$.
To do so, take $k\geq 1$, a $k$-expression $\varphi$ of a graph $G$, and $S\subseteq V_G$. The $k$-expression will play the role of the ``method of construction of the graph'' as outlined in Remark \ref{rem:trace}. 

In order to define the ``trace'' of Remark \ref{rem:trace}, we will use three different labels, called $0$, $1$ and $\twobar$. They should be interpreted as respectively ``zero'', ``one'' and ``at least two''. We let $\threebar:=\{0,1,\twobar\}$, and we define the operation $+$ in $\threebar$ in Table \ref{tab:add_in_treebar}  agreeing with this interpretation of $0$, $1$ and $\twobar$.

\begin{table} 
\parbox{\linewidth}{
    \centering
    \begin{tabular}{c|ccc}
         + & 0 & 1 & $\twobar$  \\\hline
         0 & 0 & 1 & $\twobar$ \\
         1 & 1 & $\twobar$ & $\twobar$ \\
         $\twobar$ & $\twobar$ & $\twobar$ & $\twobar$ 
    \end{tabular}
    \caption{Addition in $\threebar$.}
    \label{tab:add_in_treebar}
}
\end{table}

We define the following operations for all $k$-expression $\varphi$:

\begin{itemize}
    
    \item \rev{The {\em signature} of $S$ in $\varphi$ is the function
    $\sigma_{\varphi}(S)\colon 2^{[k]} \to  \threebar$ is defined by  $$\sigma_{\varphi}(S)(C)= \left\{ \begin{array}{cc}
     0 & \text{ if there exists no }t\in [q]\text{ with } C=l_G(C_t)  \\
     1 & \text{ if there exists is a unique }t\in [q]\text{ with } C=l_G(C_t)  \\
     \twobar &  \text{ if there exists several }t\in [q]\text{ with } C=l_G(C_t)
\end{array} \right., $$
    where $S=C_1\uplus\dots\uplus C_q$ with $q\geq 0$ is the partition of $S$ into connected components of $G[S]$. The label $\twobar$ should be interpreted as ``at least $2$''.}

    \rev{The quantity $(\sigma_{\varphi}(S))(C)\in\threebar$ with $C\in 2^{[k]}$ gives the answer to the question: ``How many connected components $C_t\subseteq S$ of $G[S]$ are there such that the set of labels (i.e. the labels in $[k]$ of the $k$-expression $\varphi$) that appear in $C_t$ is exactly $C$?''. The value of $(\sigma_{\varphi}(S))(C)$ is either: $0$ if the answer is ``$0$ connected components'', $1$ if the answer is ``exactly $1$ connected component'', or $\twobar$ if the answer is ``at least $2$ different connected components''.
    Note that $G[S]$ is connected if and only if $\|\sigma_{\varphi}(S)\|=1$, with 
    $$\|\sigma_{\varphi}(S)\|=\sum\limits_{C\in 2^{[k]}} (\sigma_{\varphi}(S))(C)$$ being the number of connected components of $G[S]$ (in $\threebar$).}
    
    \item \rev{The {\em domination} of $S$ in $\varphi$, denoted by $dom_{\varphi}(S) \in 2^{[k]}$ is the set $dom_{\varphi}(S) = \{d\in [k] \mid V_G^d \subseteq N_G[S]\}$ \rev{(recall that $V_G^d$ is the set of vertices with label $d$)}. Note that $S$ is a dominating set of $G$ if and only if $dom_{\varphi}(S)=[k]$.} %
    
    \item \rev{Finally, the {\em trace} of $S$ in $\varphi$ is $Tr_{\varphi}(S)=(\sigma_{\varphi}(S),dom_{\varphi}(S)) \in \threebar^{2^{[k]}}\times 2^{[k]}$.}
    
\end{itemize}

Also, for all $(\sigma,D)\in \threebar^{2^{[k]}}\times 2^{[k]}$, $Tr_{\varphi}^{-1}(\sigma,D) \subseteq V_G$ denotes the inverse image of $\{(\sigma,D)\}$ by the function $Tr_{\varphi}$, {\it {\it  i.e.,}} $$Tr_{\varphi}^{-1}(\sigma,D):= \{ S\subseteq V_G \mid Tr_{\varphi}(S) = (\sigma,D)\}.$$ 
Note that here $\sigma_{\varphi}$ is  considered the function $\sigma_{\varphi}\colon  2^{V_G}  \to  \threebar^{2^{[k]}}$ defined by  $$\sigma_{\varphi}(S) =\sigma_{\varphi}(S)\in\threebar^{2^{[k]}}, $$ from which we consider the inverse images.

The idea of the algorithm is then to compute the values of $Tr_{\varphi}^{-1}(\sigma,D)$ for all $(\sigma,D)\in \threebar^{2^{[k]}}\times 2^{[k]}$, by induction over the structure of $\varphi$. What will be especially useful is that:

\begin{enumerate}
    
    \item $\{ Tr_{\varphi}^{-1}(\sigma,D) \mid (\sigma,D)\in \threebar^{2^{[k]}}\times 2^{[k]} \}$ forms a partition of $V_G$, which makes it possible to use $\uplus$. This is a partition into a {\sf FPT} number of subsets ($3^{2^k}\times 2^k$) with respects to $k=\mathbf{cw}(G)$,
    
    \item the set of connected dominating sets of $G$ is exactly $$\underset{\|\sigma\|=1}{\uplus} Tr_{\varphi}^{-1}(\sigma,[k]),\quad \text{with $\|\sigma\|=\sum\limits_{C\in 2^{[k]}} \sigma(C)$, and}$$
    
    \item for $S\subseteq V_G$, the value of $Tr_{\varphi}(S)$ depends only on the $Tr_{\varphi'}(S)$ with $\varphi$ being of the form $\varphi=\eta_{i,j}(\varphi')$ or $\varphi=\rho_{i\rightarrow j}(\varphi')$. A similar remark holds for $\varphi=\varphi_1\oplus\varphi_2$. Also, $Tr_{\bullet_i(u)}^{-1}(\sigma,D)$ (for $(\sigma,D)\in \threebar^{2^{[k]}}\times 2^{[k]}$) are easy to compute.
    
\end{enumerate}

We now focus on justifying this third remark.
Since the join/union expressions are made to express sets of functions and not subsets of vertices, we will give a join/union expression $E_{\varphi}(\sigma,D)$ of the sets of indicator functions of the subsets of vertices in $Tr_{\varphi}^{-1}(\sigma,D)$.
\rev{We inductively compute the values of the $Tr_{\varphi}^{-1}(\sigma,D)$ for $(\sigma,D)\in \threebar^{2^{[k]}}\times 2^{[k]}$ over a $k$-expression $\varphi$ as follows.}

\begin{itemize}
    
    \item \rev{\textbf{Single vertex:} For $i\in [k]$, $Tr_{\bullet_i(u)}(\{u\})= (\sigma_{\bullet_i(u)}(\{u\}),[k])$ and $Tr_{\bullet_i(u)}(\emptyset) = (\sigma_{\bullet_i(u)}(\emptyset),[k]\setminus\{i\})$, with
    $\sigma_{\bullet_i(u)}(\{u\})\colon 2^{[k]} \to  \threebar$ 
    defined by 
    $$\sigma_{\bullet_i(u)}(\{u\})(C)= \left\{ \begin{array}{cc}
         1 & \text{ if }C=\{i\}  \\
         0 & \text{ otherwise}
    \end{array} \right. $$}

    \rev{and
    $\sigma_{\bullet_i(u)}(\emptyset)
\colon 2^{[k]}  \to  \threebar$ defined by 
$$\sigma_{\bullet_i(u)}(\emptyset)(C)= 0.$$
 Since $\{u\}$ and $\emptyset$ are the only two subsets of $\{u\}$ ($\{u\}$ is the set of vertices of the graph expressed by the $k$-expression $\bullet_i(u)$), we have handled all cases. In other words, we have proven for all $(\sigma,D)\in \threebar^{2^{[k]}}\times 2^{[k]}$:}

    \rev{$Tr^{-1}_{\bullet_i(u)}(\sigma,D) = \left\{ \begin{array}{cc}
         \{\{u\}\} & \text{ if }(\sigma,D)=(\sigma_{\bullet_i(u)}(\{u\}),[k])\\
         \{\emptyset\} & \text{ if }(\sigma,D)=(\sigma_{\bullet_i(u)}(\emptyset),[k]\setminus\{i\}) \\
         \emptyset & \text{ otherwise.}
    \end{array} \right.$}

    \rev{Note that the indicator function $\mathbf{1}_{\{u\}}^{\{u\}}$ (respectively $\mathbf{1}_{\emptyset}^{\{u\}}$) of the set $\{u\}$ (respectively $\emptyset$) over the domain $\{u\}$ is $(u\mapsto 1)$ (respectively $(u\mapsto 0)$). Thus, a join/union expression of the set of indicator functions of the subsets of $Tr^{-1}_{\bullet_i(u)}(\sigma,D)$ would be:}

    \rev{$E_{\bullet_i(u)}(\sigma,D) := \left\{ \begin{array}{cc}
         (u\mapsto 1) & \text{ if }(\sigma,D)=(\sigma_{\bullet_i(u)}(\{u\}),[k])\\
         (u\mapsto 0) & \text{ if }(\sigma,D)=(\sigma_{\bullet_i(u)}(\emptyset),[k]\setminus\{i\}) \\
         \emptyset & \text{ otherwise.}
    \end{array} \right.$}

    \item \textbf{Disjoint union:} Let $\varphi$ be a $k$-expression of the form $\varphi=\varphi_1\oplus\varphi_2$. Let $G=[\varphi]$, $G_1=[\varphi_1]$ and $G_2=[\varphi_2]$. Then, $V_{G_1}$ and $V_{G_2}$ are disjoint.
    Let $S\subseteq V_G$, decompose $S=S_1\uplus S_2$, with $S_1\subseteq V_{G_1}$ and $S_2\subseteq V_{G_2}$ or equivalently, $\mathbf{1}_S^{V_G}=\mathbf{1}_{S_1}^{V_{G_1}} \Join \mathbf{1}_{S_2}^{V_{G_2}}$. %
    Then, denoting $Tr_{\varphi_1}(S_1) = (\sigma_1,D_1)$ and $Tr_{\varphi_2}(S_2) = (\sigma_2,D_2)$, we have $Tr_{\varphi}(S) = (\sigma_1+\sigma_2,D_1\cap D_2)$
    with
    $\sigma_1+\sigma_2\colon 2^{[k]} \to \threebar$ given by  $$(\sigma_1+\sigma_2)(C)= \sigma_1(C)+\sigma_2(C).$$
    Indeed, the connected components of $G$ are exactly the connected components of $G_1$ and the connected components of $G_2$.
    Also, a label is dominated in $G$ if and only if it is dominated both in $G_1$ and in $G_2$.
    This proves that:
    
    \[Tr_{\varphi}^{-1}(\sigma,D) = \underset{ \begin{matrix} \sigma_1+ \sigma_2 = \sigma \\ D_1\cap D_2=D \end{matrix} }{\biguplus} \{ S_1\uplus S_2 \mid (S_1,S_2)\in Tr_{\varphi_1}^{-1}(\sigma_1,D_1) \times Tr_{\varphi_2}^{-1}(\sigma_2,D_2) \}.\]
 for all $(\sigma,D)\in \threebar^{2^{[k]}}\times 2^{[k]}$.    It follows a similar relation for the indicator functions:

     \[E_{\varphi}(\sigma,D) := \underset{ \begin{matrix} \sigma_1+ \sigma_2 = \sigma \\ D_1\cap D_2=D \end{matrix} }{\biguplus} E_{\varphi_1}(\sigma_1,D_1) \Join E_{\varphi_2}(\sigma_2,D_2).\]

    \item \textbf{Relabeling:} \rev{Let $\varphi$ be a $k$-expression of the form $\varphi=\rho_{i\rightarrow j}(\varphi')$, $G=[\varphi]$ and $G'=[\varphi']$. 
    Let $S\subseteq V_G$ and $(\sigma',D')=Tr_{\varphi'}(S)$.
    Note that the connected components of $G[S]$ are exactly the connected components of $G'[S]$. Moreover, if the set of labels of a connected component $C_t$ of $G'[S]$ is exactly $C'\in 2^{[k]}$, then the set of labels of $C_t$ in $G[S]$ will be exactly $\rho_{i\rightarrow j}(C')$ with:}

    \rev{$$\rho_{i\rightarrow j}(C')= \left\{ \begin{array}{cc}
         C' & \text{ if }i\notin C'  \\
         (C' \cup \{j\}) \setminus \{i\} & \text{ if }i\in C'.
    \end{array} \right.$$}

    \rev{It follows that the signature $\sigma_{\varphi}(S)$ is exactly $\sigma_{\varphi}(S)=\rho_{i\rightarrow j}(\sigma')$ with}
\begin{itemize}

    \item[(i)] \rev{$\rho_{i\rightarrow j}(\sigma')\colon  2^{[k]}  \to  \threebar$ 
    given by $$\rho_{i\rightarrow j}(\sigma')(C)= \sum\limits_{C'\in 2^{[k]}, \rho_{i\rightarrow j}(C')=C} \sigma'(C'), $$
    with the convention that the empty sum equals $0$.}
\end{itemize}

Also, as the domination of $S$ in $G'$ is $D'$, the domination of $S$ in $G$ is exactly $dom_{\varphi}(S)=\rho_{i\rightarrow j}(D')$ with

   \begin{itemize}
    \item[(ii)]  \rev{$$\rho_{i\rightarrow j}(D') = \left\{\begin{array}{cc}
        D' & \text{if }\{i,j\}\subseteq D' \\
        D'\cup\{i\}\setminus\{j\} & \text{otherwise.} 
    \end{array} \right.$$}
\end{itemize}

    This proves that:
    \[Tr_{\varphi}^{-1}(\sigma,D) = \underset{(\rho_{i\rightarrow j}(\sigma'),\rho_{i\rightarrow j}(D'))=(\sigma,D)}{\biguplus}Tr_{\varphi'}^{-1}(\sigma',D')\]
    for all $(\sigma,D)\in\threebar^{2^{[k]}}\times 2^{[k]}$.
    The same relation holds for the indicator function, leading to the join/union expression:

    \[E_{\varphi}(\sigma,D) = \underset{(\rho_{i\rightarrow j}(\sigma'),\rho_{i\rightarrow j}(D'))=(\sigma,D)}{\biguplus}E_{\varphi'}(\sigma',D').\]

    \item \textbf{Edge creation:} Let $\varphi$ be a $k$-expression of the form $\varphi=\eta_{i,j}(\varphi')$, $G=[\varphi]$ and $G'=[\varphi']$.
    Let $S\subseteq V_G$ and $(\sigma',D')=Tr_{\varphi'}(S)$.

    First, recall that all the vertices have the same label in $G$ and in $G'$, {\it  i.e.,} $l_G=l_{G'}$.
        Note that if the connected components of $G'[S]$ are $C_1,\dots,C_q$ with $q\ge 1$, and if the connected components of $G'[S]$ that contains either a $i$-vertex or a $j$-vertex are $C_{\iota+1},\dots,C_q$ with $\iota \in [q]$, then the connected components of $G[S]$ are $(i)
        C_1,\dots,C_q,  \text{ if } \{i,j\} \ \cancel{\subseteq}\ l_G(S)$, or       $C_1,\dots,C_{\iota}\text{ and } \bigcup\limits_{t=\iota+1}^q C_t,  \text{ otherwise. }$ 
        
    Let $C_0^{i,j}(\sigma'):=\bigcup\limits_{\begin{array}{c}
         C\in 2^{[k]},\sigma'(C)\neq 0 \\ \{i,j\}\cap C\neq\emptyset
    \end{array}} C \in 2^{[k]}$.
    
    By definition of the signature 
    $$l_G(S) = \bigcup\limits_{
         C\in 2^{[k]},\sigma'(C)\neq 0} C\in 2^{[k]}.$$
    Summing up, we have
    $\{i,j\}\cap C_0^{i,j}(\sigma') = \{i,j\}\cap l_G(S)$.
    
    Thus if $C_0^{i,j}(\sigma')$ does not contain both a $i$ and $j$, it means that $S$ does not contain either a $i$-vertex or a $j$-vertex. Then, since no edge is created in $G'[S]$ after performing $\eta_{i,j}$, we have $G[S]=G'[S]$, and the signature is thus the same in both graphs. 
    
    Otherwise, all connected components that contain either a $i$-vertex or a $j$-vertex are merged into a greater connected component $C_0$, and we have $l_G(C_0)=C_0^{i,j}(\sigma')$. The other connected components of $G'[S]$ (that contains neither $i$ nor $j$) are unchanged.
    
    It follows that the signature of $\sigma_{\varphi}(S)$ is exactly $\eta_{i,j}(\sigma')$ with:
    $$  \eta_{i,j}(\sigma'):=\sigma' \text{ if } \{i,j\}\cap C_0^{i,j}(\sigma') \neq \{i,j\},$$
    otherwise:
    $ \eta_{i,j}(\sigma')\colon  2^{[k]}  \to  \threebar$ is given by
       $$ \eta_{i,j}(\sigma') (C)= \left\{ \begin{array}{cc}
             \sigma'(C) & \text{ if } \{i,j\}\cap C =\emptyset \\
              & \\
             1 & \text{ if }C=C_0^{i,j}(\sigma') \\
              & \\
             0 &\text{ otherwise.}
        \end{array}\right.
        $$

    Also, the domination of $S$ in $G'$ is exactly $D_{i,j}(\sigma',D')$ with:
    $$ D_{i,j}(\sigma',D') := \left\{ \begin{array}{cc}
        D' & \text{ if } \{i,j\} \cap C_0^{i,j}(\sigma') = \emptyset \\
        D'\cup \{i\} & \text{ if } \{i,j\} \cap C_0^{i,j}(\sigma') = \{j\} \\
        D'\cup \{j\} & \text{ if } \{i,j\} \cap C_0^{i,j}(\sigma') = \{i\} \\
        D'\cup \{i,j\} & \text{ if } \{i,j\} \cap C_0^{i,j}(\sigma') = \{i,j\} \\
    \end{array} \right. $$

    Indeed, assume that $S$ contains a $i$-vertex. Since, in $G$, every $i$-vertex shares an edge with every $j$-vertex (because of $\eta_{i,j}$), then every $j$-vertex is dominated by $S$ in $G$. A symmetric argument applies if $S$ contains a $j$-vertex.

    Finally, letting
    $\eta_{i,j}(\sigma',D'):= (\eta_{i,j}(\sigma'),D_{i,j}(\sigma',D')) \in \threebar^{2^{[k]}}\times 2^{[k]},$
    we have:
    \[Tr_{\varphi}^{-1}(\sigma,D) = \underset{\eta_{i,j}(\sigma',D')=(\sigma,D)}{\biguplus}Tr_{\varphi'}^{-1}(\sigma',D').\]

    for all $(\sigma,D)\in\threebar^{2^{[k]}}\times 2^{[k]}$.
    The same relation is true for the indicator function, leading to the join/union expression:

    \[E_{\varphi}(\sigma,D) = \underset{\eta_{i,j}(\sigma',D')=(\sigma,D)}{\biguplus}E_{\varphi'}(\sigma',D').\]
    
\end{itemize}

This proves that during the execution of Algorithm~\ref{algo:semiring_connected_dom_set_cliquewidth}, at the end of each call of \textbf{CONNECTED-DOM-SET}$(\varphi)$, the variable $E_{\varphi}$ contains a join/union expression of the set $Tr_{\varphi}^{-1}(\sigma,D)$ for all $(\sigma,D)\in \threebar^{2^{[k]}}\times 2^{[k]}$. %

As Algorithm \ref{algo:semiring_connected_dom_set_cliquewidth} outputs $$\underset{\sigma\in \threebar^{2^{[k]}}, \|\sigma\|=1}{\biguplus} E_{\varphi_G}[\sigma,[k]] \quad\text{(with $\|\sigma\|=\sum\limits_{\sigma\in \threebar^{2^{[k]}}}\sigma(C)$),}$$ it thus outputs a join/union expression of the set of connected dominating sets of $G$. Indeed, for all $S\subseteq V_G$, letting $Tr_{\varphi_G}(S)=(\sigma,D)$, $S$ is a dominating-set if and only if $D=[k]$, and $G[S]$ is connected if and only if $\|\sigma\|=1$.

Moreover,  Algorithm \ref{algo:semiring_connected_dom_set_cliquewidth} runs in $(3^{2^k}\times 2^k)^2 = 3^{2^{k+1}}\times 4^k$ time, from the case $\varphi=\varphi_1\oplus\varphi_2$,
\end{proof}

It would be interesting to investigate whether one could even achieve a single exponential running time close to the $O( (2^{\omega+4})^k)$ time algorithm for {\sc Connected-Dominating-Set}~\cite{bergougnoux2019fast} (with $\omega<2.37188$ the exponent of the optimal time complexity of matrix multiplication \cite{duan2022faster}).

\begin{algorithm}
\label{algo:semiring_connected_dom_set_cliquewidth}
\KwData{A graph $G$, a $k$-expression $\varphi_G$ of $G$.}
\KwResult{A join/union expression of the set of connected dominating sets of $G$.}

$\empty$\\

Run \textbf{CONNECTED-DOM-SET}$(\varphi_G)$ and return $\underset{\sigma\in \threebar^{2^{[k]}}, \|\sigma\|=1}{\uplus} E_{\varphi_G}[\sigma,[k]]$, with \textbf{CONNECTED-DOM-SET}$(\varphi)$ being defined for all $k$-expression $\varphi$ as:

$\empty$\\

\textbf{CONNECTED-DOM-SET}$(\varphi)$:

\For{$(\sigma,D)\in \threebar^{2^{[k]}}\times 2^{[k]}$}
{
$E_{\varphi}[\sigma,D]:=\emptyset$
}

\If{$\varphi=\bullet_i(u)$:}
{
$E_{\varphi}[ \sigma_{\bullet_i(u)}(\emptyset) , [k]\setminus \{i\} ] \leftarrow (u\mapsto 0)$\\
$E_{\varphi}[ \sigma_{\bullet_i(u)}(\{u\}) , [k] ] \leftarrow (u\mapsto 1)$
}

\If{$\varphi=\rho_{i\rightarrow j}(\varphi')$:}
{
Run \textbf{CONNECTED-DOM-SET}$(\varphi')$
\\
\For{$(\sigma,D)\in \threebar^{2^{[k]}}\times 2^{[k]}$}
{
$E_{\varphi}[ \rho_{i\to j}(\sigma),\rho_{i\to j}(D) ] \leftarrow E_{\varphi}[ \rho_{i\to j}(\sigma),\rho_{i\to j}(D) ] \uplus E_{\varphi'}[\sigma,D] $
}
}

\If{$\varphi=\eta_{i,j}(\varphi')$:}
{
Run \textbf{CONNECTED-DOM-SET}$(\varphi')$
\\
\For{$(\sigma,D)\in \threebar^{2^{[k]}}\times 2^{[k]}$}
{
$E_{\varphi}[ \eta_{i,j}(\sigma,D) ] \leftarrow E_{\varphi}[ \eta_{i,j}(\sigma,D) ] \uplus E_{\varphi'}[\sigma,D]$
}

}

\If{$\varphi=\varphi_1\oplus\varphi_2$}
{
Run \textbf{CONNECTED-DOM-SET}$(\varphi_1)$ and \textbf{CONNECTED-DOM-SET}$(\varphi_2)$\\
\For{$(\sigma_1,D_1)\in \threebar^{2^{[k]}}\times 2^{[k]}$}
{
\For{$(\sigma_2,D_2)\in \threebar^{2^{[k]}}\times 2^{[k]}$}
{
$E_{\varphi}[\sigma_1 + \sigma_2, D_1\cap D_2] \leftarrow E_{\varphi}[\sigma_1 + \sigma_2, D_1\cap D_2] \uplus (E_{\varphi_1}[\sigma_1,D_1]\Join E_{\varphi_2}[ \sigma_2,D_2 ])$
}
}

}
\caption{An algorithm that solves {\sc Semiring-Connected-Dominating-Set} in time $O^*(3^{2^{k+1}} \times 4^k)$. The notations used are introduced in the proof of Theorem \ref{thm:semiring-connected-dominating-set-FPT}.}
\end{algorithm}

By performing a simpler variant of this idea, it is also possible to derive an algorithm solving {\sc Semiring-Dominating-Set} in time $O^*(16^k)$, assuming a $k$-expression of the input graph is given. The main difference with Algorithm \ref{algo:semiring_connected_dom_set_cliquewidth} is that we instead define the signature of $S\subseteq V_G$ simply as $l_G(S) \in 2^{[k]}$ (it is not necessary to distinguish the connected components of $G[S]$). Thus the signature takes its values in $2^{[k]}$ instead of $\threebar^{2^{[k]}}$, justifying that the complexity is $O^*(  (2^k\times 2^k)^2  ) = O^*(16^k)$ instead of $O^*(  (3^{2^k}\times 2^k)^2  ) = O^*(  (3^{2^{k+1}}\times 4^k)^2 )$ (the square comes from the case where the $k$-expression $\varphi$ is of the form $\varphi:=\varphi_1\oplus\varphi_2$).

If one is only interested in solving $\#${\sc Dominating-Set}, the running time $O^*(16^k)$ is not optimal since there is a faster $O^*(4^k)$ time algorithm~\cite{bodlaender2010faster}. However, this algorithm could still be generalized while keeping the complexity $O^*(4^k)$ if we only consider {\em rings} instead of all semirings. The issue of semirings is that the algorithm requires the use of subtractions, which is not always possible in semirings. If we allow a new operation $\setminus$ in the join/union expression, we could define the {\em ring-expressions}: the semantic $[E_1\setminus E_2]$ of $E_1\setminus E_2$ would be $[E_1]\setminus [E_2]$ if $[E_2]\subseteq [E_1]$ (assuming $E_1$ and $E_2$ are legal), and would fail otherwise. Note that then, for any measure $\mu$ taking its value into a ring, we would have the property that $\mu( [E_1\setminus E_2] ) = \mu([E_1]) - \mu([E_2])$ for all $E_1$ and $E_2$ (assuming the expressions involved are legal), thus Lemma \ref{lem:expression_computation} could easily be extended to ring-expressions. This motivates the introduction of the {\sc Ring-}$\mathbf{\Pi}$ problem with $\mathbf{\Pi}\in$ {\sf NP}, where one is asked to compute a ring-expression of the set of solutions of $\mathbf{\Pi}$. The {\sc Ring-}$\mathbf{\Pi}$ problem is at least not harder than {\sc Semiring-}$\mathbf{\Pi}$ since every join/union expression is in particular a ring-expression. However, it derives strictly less algorithmic applications, since Lemma \ref{lem:expression_computation} then only applies if the semiring considered is a ring.

To our knowledge, even if one is only interested in solving the decision problem {\sc Dominating-Set} parameterized by clique-width, the algorithm running in time $O^*(4^k)$ referred to earlier~\cite{bodlaender2010faster} reaches the best time complexity proposed in the literature. Note that the extension to {\sc Ring-Dominating-Set} that we propose would preserve this complexity.

\section{Semiring-CSP, Sum-Product-CSP and Primal Treewidth}\label{sec:Csp}

As a general application of our approach, we turn to {\sc Semiring-CSP} and the related problem {\sc Sum-Product-CSP} (see Section~\ref{sec:sum_product} for a comparison).
The goal of this section is to obtain a join/union expression for the set of solutions of an instance of {\sc Csp} in {\sf FPT} time parameterized by the so-called primal treewidth.

\begin{figure}[!ht]
\begin{center}
\scalebox{0.8}{
\begin{tikzpicture}

\tikzstyle{vertex}=[circle, draw, inner sep=1pt, minimum width=6pt]
\tikzstyle{sq}=[rectangle, draw, inner sep=1pt, minimum width=6pt]

\node[sq,color=blue] (NPproblem) at (0,5) {$\begin{matrix}\text{Parameterized problem } \\ (\mathbf{\Pi},\lambda)\text{ with }\mathbf{\Pi}\in{\sf NP}\end{matrix}$};

\node[sq,color=blue] (SemiringNP) at (0,2.5) {$\begin{matrix}{\sf FPT}\text{ algorithm that solves} \\ \textsc{Semiring-}\mathbf{\Pi}\end{matrix}$}
    edge[<-,color=blue] (NPproblem);

\node () at (-2.5,3.75) {$\begin{matrix} \text{See Section }\ref{sec:dom-set}\text{ for } (\mathbf{\Pi},\lambda)= \\ (\textsc{Connected-Dom-Set},\mathbf{cw}) \end{matrix}$};

\node[color=blue] () at (2.5,3.75) {$\begin{matrix} \text{See Section }\ref{sec:Csp}\text{ for } (\mathbf{\Pi},\lambda)= \\ (\textsc{CSP},\mathbf{tw}) \end{matrix}$};

\node () at (-7.5,0) {$\begin{matrix} \text{See Section }\ref{sec:semirings} \\ \text{for the} \\\text{construction} \\ \text{of measures}\end{matrix}$};

\node[sq] (B) at (-4.5,0) {$\B$-measures}
    edge (SemiringNP);

\node[sq] (N) at (-1.5,0) {$\N$-measures}
    edge (SemiringNP);

\node[sq] (Rmin) at (1.5,0) {$\R_{\min}$-measures}
    edge (SemiringNP);

\node[sq] (Delta) at (5,0) {$(\R_{\min}\Delta\N)$-measures}
    edge (SemiringNP);

\node[] () at (5,1) {$\begin{matrix} \text{See Section }\ref{sec:Delta-product} \\ \text{for the} \\ \Delta\text{-product}\end{matrix}$};

\node (midB) at (-2.58,1.2) {};
\node (midN) at (-0.96,1.2) {};
\node (midRmin) at (0.96,1.2) {};
\node (midDelta) at (2.85,1.2) {};

\node[sq] (Decision) at (-4.5,-2.5) {$\begin{matrix} \textsf{FPT}\text{ algorithm} \\ \text{that solves }\textsc{List-}\mathbf{\Pi}  \end{matrix}$}
    edge[<-,bend right] (midB);

\node[sq] (Counting) at (-1.5,-2.5) {$\begin{matrix} \textsf{FPT}\text{ algorithm} \\ \text{that solves \#}\mathbf{\Pi}  \end{matrix}$}
    edge[<-,bend right=50] (midN);

\node[sq] (Cost) at (1.5,-2.5) {$\begin{matrix} \textsf{FPT}\text{ algorithm} \\ \text{that solves }\textsc{Cost-}\mathbf{\Pi}  \end{matrix}$}
    edge[<-,bend left=50] (midRmin);

\node[sq] (CountCost) at (5,-2.5) {$\begin{matrix} \textsf{FPT}\text{ algorithm} \\ \text{that solves }\#\textsc{Cost-}\mathbf{\Pi}  \end{matrix}$}
    edge[<-,bend left] (midDelta);

\node () at (-7.5,-2.5) {$\begin{matrix}\text{See Lemma } \ref{lem:expression_computation} \\ \text{and Example }\ref{example:measures} \end{matrix}$};

\node at (5,-1.5) {See Section \ref{sec:count_cost_application}};

\end{tikzpicture}
}
\end{center}
\caption{A general overview of our contributions where the contribution in this section is colored in blue.}
\label{fig:plan_of_the_paper_csp}
\end{figure}
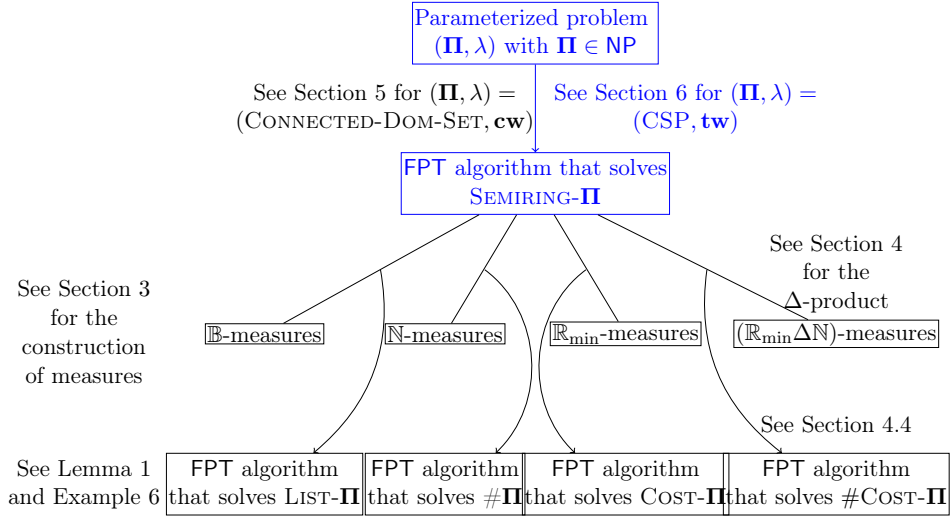

The contribution of this section to the semiring formalism presented in this paper is summarized in Figure \ref{fig:plan_of_the_paper_csp}.

\subsection{Primal tree-width}\label{sec:treewidth}

Recall that the Gaifman graph \cite{rossi2006handbook} of an instance $\I = (V,\mathcal{C})$ of {\sc Csp}$(\Gamma)$ is the graph $G=(V,E_G)$ where $\{u,v\} \in E_G$ if there exists a constraint in $\mathcal{C}$ containing both $u$ and $v$. 
 The {\em primal treewidth} of an instance of {\sc Csp} is the treewidth of its Gaifman graph. \rev{We recommend the reader to refresh the notion of a nice tree decomposition from Section~\ref{sec:width} and the associated notions of $bag(N)$ (the bag associated to a node $N$ of the rooted tree decomposition) and  $descbag(N)$ (the union of the bags of all descendents of $N$).}

We notice in Lemma \ref{lem:bagdescbag} that $descbag(N)$ and $bag(N)$ depend only on the type of the node $N$ and on the $descbag(N')$, and $bag(N')$ for the children $N'$ of $N$.

\begin{lemma}\label{lem:bagdescbag}

Let $T$ be a nice tree decomposition of a graph $G$, and $N\in V_T$.

\begin{enumerate}
    \item If $N$ is a \textbf{\em Leaf}: then $descbag(N)=bag(N)=\emptyset$
    \item If $N$ is of the type \textbf{\em Forget}$(v)$:  then, for $v\in V_G$,
    $$descbag(N)=descbag(child(N))\,\text{ and }\, bag(N)=bag(child(N))\setminus \{v\}.$$
    \item If $N$ is of the type \textbf{\em Introduce}$(v)$: then, for $v\in V_G$,  $v\notin descbag(child(N))$, and $$descbag(N)=descbag(child(N))\uplus \{v\} \text{ and } bag(N)=bag(child(N))\uplus \{v\}.$$ Moreover, for each edge $\{u,v\}\in E_G$ with $u\in descbag(N)$,  $u\in bag(N)$.
    \item If $N$ is of the type \textbf{\em Join}, and its two children of $N$ are $N_1\in N_T$ and $N_2\in N_T$, then $descbag(N)=descbag(N_1)\cup descbag(N_2)$ and $$bag(N)= descbag(N_1)\cap descbag(N_2) = bag(N_1)=bag(N_2).$$ Moreover, there is no edge in $G$ with an endpoint in $descbag(N_1) \setminus bag(N) $ and an endpoint in $descbag(N_2) \setminus bag(N) $.
\end{enumerate}

\end{lemma}

\begin{proof}

We prove each item individually, depending on the type of $N$.

\begin{enumerate}
    
    \item If $N$ is a leaf,  then $descbag(N)=bag(N)=\emptyset$ by definition of a nice tree decomposition.
    
    \item If $N$ is of the type \textbf{Forget}$(v)$, the result is clear.
    
    \item If $N$ is of the type \textbf{Introduce}$(v)$, we fist prove that $v\notin descbag(child(N))$. Otherwise, there exists $N_v\in desc(child(N))$ that contains $v$, and $child(N)$ is thus on the path between $N$ and $N_v$. By definition of a tree decomposition, we would have $bag(N)\cap bag(N_v)\subseteq bag(child(N))$ and $v\in bag(child(N))$ which is false. Thus, $v\notin descbag(child(N))$.

    We now prove that for every edge $\{u,v\}\in E_G$ with $u\in descbag(N)$, we have $u\in bag(N)$. Let $N_u$ be a descendant of $N$ with $u\in bag(N_u)$. By definition of a tree decomposition, there exists a node $N_{u,v}$ whose bag contains both $u$ and $v$. If $N_{u,v}$ is a descendant of $N$ then $child(N)$ is on the path between $N$ and $N_{u,v}$, and we have that $bag(N) \cap bag(N_{u,v}) \subseteq  bag(child(N))$. This is a contradiction since $v\in bag(N) \cap bag(N_{u,v})$ but $v\notin bag(child(N))$, so this case is impossible. Thus, since $N_{u,v}$ is not a descendant of $N$, $N$ is on the path between $child(N)$ and $N_{u,v}$ so $bag(N_u)\cap bag(N_{u,v})\subseteq bag(N)$, which implies that $u\in bag(N)$.

    \item If $N$ is of the type \textbf{Join}, and its two children of $N$ are $N_1\in N_T$ and $N_2\in N_T$, then $bag(N)=bag(N_1)=bag(N_2)$ by definition of a nice tree decomposition. It follows $descbag(N)=descbag(N_1) \cup descbag(N_2)$. Moreover, by definition of a tree decomposition, we get that $descbag(N_1) \cap descbag(N_2) \subseteq bag(N)$ through a similar reasoning as in the ``Introduce'' case. Moreover, $bag(N)=bag(N_1) \subseteq descbag(N_1)$, and $bag(N)=bag(N_2) \subseteq descbag(N_2)$. It follows that $bag(N) \subseteq descbag(N_1) \cap descbag(N_2)$, and the equality follows.

    We now prove that an edge $\{u,v\}\in E_G$ of $G$ can not have an endpoint $u\in descbag(N_1)\setminus bag(N)$ and the other endpoint $v\in descbag(N_2)\setminus bag(N)$. Assume by contradiction that this is the case. Let $N_u$ a descendant of $N_1$ with $u\in bag(N_u)$ and symmetrically $N_v$ a descendant of $N_2$ with $v\in bag(N_v)$. By definition of a tree decomposition, there exists a node $N_{u,v}$ containing both $u$ and $v$. We remark that $N$ is either on the path between $N_u$ and $N_{u,v}$ contradicting that $u\notin bag(N)$ or on the path between $N_v$ and $N_{u,v}$, contradicting that $v\notin bag(N)$.
    
\end{enumerate}

\end{proof}

It can be useful to interpret a nice tree decomposition $(T,bag)$ of a graph $G$ as a construction of a graph with two labels on its vertices, say black and white. For every node $N$, we can associate a graph $G_N$ built at the node $N$, whose vertices satisfy:

\begin{itemize}

    \item The set of vertices of $G_N$ is $descbag(N)$.

    \item The set of black vertices of $G_N$ is $bag(N)$. The other vertices in $descbag(N)\setminus bag(N)$ are white.

\end{itemize}

With this interpretation, the treewidth of the tree decomposition is the maximum number of black vertices in a graph $G_N$ (for a node $N$ of $T$) minus 1. 
The interpretation $G_N$ of a node $N$ can be recursively derived from its children, depending on the four possible types of $N$.

\begin{enumerate}
    \item \textbf{Leaf:} Construct the empty graph.

    \item \textbf{Forget$(v)$:} The black vertex $v$ in $G_{child(N)}$ is now white in $G_N$.

    \item \textbf{Introduce$(v)$:} To build $G_N$, we add a new black vertex $v$ in $G_{child(N)}$, and we can build an edge between $v$ and any black vertex.

    \item \textbf{Join:} Letting $N_1$ and $N_2$ be the two children of $N$, $G_{N_1}$ and $G_{N_2}$ must have the same black graphs (i.e., their black vertices must induce the same graph). We then obtain $G_N$ from $G_{N_1}$ and $G_{N_2}$ by identifying their black graphs. %
    
\end{enumerate}

\begin{example}

\rev{In Figure \ref{fig:join_node}, an example of $N$ being a Join node with two children $N_1$ and $N_2$. Here, $descbag(N_1)=\{a,b,c,d,e,h\}, descbag(N_2)=\{e,f,g,h,i\}$ and thus $descbag(N)=descbag(N_1)\cup descbag(N_2) = \{a,b,c,d,e,f,g,h,i\}$ and $bag(N)=bag(N_1)=bag(N_2) = descbag(N_1)\cap descbag(N_2) = \{e,h\}$.} %

\end{example}

\begin{center}
\begin{figure}

\begin{tikzpicture}[scale=1]
\tikzstyle{vertex}=[draw,shape=circle];

\begin{scope}
    \node[vertex] (u) at (0,0) {$a$};
    
    \node[vertex] (v) at (-1,1) {$b$}
        edge (u);
        
    \node[vertex] (w) at (1,1) {$c$}
        edge (u)
        edge (v);

    \node[vertex] (r) at (0,-1) {$d$}
        edge (u);

    \node[vertex,fill=black] (y) at (2,0) {$\mathcolor{white}{e}$}
        edge (u)
        edge (w);

    \node[vertex,fill=black] (p) at (2,-1) {$\mathcolor{white}{h}$}
        edge (r)
        edge (y);

    \node () at (1,-2) {$G_{N_1}$};

\end{scope}

\begin{scope}[xshift=2cm]

    \node[vertex,fill=black] (y) at (2,0) {$\mathcolor{white}{e}$}
;

    \node[vertex] (z) at (3.5,0) {$f$}
        edge (y);

    \node[vertex] (q) at (3.5,-1) {$g$}
        edge (z);

    \node[vertex,fill=black] (p) at (2,-1) {$\mathcolor{white}{h}$}
        edge (q)
        edge (y);

    \node[vertex] (i) at (3.5,1) {$i$}
        edge (z);

    \node () at (2.75,-2) {$G_{N_2}$};

\end{scope}

\begin{scope}[xshift=8cm]
    \node[vertex] (u) at (0,0) {$a$};
    
    \node[vertex] (v) at (-1,1) {$b$}
        edge (u);
        
    \node[vertex] (w) at (1,1) {$c$}
        edge (u)
        edge (v);

    \node[vertex] (r) at (0,-1) {$d$}
        edge (u);

    \node[vertex,fill=black] (y) at (2,0) {$\mathcolor{white}{e}$}
        edge (u)
        edge (w);

    \node[vertex] (z) at (3.5,0) {$f$}
        edge (y);

    \node[vertex] (q) at (3.5,-1) {$g$}
        edge (z);

    \node[vertex,fill=black] (p) at (2,-1) {$\mathcolor{white}{h}$}
        edge (r)
        edge (q)
        edge (y);

    \node[vertex] (i) at (3.5,1) {$i$}
        edge (z);

    \node () at (2,-2) {$G_N$};

\end{scope}

\end{tikzpicture}
\caption{An example of $N$ being a Join node with two children $N_1$ and $N_2$.}
\label{fig:join_node}
\end{figure}
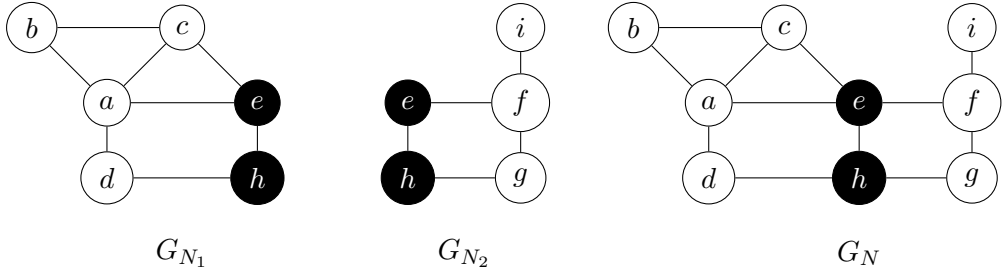
\end{center}

\subsection{FPT algorithm}

We now give a concrete application of the semiring formalism to the constraint satisfaction problem, by giving an {\sf FPT} algorithm when parameterized by the primal tree-width. Note that \textsc{Semiring-Csp} subsumes \textsc{Semiring-Sat}. \rev{It is furthermore known that the $\mathcal{R}$-semiring extensions of \textsc{Sat} are hard in the setting of the complexity class \textsf{NP}($\mathcal{R}$) defined via {\em semiring Turing machines}~\cite{eiter2023semiring}.}
Thus, we should not expect to solve these problems in polynomial time and an \textsf{FPT} algorithm is then highly desirable.

\begin{theorem}\label{thm:semiring_csp}

For every finite set $\Gamma$ of relation over a finite domain $D$, {\sc Semiring-Csp}$(\Gamma)$ is solvable in time $O^*(|D|^{\mathbf{tw}(G)})$ on any instance $\I$, with $G$ the Gaifman graph of $\I$ (assuming that a tree decomposition of $G$ is given).

\end{theorem}

\begin{proof}
Let $(T, bag)$ be a nice tree decomposition of the Gaifman graph $G = (V_G, E_G)$ of an instance $\I$ of {\sc Semiring-Csp}$(\Gamma)$ \rev{of treewidth \textbf{tw}$(G)$} \rev{(recall that we can convert a given tree decomposition to a nice tree decomposition in linear time~\cite{bodlaender1991better}).}

\rev{Notice that we can assume, without loss of generality, that $bag(\mathrm{root}(T)) = \emptyset$. This can be achieved by adding \textbf{Forget} nodes above the original root, one for each vertex in $bag(\mathrm{root}(T))$, until all these vertices have been forgotten. Crucially, this transformation does not affect the treewidth, as the newly added bags are strictly smaller than the original root's bag.}

Let $\lambda \colon V_G\to N_T$ be a function which maps every $v\in V_G$ to an arbitrary \textbf{Introduce}$(v)$ node $\lambda(v)\in N_T$ with $v\in\lambda(v)$ (clearly, such a node exists for all $v\in V_G$). The interest of the function $\lambda$ is to prevent the following issue:

In the example of Figure \ref{fig:join_node}, if we had join/union expressions $E_{N_1}$ and $E_{N_2}$ of the set of solutions over $descbag(N_1)$ and $descbag(N_2)$, the join/union expression $E_{N_1}\Join E_{N_2}$ would not be legal, as the domains of $E_{N_1}$ and $E_{N_2}$, which are $descbag(N_1)$ and $descbag(N_2)$, would not be disjoint since their intersection is $bag(N)= \{e,h\}$. To solve the problem, we only express restrictions of such solutions, and $\lambda$ serves to decide (arbitrarly) who (if any) between $N_1$ and $N_2$ takes $e$ and/or $h$ in its domain. \rev{Naturally, any concrete choice of $\lambda$ (e.g., lexicographical) would also work here.}

Keeping this technical issue in mind, we introduce  the following definitions for each node $N\in N_T$:

\begin{itemize}

\item $dom(N):=\{ v\in V_G\mid \lambda(v)\in desc(N) \} \subseteq descbag(N)$.

\item For each $f \colon bag(N)\to D$, 
$$\F_N(f):=
\{ F: descbag(N)\to D \, \mid \,F\text{ is a partial solution  s.t. } F|_{bag(N)} = f\},$$ {\it  i.e.,} for every constraint $C$ that involves only variables in $descbag(N)$, each $F\in\F_N(f)$ satisfies $C$ and

$$\E_N(f):=\{ F|_{dom(N)} \mid F\in\F_N(f) \}.$$

\end{itemize}

According to Remark \ref{rem:trace}, when treating a node $N$, the "trace" of a partial solution $F:descbag(N)\to D$ will be its restriction to the current $bag(N)$ (i.e., its restriction to the ``black vertices'' of $G_N$). Then, $\F_N(f)$ is the set of partial solutions that have the trace $f$. We provide a join/union expression for $\E_N(f)$ instead, as discussed above.

The first item of Remark \ref{rem:trace} is indeed respected, as there are at most $|D|^{\mathbf{tw}(G)}$ possible traces at every node (which is \textsf{FPT} with respect to the primal tree-width).
Since we give a join/union expression of the sets of partial {\em solutions}, the second item of Remark \ref{rem:trace} is irrelevant.

We now justify the third item of Remark \ref{rem:trace} stating that  $\E_N(f)$ can be inductively computed, when given a nice tree decomposition.

\begin{claim}
\rev{
We prove by induction (on the structure of $T$) that every node $N\in N_T$ satisfies the property:}

\hfill

\rev{
$\mathcal{P}(N)$: ``At the end of the treatment of the node $N$ in Algorithm \ref{algo:tree-width_semiring}, for each $f \colon bag(N)\to D$, we have $[E_N(f)]=\E_N(f)$, \rev{where $[E_N(f)]$ refers to the semantic of the expression $E_N(f)$}.''
}
\end{claim}

\begin{proofclaim}
    
\begin{itemize}

    \item[1.] If $N$ is a \textbf{Leaf}, $bag(N)=descbag(N)=dom(N)=\emptyset$: the only function $f$ to consider is $f:\emptyset \to D$. Clearly $\F_N(f)=\E_N(f)= D^{\emptyset}$. Hence, $\mathcal{P}(N)$ is true. 

    \item[2.] If $N$ is of the form \textbf{Forget}($v$) with $v\in bag(child(N))$, assume that $\mathcal{P}(child(N))$ is true.
Recall that $bag(N)=bag(child(N))\setminus \{v\}$ and
note that $descbag(N)=descbag(child(N))$ and $dom(N)=dom(child(N))$. We prove that, for each $f:bag(N)\to D$, we have

$$ \E_N(f) = \underset{\begin{matrix} f' : bag(N)\uplus\{v\} \to D \\ f'|_{bag(N)}=f \end{matrix}}{\biguplus} \E_{child(N)}(f') $$

which implies $\mathcal{P}(N)$ by $\mathcal{P}(child(N))$ and the definition of $E_N(f)$ in Algorithm~\ref{algo:tree-width_semiring}.

Firstly, notice that $v\in dom(N)$ (i.e., $\lambda(v)\in desc(N)$), because otherwise $N$ is on the path between $child(N)$ and $\lambda(v)$ in $T$, and we obtain by definition of a tree decomposition that $bag(child(N)) \cap bag(\lambda(v)) \subseteq bag(N)$, which is a contradiction, since $v\in (bag(child(N)) \cap bag(\lambda(v))) \setminus bag(N)$.

Then, the partitioning of $\E_N(f)$ into equivalence classes of $\sim$ defined by 
$$\forall (F_1,F_2)\in (\E_N(f))^2, F_1\sim F_2 \iff F_1(v)=F_2(v),$$  gives the desired equality.

\item[3.] If $N$ is of the form \textbf{Introduce}($v$) with $v\notin bag(child(N))$, assume that $\mathcal{P}(child(N))$ is true.
Recall that $bag(N)=bag(child(N))\uplus \{v\}$ and that $descbag(N)=descbag(child(N))\uplus \{v\}$, and that

$dom(N)=\left \{\begin{array}{cc}
     dom(child(N)) & \text{ if }\lambda(v)\neq N  \\
     dom(child(N))\uplus \{v\} & \text{ if }\lambda(v)=N 
\end{array} \right.$

Indeed, $v\notin descbag(child(N))$ by Lemma \ref{lem:bagdescbag}, which implies in particular that $v\notin dom(child(N))$.

Recall that proving $\mathcal{P}(N)$  knowing $\mathcal{P}(child(N))$, requires a relation between $\E_N(f)$ and $\E_{child(N)}(f|_{bag(child(N))})$. Since $\E_N(f)$ is defined using $\F_N(f)$, it will be useful to compute $\F_N(f)$.

We prove:

$$\F_N(f) = \left\{ \begin{array}{cc}
     \emptyset & \text{if }f\text{ violates some constraint}  \\
     \F_{child(N)}(f|_{bag(child(N))})\Join \{ v\mapsto f(v)\} & \text{if }f\text{ respects every constraint}
\end{array} \right. $$

Firstly, notice that $\F_{child(N)}(f|_{bag(child(N))})\Join \{v\mapsto f(v)\}$ is a valid operation, since $v\notin descbag(N)$.

If $f$ violates a constraint, every function that has $f$ as a restriction also violates a constraint, and is thus not a partial solution. This justifies that $\F_N(f) = \emptyset$ if $f$ violates a constraint.

If $f$ respects every constraint, notice that every element $F:descbag(N)\mapsto D$ of $\F_N(f)$ is in particular in $\F_{child(N)}(f|_{bag(child(N))})\Join \{ v\mapsto f(v) \}$. Indeed, since $F$ coincides with $f$ on $bag(N)=bag(child(N))\uplus \{v\}$, we have $F(v)=f(v)$, and $F|_{bag(child(N))}$ is a partial solution (it is a restriction of the partial solution $F$) and coincides with $f|_{bag(child(N))}$ on $bag(child(v))$.

Conversely, we prove that every map $F:descbag(N)\to D$ of $$\F_{child(N)}(f|_{bag(child(N))})\Join \{ v\mapsto f(v) \}$$ is a partial solution. 

Notice first that $F|_{bag(child(N))}$ is in $\F_{child(N)}(f|_{bag(child(N))})$, which is a solution. Therefore, $F$ respects every constraint that is not involving $v$. We now prove that $F$ respects every constraint involving $v$.

Let $u\in descbag(N)$ be such that $\{u,v\}$ is an edge of the Gaifman graph $G$. Let $N_u$ be a descendant of $N$ with $u\in bag(N_u)$. By Lemma \ref{lem:bagdescbag}, we have $u\in bag(N)$ and, by definition of the Gaifman graph $G$ of the instance $\I$ of {\sc Csp}, this proves that any constraint over $descbag(N)$ that involves $v$ involves only variables of $bag(N)$.

Since $f \colon bag(N)\to D$ respects the constraints over $bag(N)$ (by assumption), and since $F$ coincides with $f$ on $bag(N)$, this shows that $F$ respects every constraint involving $v$, and thus $F$ respects every constraint over $descbag(N)$. Hence, $F$ is a  partial solution. This proves that $$\F_N(f) = \F_{child(N)}(f|_{bag(child(N))})\Join \{ v\mapsto f(v) \}.$$ 
By restricting to $dom(N)$, we get that:

$$\E_N(f) = \left\{ \begin{array}{cc}
     \emptyset & \text{if }f\text{ violates a constraint}  \\
      & \\
     \E_{child(N)}(f|_{bag(child(N))}) & \begin{matrix} \text{if }f \text{ respects every} \\ \text{constraint and }\lambda(v)\neq N \end{matrix} \\
      & \\
     \E_{child(N)}(f|_{bag(child(N))})\Join \{ v\mapsto f(v)\} & \begin{matrix} \text{if }f\text{ respects every} \\\text{constraint and }\lambda(v) = N \end{matrix}
\end{array} \right. $$

and thus $\mathcal{P}(N)$ conforms to $\mathcal{P}(child(N))$ and the definition of $E_N(f)$ in Algorithm \ref{algo:tree-width_semiring}.

\item[4.] If $N$ is of the form \textbf{Join}($N_1,N_2$), assume that $\mathcal{P}(N_1)$ and $\mathcal{P}(N_2)$ are true.
Recall that $bag(N)=bag(N_1)=bag(N_2)$ and  that $descbag(N) = descbag(N_1) \cup descbag(N_2)$, $bag(N)= descbag(N_1) \cap descbag(N_2)$ by Lemma \ref{lem:bagdescbag}. Note also that $dom(N)=dom(N_1)\uplus dom(N_2)$.

Let $f \colon bag(N)\to D$. We prove that
$\F_N(f) = \F_{N_1,N_2}(f)$,
where
\begin{eqnarray*}
 \F_{N_1,N_2}(f)&:=& \{ F_1|_{descbag(N_1)\setminus bag(N)} \mid F\in \F_{N_1}(f) \}\\
 &\Join& \{ F_2|_{descbag(N_2)\setminus bag(N)} \mid F\in \F_{N_2}(f) \} \Join \{f\}\\
 &=&\F_{N_1}(f) \Join \{ F_2|_{descbag(N_2)\setminus bag(N)} \mid F\in \F_{N_2}(f) \}\\
 &=&\{ F_1|_{descbag(N_1)\setminus bag(N)} \mid F\in \F_{N_1}(f) \} \Join \F_{N_2}(f).
\end{eqnarray*}

The idea is that every pair $(F_1,F_2)\in \F_{N_1}(f)\times \F_{N_2}(f)$ coincides on the intersection of their domain, since by definition, they both coincide with $f$ on $bag(N)=descbag(N_1)\cap descbag(N_2)$. Note that we cannot use the notation $F_1\Join F_2$ since their domain is not disjoint but that the sets $ \E_N(f), \E_{N_1}(f)$ and $\E_{N_2}(f)$ are more convenient to use since $dom(N_1)$ and $dom(N_2)$ are disjoint.

Since a restriction of a partial solution is still a partial solution, every element of $\F_N(f)$ is in $\F_{N_1,N_2}(f)$\rev{, ie. $\F_N(f) \subseteq \F_{N_1,N_2}(f)$ }. We now prove that $\F_{N_1,N_2}(f)\subseteq \F_N(f)$.

By Lemma \ref{lem:bagdescbag}, every edge $\{u,v\}\in E_G$ in $descbag(N)$ can not have an endpoint in $descbag(N_1)\setminus bag(N)$ and the other endpoint in $descbag(N_2)\setminus bag(N)$. So, by definition of the Gaifman graph, every constraint in $descbag(N)$ is either a constraint over $descbag(N_1)$ or over $descbag(N_2)$ (recall that $bag(N)=descbag(N_1)\cap descbag(N_2)$ by Lemma \ref{lem:bagdescbag}. 

Take $F:descbag(N)\to D$ of $\F_{N_1,N_2}(f)$. We have that $F$ coincides with $f$ on $bag(N)=bag(N_1)=bag(N_2)$, and thus $F$ coincides with a partial solution $F_1\in \F_{N_1}(f)$ on $descbag(N_1)$. This proves that $F$ respects every constraint over $descbag(N_1)$. Symmetrically, $F$ respects every constraint over $descbag(N_2)$. We have proven that $F$ respects all the constraints, {\it  i.e.,} $F$ is a partial solution, and thus $F$ is in $\F_N(f)$.
\rev{This proves that $\F_{N_1,N_2}(f) \subseteq \F_N(f)$, and then that $\F_{N_1,N_2}(f) = \F_N(f)$}

We get as a corollary (by restricting to $dom(N)=dom(N_1)\uplus dom(N_2)$) that:

$$ \E_N(f) = \E_{N_1}(f) \Join \E_{N_2}(f) $$

and $\mathcal{P}(N)$ follows.

\end{itemize}

This proves $\mathcal{P}(N)$, for every node $N\in N_T$.

\end{proofclaim}

Now, since $bag(root(T))=\emptyset$, $dom(root(T))=descbag(root(T))=V_G$, $\mathcal{P}(root(T))$ justifies that $[E_{root(T)}(\emptyset\to D)]$ encodes the set of solutions of $\I$.

Lastly, it is straightforward to see that the complexity of Algorithm~\ref{algo:tree-width_semiring} is indeed $O^*(|D|^{\mathbf{tw}(G)})$\rev{, as the only non-polynomial part of the algorithm is the loop iterating over ``\textbf{for} $f: bag(N)\to D$'', for each node $N$. The number of functions to explore for each node $N$ is $|D|^{|bag(N)|} \le |D|^{\textbf{tw}(G)+1}$, since $(T,bag)$ has treewidth \textbf{tw}(G). Recall that $|D|$ is considered to be a constant in regard to the complexity of Algorithm \ref{algo:tree-width_semiring}}. %
\end{proof}

\begin{algorithm}
\caption{Algorithm to construct an expression of the set of solutions on any instance of {\sc Csp}($\Gamma$) in time $|D|^{\mathbf{tw}(G)}$ (with $D$ being the domain of the relations in $\Gamma$).} 
\label{algo:tree-width_semiring}
\KwData{An instance $\I$ of {\sc Csp}($\Gamma$), a nice tree decomposition $T=(N_T,E_T)$ of its Gaifman graph $G$, whose root has empty bag, and a function $\lambda:V_G\to N_T$ where $\forall v\in V_G$, $\lambda(v)$ is an Introduce node with $v\in bag(\lambda(v))$.}
\KwResult{A join/union expression $E_{SOL}$ over $(V_G,D)$ encoding the set of solutions of $\I$.}

$\empty$\\

Run \textbf{Expressions}($root(T)$) and return $E_{root(T)}(\emptyset\to D)$ with \textbf{Expressions}($N$) being defined for each node $N$ of $T$ as:

$\empty$\\

\textbf{Expressions}($N$):\\

\If{$N$ is a leaf ($bag(N)=\emptyset$)}
{
$E_N(\emptyset\to D):= D^{\emptyset}$
}

\If{$N$ is Forget$(v)$}
{
Run $\textbf{Expressions}(child(N))$\\
\For{$f:bag(N)\to D$}
{
$E_N(f):= \emptyset$
}
\For{$f':bag(N)\cup\{v\}\to D$}
{
$E_N(f'|_{bag(N)})\leftarrow E_N(f'|_{bag(N)})\uplus E_{child(N)}(f')$
}
}

\If{$N$ is Introduce$(v)$}
{
Run $\textbf{Expressions}(child(N))$\\
\For{$f:bag(N)\to D$}
{
\If{$f$ respects every constraint and $\lambda(v) = N$}
{
$E_N(f) := E_{child(N)}(f|_{bag(N)\setminus\{v\}}) \Join \{ v\mapsto f(v) \}$
}
\If{$f$ respects every constraint and $\lambda(v) \neq N$}
{
$E_N(f) := E_{child(N)}(f|_{bag(N)\setminus\{v\}})$
}
\If{$f$ violates a constraint}
{
$E_N(f) := \emptyset$
}
}
}
\If{$N$ is Join$(N_1,N_2)$}
{
Run $\textbf{Expressions}(N_1)$ and $\textbf{Expressions}(N_2)$\\
\For{$f:bag(N)\to D$}
{$E_N(f):=E_{N_1}(f) \Join E_{N_2}(f)$
}
}
\end{algorithm}

In particular, for {\sc Semiring}-$H$-{\sc Coloring} we obtain the following corollary.

\begin{corollary}\label{cor:H-coloring_treewidth}

For every graph $H$, {\sc Semiring-}$H$-{\sc Coloring} is solvable in time $O^*(|V_H|^{\mathbf{tw}(G)})$.

\end{corollary}

In particular, Corollary \ref{cor:H-coloring_treewidth} implies that $H$-{\sc Coloring}, $\#H$-{\sc Coloring}, {\sc List}-$H$-{\sc Coloring} and {\sc \#Cost}-$H$-{\sc Coloring} can be solved in $O^*(|V_H|^{\mathbf{tw}(G)})$ time. To the best of our knowledge, no general algorithm for counting $H$-colorings of minimal cost was previously known in the literature.

We now indicate how to generalize this algorithm to solve the even more general problem {\sc Sum-Product-Csp} (Algorithm \ref{algo:tree-width_sum_product}). If we restrict Algorithm \ref{algo:tree-width_sum_product} to solve only {\sc \#Csp}, this algorithm is a particular case of an already known algorithm \cite{ganian2016combining}. 
In Theorem \ref{thm:sum-product_CSP_treewidth}, we require that each constraint $(C,x_1,\dots,x_{ar(C)})$ is mapped to a unique \textbf{Introduce} node \rev{$\lambda(C,x_1,\dots,x_{ar(C)})$} such that every vertex of the Gaifman graph (i.e., variable) involved in constraint $C$ belongs to $bag(\rev{\lambda(C,x_1,\dots,x_{ar(C)})})$. This property can always be ensured because such a set of vertices is a clique (complete graph), and because a clique in a graph is necessarily contained in the bag of one node in every tree decomposition.

\begin{theorem}\label{thm:sum-product_CSP_treewidth}

For every finite set $\Gamma$ of constraints over a finite domain $D$, {\sc Sum-Product-Csp}$(\Gamma)$ is solvable in time $O^*(|D|^{\mathbf{tw}(G(\I))})$ on any instance $\I$, where $\mathbf{tw}(G(\I))$ is the treewidth of the Gaifman graph of $\I$ (assuming that a tree decomposition is given).

\end{theorem}

\begin{proof}
Similarly to Theorem \ref{thm:semiring_csp} we give a proof by induction on every node $N$ of the rooted tree decomposition via the claim:

\begin{claim}

$\mathcal{P}(N):$``For every $f:bag(N)\to D$, in the execution of Algorithm \ref{algo:tree-width_sum_product}, \rev{the variable} $E_N(f)$ contains %
$$\sum\limits_{\tiny\begin{matrix}F:descbag(N)\to D \\ F|_{bag(N)}=f\end{matrix}}\prod\limits_{\tiny\begin{matrix}(C,x_1,\dots,x_{\footnotesize ar(C)}) \text{ constraint with} \\ \{x_1,\dots,x_{\footnotesize ar(C)}\} \subseteq descbag(N) \text{ and} \\ \rev{\lambda(C,x_1,\dots,x_{ar(C)})} \text{ is a descendant of }N\end{matrix}} C(F(x_1),\dots,F(x_{ar(C)})).\text{''}\hfill \qedhere$$

\end{claim}

\end{proof}

\begin{algorithm}
\caption{Algorithm to solve an instance of {\sc Sum-Product-Csp}($\Gamma$) in time $|D|^{\mathbf{tw}(G)}$ (with $D$ being the domain of the relations of $\Gamma$).}\label{algo:tree-width_sum_product}
\KwData{An instance $\I$ of {\sc Sum-Product-Csp}($\Gamma$), a nice tree decomposition $T$ of its Gaifman graph $G$, whose root has empty bag, and such that every constraint \rev{$(C,x_1,\dots,x_{ar(C)})$} is mapped to a unique \textbf{Introduce} node $\rev{\lambda(C,x_1,\dots,x_{ar(C)})}$.}
\KwResult{An element $E_{SOL}$ of the semiring containing $\sum\limits_{F:V_G\to D}\,\, \prod\limits_{(C,x_1,\dots x_{ar(C)}) \text{ constraint}} C(F(x_1),\dots,F(x_{ar(C)} ))$ 
}

$\empty$\\

Run \textbf{Sum-Product}($root(T)$) and return $E_{root(T)}(\emptyset\to D)$ with \textbf{Sum-Product}($N$) being defined for every node $N$ of $T$ as:

$\empty$\\

\textbf{Sum-Product}($N$):\\
\If{$N$ is a leaf ($bag(N)=\emptyset$)}
{
$E_N(\emptyset\to D):=1$
}
\If{$N$ is Forget$(v)$}
{
Run $\textbf{Sum-Product}(child(N))$\\
\For{$f:bag(N)\to D$}
{
$E_N(f):= 0$
}
\For{$f':bag(N)\cup\{v\}\to D$}
{
$E_N(f'|_{bag(N)})\leftarrow E_N(f'|_{bag(N)}) + E_{child(N)}(f')$
}
}

\If{$N$ is Introduce$(v)$}
{
Run $\textbf{Sum-Product}(child(N))$\\
\For{$f:bag(N)\to D$}
{
$E_N(f) := E_{child(N)}(f|_{bag(N)\setminus\{v\}})$

\For{\rev{$(C,x_1,\dots,x_{ar(C)})$} constraint with $\rev{\lambda(C,x_1,\dots,x_{ar(C)})}=N$}
{
$E_N(f) \leftarrow E_N(f)\times \rev{C(f(x_1),\dots,f(x_{ar(C)}))}$
}
}
}
\If{$N$ is Join$(N_1,N_2)$}
{
Run $\textbf{Sum-Product}(N_1)$ and $\textbf{Sum-Product}(N_2)$\\
\For{$f:bag(N)\to D$}
{
$E_N(f)\leftarrow E_{N_1}(f) \times E_{N_2}(f)$
}
}
\end{algorithm}

Moreover, this running time seems to be optimal at least in the case of $k$-{\sc Coloring} and in most cases of $H$-{\sc Coloring}, while giving access to every semiring extension. Indeed, we even have a lower bound for the $k$-{\sc Coloring} problem, that relies on the {\em strong exponential time hypothesis} ({\sf SETH})~\cite{impagliazzo2009}. The {\sf SETH} is a well-known conjecture within complexity theory which states that for every $\varepsilon>0$, there exists $k\ge 3$ such that  {\sc $k$-Sat} problem can not be solved in time $O^*((2-\varepsilon)^n)$, where $n$ is the number of variables.

\begin{theorem}\cite{okrasa2020finegrained}
For all $k\ge 3$ and $\varepsilon>0$, $k$-{\sc Coloring} is not solvable in time $O^*((k-\varepsilon)^{\mathbf{tw}(G)})$ under the {\sf SETH}.
\end{theorem}

This lower bound applies to $H$-{\sc Coloring} problems for a broad class of graphs. Precisely, Theorem \ref{thm:lower_bound_H} applies as long as $H$ is a so-called {\em projective core} on at least $3$ vertices \cite{okrasa2020finegrained}. %

\begin{theorem}\cite{okrasa2020finegrained}\label{thm:lower_bound_H}
If $H$ is a projective core on at least three vertices, then, for every $\varepsilon>0$, $H$-{\sc Coloring} is not solvable in time $O^*((|V_H|-\varepsilon)^{\mathbf{tw}(G)})$ under the {\sf SETH}.
\end{theorem}

Thus running time obtained in Corollary \ref{cor:H-coloring_treewidth} (and even more so in Theorem \ref{thm:semiring_csp} and Theorem \ref{thm:sum-product_CSP_treewidth}) cannot not be asymptotically improved under the {\sf SETH}. Note also that, asymptotically, almost all graphs are projective cores \cite{hell1992core,luczak2004note,okrasa2020finegrained}, and thus Theorem \ref{thm:lower_bound_H} applies to almost all graphs.

\section{Conclusion}\label{sec:discussion}

In this article we explored semiring extensions for computational problems. We proposed a general approach (Section~\ref{sec:semirings}) that does not require particular a semiring structure, and that is applicable to any problem with a reasonable notion of a certificate (e.g., any problem in ${\sf NP}$). In Section~\ref{sec:Delta-product} we proceeded to defining a novel operation on semirings, the $\Delta$-product, which increases the scope of semiring algorithms and, in particular, solves the problem of counting solutions of minimal cost. To illustrate our framework, we turned our attention to two well-known problems with a vast number of applications:  {\sc Connected-Dominating-Set} and finite-domain \textsc{Csp}s. For the former, we successfully gave an algorithm that outputs a join/union expression in \textsf{FPT} time when parameterized by the clique-width of the input graph (Section~\ref{sec:dom-set}). Despite being a problem with many practical applications~\cite{10.5555/2412083},  the counting extension of this problem was not known to be \textsf{FPT}. Moreover, we proved a similar result for the \textsc{Csp} problem parameterized by primal treewidth, and showed how to extend the algorithm to the more general \textsc{Sum-Product-Csp} problem (Section~\ref{sec:Csp}). 

Let us now discuss some potential directions of future research. 

\subsection{Algorithmic applications}

We have seen that existing algorithms in the literature (e.g., the algorithm by Ganian et al.~\cite{ganian2016combining} for \#{\sc Csp}) can be generalized (under some assumptions) to semiring extensions. Is this a common behaviour? Are there more examples of algorithms in the literature, whether it is for the counting, list, or cost version, which can be easily adapted to solving the semiring extension? Or are there problems where, e.g., counting the number of solutions of minimal cost is prohibitively more computationally expensive than merely counting solutions?

As discussed at the end of Section \ref{sec:dom-set}, our work opens up the study of a novel {\em ring formalism} built on {\em ring-expressions}, where a new operation $\setminus$ corresponding to the set difference is introduced. Thus, the semantics of $[E_1\setminus E_2]$ for two legal ring expressions $E_1$ and $E_2$ such that  $[E_2]\subseteq [E_1]$, would be $[E_1]\setminus [E_2]$. Indeed, it is not hard to verify that if a semiring measure takes its value in a ring instead, Lemma~\ref{lem:expression_computation} still applies even if  a ring-expression  is given (by mapping $\setminus$ to the additive inverse $-_A$ of the ring $A$), which enables similar algorithmic applications. 

From a complexity perspective, it would be interesting to unify problems of the form \textsc{Semiring}-$\mathbf{\Pi}$ with $\mathbf{\Pi}$ in a new complexity class and, similarily, counting, optimizing, and counting solutions of optimal costs to \textsf{NP} problems belonging to the classes \#\textsf{P}, {\sf OptP} and {\sf \#$\cdot$OptP}. Moreover, a convenient notion of ``reduction'' preserving this complexity class would be desirable, as well as a completeteness notion for \textsc{Semiring} problems. Note that some \#\textsf{P}-complete  problems of these classes are extensions of \textsf{NP} problems that are not \textsf{NP}-complete (unless \textsf{P=NP}). For instance, counting the number of perfect matchings of a bipartite graph is known to be \#\textsf{P}-complete~\cite{DBLP:journals/tcs/Valiant79}, whereas deciding if one exists is in \textsf{P}. Does there also exist easy problems whose semiring extensions become hard?

\subsection{Sum-Product CSP}

Even though our formalism extends decision problems (corresponding to the Boolean semiring $\B$) to arbitrary commutative semirings, our formalism does not subsume ``soft constraints'' allowed, e.g., in the {\sc Sum-Product-Csp} problem. Is there a possibility that any candidate solution $f$, instead of being either a solution ($f\in SOL$) or not  a solution ($f\notin SOL$), could be mapped to a semiring value instead? Then, the formalism of ``hard constraints'' presented in this article would correspond to the particular case where $f\in SOL$ if and only if $f$ is mapped to $\top$ (and $f\notin SOL$ if and only if $f$ is mapped to $\bot$).

\subsection{Combinatoric and algebraic results}

The definition of a join/union expression and its algorithmic applications raises the question of which sets of functions can be expressed efficiently \rev{via a join/union expression}. In particular, Theorem \ref{thm:no_semiring_permutations} and Theorem \ref{thm:no_semiring_cycles} suggest that the sets $\mathfrak{S}_n$ and $\mathfrak{C}_n$ can not be represented by an expression of polynomial size. Proving or disproving Conjecture \ref{conj:no_semiring_permutations} and Conjecture \ref{conj:no_semiring_cycles} would likely provide useful tools to better understand join/union expressions.
As a first step, it would be interesting to see whether we could determine at least whether Conjecture~\ref{conj:no_semiring_permutations} implies Conjecture \ref{conj:no_semiring_cycles}, and vice versa. Note that we indeed have the bijection:

$$\Phi:\begin{matrix} \mathfrak{S}_n & \to & \mathfrak{C}_{n+1} \\ \sigma & \mapsto & (\sigma(1)\ \sigma(2)\ \dots\ \sigma(n)\ n+1)
\end{matrix}$$

for all $n\ge 1$, whose inverse is 

$$\Psi:\begin{matrix} \mathfrak{C}_{n+1} & \to & \mathfrak{S}_n \\ c & \mapsto & \begin{pmatrix}
    1 & 2 & \dots & n \\ c(n+1) & c^2(n+1) & \dots & c^n(n+1)
\end{pmatrix}.
\end{matrix}$$

However, it is not clear if the bijection $\Phi$ (respectively $\Psi$) can be used to derive a join/union expression of $\mathfrak{C}_{n+1}$ (respectively $\mathfrak{S}_n$) from a join/union expression of $\mathfrak{S}_n$ (respectively $\mathfrak{C}_{n+1}$).

\subsection{Sequential join/union expressions}
It may be useful to introduce {\em sequences of join/union expressions}, {\it  i.e.,} sequences of the form $(\mathcal{E}_0,\dots,\mathcal{E}_m)$ with $m\ge 0$,  and where $\mathcal{E}_0$ is a set of join/union expressions and, for all $i\in [m]$, $\mathcal{E}_i$ is a set of join/union expressions whose leaves can also be expressions of $\mathcal{E}_j$ with $j<i$. We may then obtain an  result analogous  to Lemma \ref{lem:expression_computation} by iteratively computing  $\mu([E])$ for $E\in\mathcal{E}_0$, and then $\mu([E])$ for $E\in\mathcal{E}_1,\dots$. This may become relevant since, even if a join/union expression $E$ is repeated multiple times, it would not be necessary to compute $\mu([E])$ several times, similarly to dynamic programming. 

Considering that Theorem \ref{thm:no_semiring_permutations} and Theorem \ref{thm:no_semiring_cycles} also apply to sequential join/union expressions (since their proofs only exploit algorithmic applications), it could be interesting to answer Conjecture \ref{conj:no_semiring_permutations} and Conjecture \ref{conj:no_semiring_cycles} when sequential join/union expressions are considered.

\begin{reve}
\subsection{Harder problems}
Last, although we have concentrated on problems in NP, there is nothing stopping us from defining semiring extensions for richer complexity classes. As a starting point one could consider algorithms in the literature which are able to count solutions of minimal cost, since such algorithms could conceivably be generalized to the semiring setting. For example, the \textsc{projected SAT} problem mentioned in Example~\ref{ex:projected} is \#NP-hard and it is likely straightforward to extend the treewidth algorithm by Fichte et al.~\cite{FICHTE2023103810} to take minimality into account. Similarly, beyond NP we find the polynomial hierarchy culminating in PSpace, where the canonical problems is the {\em quantified Boolean formula} problem. Here, there are efficient model counting algorithms based on knowledge compilation~\cite{DBLP:conf/stacs/CapelliM19}, and it would be interesting to see if similar techniques are applicable to semiring extensions for PSpace-complete problems.
\end{reve}

\section*{Acknowledgement}

The authors wish to thank 
Pierre Charbit, Eun Jung Kim and Noleen Köhler
for the  many fruitful discussions. We are particularly thankful for their ideas and  suggestions for improving the contents of the paper. We also thank the two anonymous reviewers for their extensive feedback.

The first author is partially supported by the ENS de Lyon.

The second author was partially supported by TAILOR, a project funded by EU
Horizon 2020 research and innovation program under GA No 952215.

Additionally, the third author was partially supported by the Swedish Research Council (VR) under grants VR-2019-03690 and VR-2022-03214.

\bibliographystyle{elsarticle-harv} 
\bibliography{biblio}

@inproceedings{DBLP:conf/stacs/CapelliM19,
  author       = {F. Capelli and
                  S. Mengel},
  title        = {Tractable {QBF} by Knowledge Compilation},
  booktitle    = {Proceedings of the 36th International Symposium on Theoretical Aspects of Computer Science
                  ({STACS}-2019)},
  series       = {LIPIcs},
  volume       = {126},
  pages        = {18:1--18:16},
  publisher    = {Schloss Dagstuhl - Leibniz-Zentrum f{\"{u}}r Informatik},
  year         = {2019},
}

@article{okrasa2020finegrained,
      title={Fine-grained complexity of the graph homomorphism problem for bounded-treewidth graphs}, 
  author={K. Okrasa and P. Rzazewski},
  journal   = {{SIAM} J. Comput.},
  volume    = {50},
  number    = {2},
  pages     = {487--508},
  year      = {2021}
}

@inproceedings{eiter2021complexity,
  title={On the Complexity of Sum-of-Products Problems over Semirings},
  author={Eiter, Thomas and Kiesel, Rafael},
  booktitle={Proceedings of the AAAI Conference on Artificial Intelligence},
  volume={35},
  number={7},
  pages={6304--6311},
  year={2021}
}

@book{cygan2015parameterized,
   title={Parameterized algorithms},
   author={Cygan, Marek and Fomin, Fedor V and Kowalik, {\L}ukasz and Lokshtanov, Daniel and Marx, D{\'a}niel and Pilipczuk, Marcin and Pilipczuk, Micha{\l} and Saurabh, Saket},
   year={2015},
   publisher={Springer}
 }

@inproceedings{ganian2016combining,
  author       = {Robert Ganian and
                  M. S. Ramanujan and
                  Stefan Szeider},
  editor       = {Heribert Vollmer and
                  Brigitte Vall{\'{e}}e},
  title        = {Combining Treewidth and Backdoors for {CSP}},
  booktitle    = {34th Symposium on Theoretical Aspects of Computer Science, {STACS}
                  2017, March 8-11, 2017, Hannover, Germany},
  series       = {LIPIcs},
  volume       = {66},
  pages        = {36:1--36:17},
  publisher    = {Schloss Dagstuhl - Leibniz-Zentrum f{\"{u}}r Informatik},
  year         = {2017}
}

@inproceedings{fan2023fine,
  author       = {Austen Z. Fan and
                  Paraschos Koutris and
                  Hangdong Zhao},
  editor       = {Kousha Etessami and
                  Uriel Feige and
                  Gabriele Puppis},
  title        = {The Fine-Grained Complexity of Boolean Conjunctive Queries and Sum-Product
                  Problems},
  booktitle    = {50th International Colloquium on Automata, Languages, and Programming,
                  {ICALP} 2023, July 10-14, 2023, Paderborn, Germany},
  series       = {LIPIcs},
  volume       = {261},
  pages        = {127:1--127:20},
  publisher    = {Schloss Dagstuhl - Leibniz-Zentrum f{\"{u}}r Informatik},
  year         = {2023}
}

@article{berman2001stipulations,
  title={Stipulations, multi-valued logic and De Morgan algebras},
  author={Berman, J and Blok, W},
  journal={Multi-valued Logic},
  volume={7},
  number={5-6},
  pages={391--416},
  year={2001}
}

@article{courcelle2000upper,
  title={Upper bounds to the clique width of graphs},
  author={Courcelle, Bruno and Olariu, Stephan},
  journal={Discrete Applied Mathematics},
  volume={101},
  number={1-3},
  pages={77--114},
  year={2000},
  publisher={Elsevier}
}

@article{bergougnoux2019fast,
  title={Fast exact algorithms for some connectivity problems parameterized by clique-width},
  author={Bergougnoux, Benjamin and Kant{\'e}, Mamadou},
  journal={Theoretical Computer Science},
  volume={782},
  pages={30--53},
  year={2019},
  publisher={Elsevier}
}

@inproceedings{bodlaender2010faster,
  title={Faster algorithms on branch and clique decompositions},
  author={Bodlaender, Hans L and Van Leeuwen, Erik Jan and Van Rooij, Johan MM and Vatshelle, Martin},
  booktitle={Mathematical Foundations of Computer Science 2010: 35th International Symposium, MFCS 2010, Brno, Czech Republic, August 23-27, 2010. Proceedings 35},
  pages={174--185},
  year={2010},
  organization={Springer}
}

@article{DBLP:journals/tcs/Valiant79,
  author       = {Leslie G. Valiant},
  title        = {The Complexity of Computing the Permanent},
  journal      = {Theoretical Computer Science},
  volume       = {8},
  pages        = {189--201},
  year         = {1979}
}

@article{floyd1968,
  author = {R. W. Floyd},
  title = {Algorithm 97 (SHORTEST PATH)},
  journal = {Communications of the ACM},
  volume = {18},
  year = {1968}
}

@article{warshall1962,
  author = {S. Warshall},
  title = {A Theorem on Boolean Matrices},
  journal = {Journal of the ACM},
  volume = {9(1)},
  year = {1962},
  pages = {11--12}
}

@inbook{kleene1956,
  author = {S. C. Kleene},
  title = {Representation of Events in Nerve Nets and Finite Automata},
  booktitle = {Automata Studies},
  series = {Annals of Mathematics Studies},
  volume = {34},
  publisher = {Princeton University Press},
  year = {1956},
  pages = {3--42}
}

@article{lehmann1977,
  author = {D. J. Lehmann},
  title = {Algebraic Structures for Transitive Closures},
  journal = {Theoretical Computer Science},
  volume = {4},
  year = {1977},
  pages = {59--76}
}

@article{FICHTE2023103810,
title = {Solving Projected Model Counting by Utilizing Treewidth and its Limits},
journal = {Artificial Intelligence},
volume = {314},
pages = {103810},
year = {2023},
issn = {0004-3702},
doi = {https://doi.org/10.1016/j.artint.2022.103810},
url = {https://www.sciencedirect.com/science/article/pii/S0004370222001503},
author = {Johannes K. Fichte and Markus Hecher and Michael Morak and Patrick Thier and Stefan Woltran},
}

@article{mohri2002,
  author       = {Mehryar Mohri},
  title        = {Semiring Frameworks and Algorithms for Shortest-Distance Problems},
  journal      = {Journal of Automata, Languages and Combinatorics},
  volume       = {7},
  number       = {3},
  pages        = {321--350},
  year         = {2002},
  url          = {https://doi.org/10.25596/jalc-2002-321}
}

@inproceedings{impagliazzo2009,
  author    = {Chris Calabro and
               Russell Impagliazzo and
               Ramamohan Paturi},
  title     = {The Complexity of Satisfiability of Small Depth Circuits},
  booktitle = {Parameterized and Exact Computation, 4th International Workshop ({IWPEC} 2009)}, 
  pages     = {75--85},
  year      = {2009},
  url       = {https://doi.org/10.1007/978-3-642-11269-0_6},
  doi       = {10.1007/978-3-642-11269-0_6},
  bibsource = {dblp computer science bibliography, http://dblp.org}
}

@article{kleene1938notation,
  title={On notation for ordinal numbers},
  author={Kleene, Stephen Cole},
  journal={The Journal of Symbolic Logic},
  volume={3},
  number={4},
  pages={150--155},
  year={1938},
  publisher={Cambridge University Press}
}

@article{bistarelli1997semiring,
  title={Semiring-based constraint satisfaction and optimization},
  author={Bistarelli, Stefano and Montanari, Ugo and Rossi, Francesca},
  journal={Journal of the ACM (JACM)},
  volume={44},
  number={2},
  pages={201--236},
  year={1997},
  publisher={ACM New York, NY, USA}
}

@article{bacchus2009solving,
  title={Solving SAT and Bayesian inference with backtracking search},
  author={Bacchus, Fahiem and Dalmao, Shannon and Pitassi, Toniann},
  journal={Journal of Artificial Intelligence Research},
  volume={34},
  pages={391--442},
  year={2009}
}

@book{10.5555/2412083,
author = {Du, Ding-Zhu and Wan, Peng-Jun},
title = {Connected Dominating Set: Theory and Applications},
year = {2012},
isbn = {1461452414},
publisher = {Springer Publishing Company, Incorporated},
year = {2012},
}

@article{hermann2009complexity,
  title={Complexity of counting the optimal solutions},
  author={Hermann, Miki and Pichler, Reinhard},
  journal={Theoretical computer science},
  volume={410},
  number={38-40},
  pages={3814--3825},
  year={2009},
  publisher={Elsevier}
}

@inproceedings{hermann2008counting,
  title={Counting complexity of minimal cardinality and minimal weight abduction},
  author={Hermann, Miki and Pichler, Reinhard},
  booktitle={European Workshop on Logics in Artificial Intelligence},
  pages={206--218},
  year={2008},
  organization={Springer}
}

@article{hertrampf1990relations,
  title={Relations among MOD-classes},
  author={Hertrampf, Ulrich},
  journal={Theoretical Computer Science},
  volume={74},
  number={3},
  pages={325--328},
  year={1990},
  publisher={Elsevier}
}

@inproceedings{krentel1986complexity,
  title={The complexity of optimization problems},
  author={Krentel, Mark W},
  booktitle={Proceedings of the eighteenth annual ACM symposium on Theory of computing},
  pages={69--76},
  year={1986}
}

@article{eiter2023semiring,
  title={Semiring Reasoning Frameworks in AI and Their Computational Complexity},
  author={Eiter, Thomas and Kiesel, Rafael},
  journal={Journal of Artificial Intelligence Research},
  volume={77},
  pages={207--293},
  year={2023}
}

@inproceedings{duan2022faster,
  author       = {Ran Duan and
                  Hongxun Wu and
                  Renfei Zhou},
  title        = {Faster Matrix Multiplication via Asymmetric Hashing},
  booktitle    = {64th {IEEE} Annual Symposium on Foundations of Computer Science, {FOCS}
                  2023, Santa Cruz, CA, USA, November 6-9, 2023},
  pages        = {2129--2138},
  publisher    = {{IEEE}},
  year         = {2023}
}

@book{rossi2006handbook,
  editor       = {Francesca Rossi and
                  Peter van Beek and
                  Toby Walsh},
  title        = {Handbook of Constraint Programming},
  series       = {Foundations of Artificial Intelligence},
  volume       = {2},
  publisher    = {Elsevier},
  year         = {2006}
}

@inproceedings{bodlaender1991better,
  title={Better algorithms for the pathwidth and treewidth of graphs},
  author={Bodlaender, Hans L and Kloks, Ton},
  booktitle={Automata, Languages and Programming: 18th International Colloquium Madrid, Spain, July 8--12, 1991 Proceedings 18},
  pages={544--555},
  year={1991},
  organization={Springer}
}

@article{courcelle2000linear,
  title={Linear time solvable optimization problems on graphs of bounded clique-width},
  author={Courcelle, Bruno and Makowsky, Johann A and Rotics, Udi},
  journal={Theory of Computing Systems},
  volume={33},
  number={2},
  pages={125--150},
  year={2000},
  publisher={Springer}
}

@inproceedings{hegerfeld2023tight,
  author       = {Falko Hegerfeld and
                  Stefan Kratsch},
  editor       = {Inge Li G{\o}rtz and
                  Martin Farach{-}Colton and
                  Simon J. Puglisi and
                  Grzegorz Herman},
  title        = {Tight Algorithms for Connectivity Problems Parameterized by Clique-Width},
  booktitle    = {31st Annual European Symposium on Algorithms, {ESA} 2023, September
                  4-6, 2023, Amsterdam, The Netherlands},
  series       = {LIPIcs},
  volume       = {274},
  pages        = {59:1--59:19},
  publisher    = {Schloss Dagstuhl - Leibniz-Zentrum f{\"{u}}r Informatik},
  year         = {2023}
}

@book{jensen2011graph,
  title={Graph coloring problems},
  author={Jensen, Tommy R and Toft, Bjarne},
  year={2011},
  publisher={John Wiley \& Sons}
}

@article{thomassen19953,
  title={3-list-coloring planar graphs of girth 5},
  author={Thomassen, Carsten},
  journal={Journal of Combinatorial Theory, Series B},
  volume={64},
  number={1},
  pages={101--107},
  year={1995},
  publisher={Elsevier}
}

@article{hell1992core,
  title={The core of a graph},
  author={Hell, Pavol and Ne{\v{s}}et{\v{r}}il, Jaroslav},
  journal={Discrete Mathematics},
  volume={109},
  number={1-3},
  pages={117--126},
  year={1992},
  publisher={Elsevier}
}

@article{luczak2004note,
  title={Note on projective graphs},
  author={{\L}uczak, Tomasz and Ne{\v{s}}et{\v{r}}il, Jaroslav},
  journal={Journal of Graph Theory},
  volume={47},
  number={2},
  pages={81--86},
  year={2004},
  publisher={Wiley Online Library}
}

\appendix

\section{Proofs of Section \ref{sec:Delta-product}}\label{app:ProofsDeltaProduct}

We recall and prove Theorem \ref{thm:DeltaProductSemiring}.

\DeltaProductSemiring*

\begin{proof}

Note that we can reformulate the definition of $\oplus$ as:

$\oplus: \begin{matrix} (D\Delta A)^2 & \to & D\Delta A \\ ((d_1,a_1),(d_2,a_2)) & \mapsto & \left\{ \begin{array}{cc}
     (d_1,a_1) & \text{ if } d_1 <_D d_2  \\
     (d_2,a_2) & \text{ if } d_2 <_D d_1 \\
     (d_1,a_1+_A a_2) & \text{ if } d_1=d_2
\end{array} \right\} \end{matrix}$

by Property \ref{prop:minimum}.

First, we need to prove that $D\Delta A$ is indeed {\em stable} by $\oplus$ and $\otimes$, \rev{i.e., that applying the operations to two elements of the set gives an element of the set.}

\begin{itemize}
    
    \item We prove that $D\Delta A$ is stable by $\oplus$.
    
    Let $((d_1,a_1),(d_2,a_2))\in (D\Delta A)^2$.
    
    \begin{itemize}
    
        \item Assume that $d_1\neq d_2$. Then $(d_1,a_1)\oplus (d_2,a_2)\in\{(d_1,a_1),(d_2,a_2)\}$. It follows that $(d_1,a_1)\oplus (d_2,a_2)\in D\Delta A$.
        
        \item Assume that $d_1=d_2$ is multiplicatively regular. Then, $(d_1,a_1)\oplus (d_2,a_2)=(d_1,a_1+_Aa_2) \in D\Delta A$ because $d_1$ is regular.
    
        \item Assume that $d_1=d_2$ is not multiplicatively regular. Then, since $((d_1,a_1),(d_2,a_2))\in (D\Delta A)^2$, it follows that $a_1=a_2=0_A$. Then, $(d_1,a_1)\oplus (d_2,a_2)=(d_1,a_1+_Aa_2) \in D\Delta A$ because $a_1+_Aa_2=0_A+_A0_A = 0_A$.
    
    \end{itemize}
    
    $D\Delta A$ is stable by $\oplus$.
    
    \item We prove that $D\Delta A$ is stable by $\otimes$.
    
    Let $((d_1,a_1),(d_2,a_2))\in (D\Delta A)^2$.
    
    \begin{itemize}
        
        \item Assume that $d_1$ and $d_2$ are multiplicatively regular. Then, by Property \ref{prop:mult regular}, $d_1+_Dd_2$ is multiplicatively regular, and thus $(d_1,a_1)\otimes (d_2,a_2)=(d_1+_Dd_2,a_1\times_Aa_2)\in D\Delta A$.
        
        \item Assume that $d_1$ is not multiplicatively regular. Then since $(d_1,a_1)\in D\Delta A$, $a_1=0_A$, and thus $(d_1,a_1)\otimes (d_2,a_2) = (d_1+_Dd_2,a_1\times a_2) \in D\Delta A$ because $a_1\times_Aa_2=0_A\times a_2=0_A$.
        
        \item Similar reasoning and conclusion if $d_2$ is not multiplicatively regular.
        
    \end{itemize}
    
    $D\Delta A$ is stable by $\otimes$.

    \item $\oplus$ and $\otimes$ are clearly commutative by the commutativity of $\min_D,+_D,+_A$ and $\times_A$.

    \item We now prove the associativity of $\oplus$:
    
    Let $((d_1,a_1),(d_2,a_2),(d_3,a_3))\in(D\Delta A)^3$. We prove that $(d_1,a_1)\oplus((d_2,a_2)\oplus(d_3,a_3)) = ((d_1,a_1)\oplus(d_2,a_2))\oplus (d_3,a_3)$.
    
    We only prove the result when $d_1\leq_Dd_2\leq_Dd_3$. \rev{The other cases can then be obtained using the commutativity.} %
    
    \begin{itemize}
    
    \item Assume that $d_1<_Dd_2$.
Then there exists $a$ such that $(d_2,a_2)\oplus(d_3,a_3)=(d_2,a)$. Indeed, we take $a:=a_2$ if $d_2<d_3$ and $a:=a_2+_Aa_3$ if $d_2=d_3$.
       Hence, $$(d_1,a_1)\oplus((d_2,a_2)\oplus(d_3,a_3)) = (d_1,a_1)\oplus (d_2,a) = (d_1,a_1)$$
and    
    $$((d_1,a_1)\oplus(d_2,a_2))\oplus (d_3,a_3) = (d_1,a_1)\oplus (d_3,a_3) = (d_1,a_1),$$ because $d_1<_Dd_3$.
     
    \item Assume that $d_1=d_2<_Dd_3$.   Then $$(d_1,a_1)\oplus((d_2,a_2)\oplus(d_3,a_3)) = (d_1,a_1)\oplus (d_2,a_2) = (d_1,a_1+_Aa_2)$$
    and
    $$((d_1,a_1)\oplus(d_2,a_2))\oplus (d_3,a_3) = (d_1,a_1+_Aa_2)\oplus (d_3,a_3) = (d_1,a_1+_Aa_2).$$
    
    \item Assume that $d_1=d_2=d_3$.
    Then, \begin{eqnarray*}(d_1,a_1)\oplus((d_2,a_2)\oplus(d_3,a_3)) &=& (d_1,a_1)\oplus (d_2,a_2+_Aa_3)\\ 
    &=& (d_1,a_1+_A(a_2+_Aa_3))
    \end{eqnarray*} 
    and 
    \begin{eqnarray*}
      ((d_1,a_1)\oplus(d_2,a_2))\oplus (d_3,a_3) &=& (d_1,a_1+_Aa_2)\oplus (d_3,a_3)\\
      &=& (d_1,a_1+_A(a_2+_Aa_3))\\
      &=& (d_1,(a_1+_Aa_2)+_Aa_3),  
    \end{eqnarray*} by the associativity of $+_A$.
    \end{itemize}

    \item We now prove the associativity of $\otimes$.    
    Let $((d_1,a_1),(d_2,a_2),(d_3,a_3))\in(D\Delta A)^3$. Then
    \begin{eqnarray*} (d_1,a_1)\otimes((d_2,a_2)\otimes(d_3,a_3))&=&(d_1,a_1)\otimes(d_2+_Dd_3,a_2\times_Aa_3)\\
    &=&(d_1+_D(d_2+_Dd_3),a_1\times_A(a_2\times_Aa_3))\\ 
    &=& ((d_1+_Dd_2)+_Dd_3,(a_1\times_Aa_2)\times_Aa_3) \\
    &=& (d_1,a_1)\otimes((d_2,a_2)\otimes(d_3,a_3))\\
    &=&(d_1,a_1)\otimes(d_2+_Dd_3,a_2\times_Aa_3)\\
    &=&(d_1+_D(d_2+_Dd_3),a_1\times_A(a_2\times_Aa_3))\\ &=& ((d_1+_Dd_2)+_Dd_3,(a_1\times_Aa_2)\times_Aa_3)\\
    &=&((d_1,a_1)\otimes(d_2,a_2))\otimes(d_3,a_3),
    \end{eqnarray*}
    by associativity of $+_D$ and $\times_A$.
    
    \item $(\infty_D,0_A)$ is clearly neutral for $\oplus$ because $\infty_D$ is neutral for $\min_D$, and $0_A$ is neutral for $+_A$.
    
    \item $(0_D,1_A)$ is clearly neutral for $\otimes$, because $0_D$ is neutral for $+_D$, and because $1_A$ is neutral for $\times_A$.
    
    \item $(\infty_D,0_A)$ is clearly absorbing for $\otimes$, because $\infty_D$ is absorbing for $+_D$ and because $0_A$ is absorbing for $\times_A$.
    
    \item We now prove that $\otimes$ is distributive over $\oplus$.
    
    Let $((d_1,a_1),(d_2,a_2),(d_3,a_3))\in (D\Delta A)^3$. On the one hand,
    \begin{eqnarray*}
       ((d_1,a_1)\oplus(d_2,a_2))&\otimes& (d_3,a_3) \\
       &=& \left\{\begin{array}{cc}
         (d_1,a_1) & \text{ if }d_1<_Dd_2 \\
         (d_2,a_2) & \text{ if }d_2<_Dd_1 \\
         (d_1,a_1+_Aa_2) & \text{ if }d_1=d_2
    \end{array} \right\} \otimes (d_3,a_3)\\
    &=&\left\{\begin{array}{cc}
         (d_1+_Dd_3,a_1\times_Aa_3) & \text{ if }d_1<_Dd_2 \\
         (d_2+_Dd_3,a_2\times_Aa_3) & \text{ if }d_2<_Dd_1 \\
         (d_1+_Dd_3,(a_1+_Aa_2)\times_A a_3) & \text{ if }d_1=d_2
    \end{array} \right\}.
    \end{eqnarray*}
    On the other hand,
\begin{eqnarray*}((d_1,a_1)&\otimes&(d_3,a_3))\oplus((d_2,a_2)\otimes(d_3,a_3))\\
&=& (d_1+_Dd_3,a_1\times_1a_3)\oplus (d_2+_Dd_3,a_2\times_Aa_3)\\
&=& \left\{\begin{array}{cc}
         (d_1+_Dd_3,a_1\times_Aa_3) & \text{ if }d_1+_Dd_3<_Dd_2+_Dd_3 \\
         (d_2+_Dd_3,a_2\times_Aa_3) & \text{ if }d_2+_Dd_3<_Dd_1+_Dd_3 \\
         (d_1+_Dd_3,(a_1\times_Aa_3)+_A(a_2\times a_3)) & \text{ if }d_1+_Dd_3=d_2+_Dd_3
    \end{array} \right\}.
    \end{eqnarray*}
    \begin{itemize}
        
        \item Assume that $d_3$ is multiplicatively regular. Then, by Property \ref{prop:multiplicatively regular order} (second item) \rev{and the fact that the three conditions stated on the right-hand sides below are mutually exclusive}, 
        \begin{eqnarray*}
        d_1<_Dd_2&\iff& d_1+_Dd_3<_Dd_2+_Dd_3,\\ 
        d_2<_Dd_1&\iff& d_2+_Dd_3<_Dd_1+_Dd_3, \text{ and}\\ d_1=d_2&\iff& d_1+_Dd_3=d_2+_Dd_3.
        \end{eqnarray*}
        Since $\times_A$ is distributive over $+_A$, we also have $$(a_1+_Aa_2)\times_A a_3 = (a_1\times_Aa_3)+_A(a_2\times a_3).$$ 
        This proves that $((d_1,a_1)\oplus(d_2,a_2))\otimes (d_3,a_3) = ((d_1,a_1)\otimes(d_3,a_3))\oplus((d_2,a_2)\otimes(d_3,a_3))$.
    
        \item Assume that $d_3$ is not multiplicatively regular. Since $(d_3,a_3)\in D\Delta A$, $a_3=0_A$.
    
        \rev{Up to using the commutativity,} we may assume that $d_1\leq_D d_2$. %
        Then, by Property \ref{prop:multiplicatively regular order} (first item), $d_1+_Dd_3\leq_Dd_2+_Dd_3$.
        Since $a_3=0_A$, we deduce that:
       \begin{eqnarray*} 
        ((d_1,a_1)\oplus(d_2,a_2))\otimes (d_3,a_3)&=& (d_1+_Dd_3,0_A)\\
    (d_1,a_1)\otimes(d_3,a_3))\oplus((d_2,a_2)\otimes(d_3,a_3))&=&(d_1+_Dd_3,0_A).
    \end{eqnarray*}
        Hence, we have proved that
        $$((d_1,a_1)\oplus(d_2,a_2))\otimes (d_3,a_3)=((d_1,a_1)\otimes(d_3,a_3))\oplus((d_2,a_2)\otimes(d_3,a_3)),$$
     thus showing that 
        $\otimes$ is indeed distributive over $\oplus$.\qedhere
    \end{itemize}
\end{itemize}

\end{proof}

\LexOrder*

\begin{proof}

Since $\D_1$ and $\D_2$ are semirings, $\D_1\Delta\D_2$ is a semiring. It is idempotent: indeed, let $(d_1,d_2)\in D_1\times D_2$, then $(d_1,d_2)\oplus (d_1,d_2) \underset{\text{Because }d_1=d_1}{=} (\min_1(d_1,d_1),\min_2(d_2,d_2))=(d_1,d_2)$ by idempotence of $\D_1$ and $\D_2$.

It only remains to prove that the relation $\leq_{\mathcal{D}_1\Delta\mathcal{D}_2}$ defined by:

$\forall ((a_1,a_2),(b_1,b_2))\in (D_1\times D_2)^2, (a_1,a_2)\leq_{\mathcal{D}_1\Delta\mathcal{D}_2} (b_1,b_2) \iff \exists (c_1,c_2)\in D_1\times D_2, (a_1,a_2)= (b_1,b_2)\oplus(c_1,c_2)$

is the same relation as $\leq_{\D_1 lex \D_2}$, defined as the lexicographical order with respects to $\leq_{\D_1}$ and $\leq_{\D_2}$.

Indeed, since $\leq_{\D_1 lex \D_2}$ is a total order (because $\leq_{\D_1}$ and $\leq_{\D_2}$ are), this will prove that $\D_1\Delta\D_2$ is an idempotent totally-ordered dioid.

Let $((a_1,a_2),(b_1,b_2))\in (D_1\times D_2)^2$.

\begin{itemize}

    \item Assume that $(a_1,a_2)\leq_{\mathcal{D}_1\Delta\mathcal{D}_2} (b_1,b_2)$, {\it  i.e.,} $\exists (c_1,c_2)\in D_1\times D_2, (a_1,a_2)= (b_1,b_2)\oplus(c_1,c_2)$.

    \begin{itemize}
    
        \item If $b_1<_{\D_1}c_1$, then $(a_1,a_2)=(b_1,b_2)$. In particular, $(a_1,a_2)\leq_{\mathcal{D}_1lex\mathcal{D}_2} (b_1,b_2)$.
    
        \item If $c_1<_{\D_1}b_1$, then $(a_1,a_2)=(c_1,c_2)<_{\D_1lex\D_2}(b_1,b_2)$.

        \item If $b_1=c_1$, then $(a_1,a_2)=(b_1,\min_2(b_2,c_2))\leq_{\D_1lex\D_2}(b_1,b_2)$

    \end{itemize}
    
    \item Assume that $(a_1,a_2)\leq_{\D_1 lex\D_2}(b_1,b_2)$.

    \begin{itemize}
    
        \item If $a_1<_{\D_1}b_1$, then $(a_1,a_2) = (b_1,b_2)\oplus (a_1,a_2)$, which proves that $(a_1,a_2)\leq_{\D_1\Delta\D_2}(b_1,b_2)$.
    
        \item Else, $a_1=b_1$ and $a_2\leq_{\D_2}b_2$. Then, $(a_1,a_2) = (b_1,b_2)\oplus (a_1,a_2)$, which proves that $(a_1,a_2)\leq_{\D_1\Delta\D_2}(b_1,b_2)$.

    \end{itemize}
    
\end{itemize}

$\leq_{\D_1\Delta\D_2}$ and $\leq_{\D_1lex\D_2}$ are the same relation, which concludes the proof.

\end{proof}

\begin{lemma}\label{lemma:tuple of regular}

Let $(a_1,a_2)\in D_1\Delta D_2$. Then, $(a_1,a_2)\in reg(\D_1\Delta\D_2)$ if and only if $a_1\in reg(\D_1)$ and $a_2\in reg(\D_2)$.

\end{lemma}

\begin{proof}

\begin{itemize}
    
    \item Assume that $a_1$ and $a_2$ are regular.
    
    Let $((b_1,b_2),(c_1,c_2))\in (D_1\times D_2)^2$ be such that $$(a_1,a_2)\otimes(b_1,b_2) = (a_1,a_2)\otimes(c_1,c_2).$$
    By definition of $\otimes$, $$(a_1+_1b_1,a_2+_2b_2) = (a_1+_1c_1,a_2+_2c_2).$$ From the fact that $a_1$ and $a_2$ are regular, it follows that $b_1=c_1$ and $b_2=c_2$, {\it  i.e.,} $(b_1,b_2)=(c_1,c_2)$. This proves that $(a_1,a_2)$ is regular.
    
    \item Assume that $(a_1,a_2)$ is regular. Let $(b_1,c_1)\in A^2$ be such that $a_1+_1b_1=a_1+_1c_1$. Then, $(a_1,a_2)\otimes (b_1,0_2) = (a_1+_1b_1,a_2) = (a_1+_1c_1,a_2) = (a_1,a_2)\otimes (c_1,0_2)$. Since $(a_1,a_2)$ is regular, we have that $(b_1,0_2)=(c_1,0_2)$. In particular, $b_1=c_1$, which proves that $a_1$ is regular. A similar argument shows that $a_2$ is also regular.
    
\end{itemize}

\end{proof}

\begin{lemma}\label{lem:successive_delta}

For all $(d_1,d_2,a)\in D_1\times D_2\times A$, $(d_1\Delta d_2)\Delta a = d_1\Delta (d_2\Delta a)$.

\end{lemma}

\begin{proof}

\begin{itemize}

\item Assume that $d_1\notin reg(\D_1)$.    
Then, 
$$(d_1\Delta d_2)\Delta a = (d_1,\infty_2)\Delta a = ((d_1,\infty_2),0_A)$$ by Lemma \ref{lemma:tuple of regular},    
and
$$d_1\Delta (d_2\Delta a) = (d_1,(\infty_2,0_A)),$$ because $d_1\notin reg(\D_1)$.
    
\item Assume that $d_1\in reg(\D_1)$ and $d_2\notin reg(\D_2)$.
Then, $$(d_1\Delta d_2)\Delta a = (d_1,d_2)\Delta a = ((d_1,d_2),0_A)$$ by Lemma \ref{lemma:tuple of regular}, and     
$d_1\Delta (d_2\Delta a) = (d_1,d_2\Delta a) = (d_1,(d_2,0_A))$.
    
\item Assume that $d_1\in reg(\D_1)$ and $d_2\in reg(\D_2)$.   
Then, 
$$(d_1\Delta d_2)\Delta a = (d_1,d_2)\Delta a = ((d_1,d_2),a)$$
by Lemma \ref{lemma:tuple of regular}, and    
$d_1\Delta (d_2\Delta a) = (d_1,d_2\Delta a) = (d_1,(d_2,a))$. 
This shows that $(d_1\Delta d_2)\Delta a = d_1\Delta (d_2\Delta a)$.

\end{itemize}

\end{proof}

\Associativity*

\begin{proof}

Notice that the product, zero element, and unity elements of the two structures $(\D_1\Delta\D_2)\Delta \A$ and $\D_1\Delta(\D_2\Delta \A)$ correspond through the identification $(D_1\times D_2)\times A=D_1\times(D_2\times A)$. We only have to verify that the two structures $(\D_1\Delta\D_2)\Delta \A$ and $\D_1\Delta(\D_2\Delta \A)$ have the same sets and the same addition.

\begin{itemize}
    
    \item To prove that $(D_1\Delta D_2)\Delta A = D_1\Delta (D_2\Delta A)$, it is sufficient to prove that $\forall (d_1,d_2,a)\in D_1\times D_2\times A$, $(d_1\Delta d_2)\Delta a = d_1\Delta (d_2\Delta a)$, which is true by Lemma \ref{lem:successive_delta}.
    
    \item We now show that the two structures have the same addition. \\Let $((d_1,d_2,a),(d'_1,d'_2,a')) \in (D_1\times D_2\times A)^2$.
    By Theorem \ref{thm:Delta lex order},

    \begin{eqnarray*}
        (d_1,(d_2,a))&\underset{\D_1\Delta(\D_2\Delta \A)}{\oplus}&(d'_1,(d'_2,a'))\\ &=&\left\{\begin{array}{cc}
         (d_1,(d_2,a)) & \text{ if }d_1<_{\D_1}d'_1 \\
         (d'_1,(d'_2,a')) & \text{ if } d'_1<_{\D_1}d_1 \\
        (d_1,(d_2,a) \underset{\D_2\Delta\A}{\oplus} (d'_2,a')) & \text{ if }d_1=d'_1
    \end{array} \right\}\\
    &=&
    \left\{\begin{array}{cc}
         (d_1,(d_2,a)) & \text{ if }d_1<_{\D_1}d'_1 \\
         (d'_1,(d'_2,a')) & \text{ if } d'_1<_{\D_1}d_1 \\
        (d_1,(d_2,a)) & \text{ if }d_1=d'_1\text{ and }d_2<_{\D_2}d'_2 \\
        (d'_1,(d'_2,a')) & \text{ if }d_1=d'_1\text{ and }d'_2<_{\D_2}d_2 \\
        (d_1,(d_2,a+_Aa')) & \text{ if }d_1=d'_1\text{ and }d_2=d'_2
    \end{array} \right\} \\
    &=&\left\{\begin{array}{cc}
         (d_1,(d_2,a)) & \text{ if }(d_1,d_2)<_{\D_1lex\D_2}(d'_1,d'_2) \\
         (d'_1,(d'_2,a')) & \text{ if }(d'_1,d'_2)<_{\D_1lex\D_2}(d_1,d_2) \\
         (d_1,(d_2,a+_Aa')) & \text{ if }(d_1,d_2)=(d'_1,d'_2)
    \end{array} \right\}.
    \end{eqnarray*}
    We see that the two structures $(\D_1\Delta\D_2)\Delta \A$ and $\D_1\Delta(\D_2\Delta \A)$ have the same addition, under the identification $D_1\times(D_2\times A) = (D_1\times D_2)\times A$.

\end{itemize}

\end{proof}

\DeltaProductMeasure*

\begin{proof}

\begin{itemize}

    \item \textbf{zero axiom:} $(w\Delta\mu)(\emptyset)=(w(\emptyset),0_A) = (\infty_D,0_A)$ because $w$ is a measure and $\infty_D$ is not regular.
    
    \item \textbf{unit axiom:} $(w\Delta\mu)(T^{\emptyset})=(w(T^{\emptyset}),\mu(T^{\emptyset})) = (0_D,1_A)$ because $w$ and $\mu$ are measures, and $0_D$ is regular.
    
    \item \textbf{additivity:} Let $\F_1,\F_2\subseteq T^{S'}$ be disjoint and  $S'\subseteq S$. Since $w$ is a measure, we have 
    \begin{equation}\label{eqn*}
    w(\F_1\uplus \F_2) = \min(w(\F_1),w(\F_2)).
    \end{equation}

 Suppose first that $w(\F_1)=w(\F_2)$. By idempotence of $\min$ and \eqref{eqn*}, $w(\F_1\uplus\F_2)=w(\F_1)=w(\F_2)$, and
$${\argmin}_w(\F_1\uplus\F_2)={\argmin}_w(\F_1)\uplus{\argmin}_w(\F_2).$$
Since $\mu$ is a measure,  
\begin{eqnarray*}
(w\Delta\mu)(\F_1\uplus\F_2)&=&
    (w(\F_1\uplus\F_2)) \Delta (\mu({\argmin}_w(\F_1\uplus\F_2)))\\ %
    &=&
    (\min(w(\F_1),w(\F_2))) \Delta  (\mu({\argmin}_w(\F_1) \uplus {\argmin}_w(\F_2))) 
    \\
    &=& (\min(w(\F_1),w(\F_2))) \Delta  (\mu({\argmin}_w(\F_1)) +_A \mu({\argmin}_w(\F_2)))
    \\
    &=& ((w(\F_1)) \Delta (\mu({\argmin}_w(\F_1)))) \oplus ((w(\F_2)) \Delta (\mu({\argmin}_w(\F_2))))\\
    &=&(w\Delta\mu)(\F_1)\oplus (w\Delta\mu)(\F_2).
\end{eqnarray*}    
    
    Suppose now that $w(\F_1)<_Dw(\F_2)$. By Property \ref{prop:minimum} and \eqref{eqn*}, we have $w(\F_1\uplus \F_2)=w(\F_1)$, and thus ${\argmin}_w(\F_1\uplus\F_2)={\argmin}_w(\F_1)$, and:
 \begin{eqnarray*}
 (w\Delta\mu)(\F_1\uplus\F_2) &=& 
  (w(\F_1\uplus\F_2)) \Delta (\mu({\argmin}_w(\F_1\uplus\F_2)))\\
  &=&(w(\F_1)) \Delta (\mu({\argmin}_w(\F_1))) %
  =(w\Delta\mu)(\F_1) \\
  &=&((w(\F_1))\Delta \mu({\argmin}_w(\F_1)))\oplus ((w(\F_2))\Delta \mu({\argmin}_w(\F_2)))\\
 &=&
  (w\Delta\mu)(\F_1)\oplus(w\Delta\mu)(\F_2).
\end{eqnarray*}
    The proof in the case when $w(\F_2)<_Dw(\F_1)$ follows similarly.

    \item \textbf{elementary muliplicativity:} Let $f_1\in T^{S_1}$ and $f_2\in T^{S_2}$, for  disjoint subsets $S_1$ and $S_2$ of $S$.
By elementary multiplicativity of $w$,  
$$w(\{f_1\Join f_2\})=w(\{f_1\}) +_D w(\{f_2\}).$$
Also, ${\argmin}_w(\{f_1\})=\{f_1\}$ and ${\argmin}_w(\{f_2\})=\{f_2\}$.

Suppose that $w(\{f_1\})$ and $w(\{f_2\})$ are multiplicatively regular.
By Property \ref{prop:mult regular}, it follows that $w(\{f_1\Join f_2\})$ is multiplicatively regular, and
\begin{eqnarray*}
    (w\Delta\mu)(\{f_1\Join f_2\}) &=& (w(\{f_1\Join f_2\}),\mu(\{f_1\Join f_2\}))\\
    &=&  (w(\{f_1\})+_D w(\{f_2\}),\mu(\{f_1\})\times_A \mu(\{f_2\}))\\
    &=& (w(\{f_1\}),\mu(\{f_1\}))\otimes(w(\{f_2\}),\mu(\{f_2\}))\\
    &=& (w\Delta\mu)(\{f_1\}) \otimes (w\Delta\mu)(\{f_2\}),
\end{eqnarray*}    
    since $w(\{f_1\})\in reg(\D)$ and $w(\{f_2\})\in reg(\D)$.

 Suppose that one of $w(\{f_1\})$ and $w(\{f_2\})$ is not multiplicatively regular. Without loss of generality, assume that it is $w(\{f_1\})$.
 Then $$w(\{f_1\Join f_2\})=w(\{f_1\}) +_D w(\{f_2\})$$ is not multiplicatively regular by Property \ref{prop:mult regular}.

There exists $a_2\in A$ such that $(w\Delta\mu)(\{f_2\}) = (w(\{f_2\}),a_2)$. Indeed, we take $a_2:=\mu(\{f_2\})$ if $w(\{f_2\})\in reg(\D)$ and $a_2:=0_A$ if $w(\{f_2\}) \notin reg(\D)$.
 Then 
 \begin{eqnarray*}
(w\Delta\mu)(\{f_1\Join f_2\})&=& (w(\{f_1\Join f_2\}),0_A)\\
&=& (w(\{f_1\}+_D w(\{f_2\}),0_A\times a_2)\\ &=& (w(\{f_1\}),0_A)\otimes (w(\{f_2\}),a_2)\\
&=& (w\Delta\mu)(\{f_1\}) \otimes (w\Delta\mu)(\{f_2\}),
 \end{eqnarray*}
 by definition of $a_2$ and $w\Delta\mu$, since $w(\{f_1\})$ is not regular.

  \end{itemize}

\end{proof}

\end{document}